\documentclass[12pt]{article}
\usepackage[left=1in,top=1.2in,right=1in,bottom=1.2in]{geometry}
\usepackage{geometry}
\pdfoutput=1 
\usepackage{mathtools} 
\usepackage{setspace}
\usepackage{amsthm}
\usepackage{amsmath,amsfonts,amssymb}
\usepackage{natbib}
\usepackage{changepage}
\usepackage{graphicx}
\usepackage{subcaption}
\usepackage{rotating}
\usepackage{pdflscape}
\usepackage[utf8]{inputenc}    
\usepackage[english]{babel}

\usepackage{authblk}
\usepackage{paralist}

\usepackage[usenames]{color}
\usepackage{placeins}
\usepackage{float}
\usepackage{mathabx}
\usepackage{afterpage}
\usepackage[export]{adjustbox}
\usepackage{epstopdf}
\usepackage{url}
\usepackage{booktabs, tabularx}

\usepackage[font=footnotesize %,labelfont=bf
           ]{caption}
\usepackage{multirow}

\usepackage{breqn}

\definecolor{Pink}{rgb}{1.0, 0.5, 0.5}
\definecolor{Maroon}{rgb}{0.8, 0.0, 0.0}

\def\boxit#1{\vbox{\hrule\hbox{\vrule\kern6pt\vbox{\kern6pt#1\kern6pt}\kern6pt\vrule}\hrule}}

\newcommand{\B}{\mathbf}

\newtheorem{theorem}{Theorem}
\newtheorem{lemma}[theorem]{Lemma}
\newtheorem{corollary}[theorem]{Corollary}
\newtheorem{assumption}{Assumption}

\newcommand{\be}{\mbox{\bf e}}

\newcommand{\balpha}{\mbox{\boldmath $\alpha$}}

\newcommand{\bDelta}{\mbox{\boldmath $\Delta$}}
\newcommand{\bbeta}{\mbox{\boldmath $\beta$}}

\newcommand{\bSig}{\mbox{\boldmath $\Sigma$}}

\newcommand{\bPsi}{\mbox{\boldmath $\Psi$}}

\newcommand{\bftab}{\fontseries{b}\selectfont}

\def\beqn{\begin{eqnarray}}
\def\eeqn{\end{eqnarray}}
\def\beqns{\begin{eqnarray*}}
\def\eeqns{\end{eqnarray*}}

\def\0{{\bf 0}}
\def\A{{\bf A}}

\def\b{{\bf b}}

\def\e{{\bf e}}

\def\I{{\bf I}}

\def\S{{\bf S}}
\def\s{{\bf s}}

\def\v{{\bf v}}
\def\W{{\bf W}}
\def\w{{\bf w}}
\def\X{{\bf X}}
\def\x{{\bf x}}

\def\y{{\bf y}}
\def\Z{{\bf Z}}
\def\z{{\bf z}}
\def\1{{\bf 1}}

\def\J{{\bf J}}

\def\trans{^{\rm T}}
\def\strans{^{*\rm T}}
\newcommand{\beq}{\begin{equation}}
\newcommand{\eeq}{\end{equation}}

\newcommand{\bes}{\begin{eqnarray*}}
\newcommand{\ees}{\end{eqnarray*}}
\newcommand{\bi}{\begin{itemize}}
\newcommand{\ei}{\end{itemize}}
\newcommand{\bPhi}{\boldsymbol{\Phi}}

\usepackage{algorithm} % pseudo code
\usepackage{algpseudocode}

\begin{document}

% Title of paper
\title{Log-Contrast Regression with Functional Compositional Predictors: Linking Preterm Infant's Gut Microbiome Trajectories to Neurobehavioral Outcome}
%Linking Preterm Infant's Gut Microbiome Trajectories in Early Postnatal Period to Neurobehavioral Outcome

% List of authors, with corresponding author marked by asterisk
%\author{}
\author{Zhe Sun$^1$, Wanli Xu$^2$, Xiaomei Cong$^2$, Gen Li$^3$, Kun Chen$^1$\thanks{Corresponding author; kun.chen@uconn.edu}\\
$^1$\textit{Department of Statistics, University of Connecticut}\\%[2pt]
$^2$\textit{School of Nursing, University of Connecticut, Storrs, CT}\\
$^3$\textit{Department of Biostatistics, Columbia University}
}

\date{}
\maketitle

\begin{abstract}

			The neonatal intensive care unit (NICU) experience is known to be one of the most crucial factors that drive preterm infant's neurodevelopmental and health outcomes. It is hypothesized that stressful early life experience of very preterm neonate is imprinting gut microbiome by the regulation of the so-called brain-gut axis, and consequently, certain microbiome markers are predictive of later infant neurodevelopment. To investigate, a preterm infant study was conducted; infant fecal samples were collected during the infants' first month of postnatal age, resulting in functional compositional microbiome data, and neurobehavioral outcomes were measured when infants reached 36–38 weeks of post-menstrual age. To identify potential microbiome markers and estimate how the trajectories of gut microbiome compositions during early postnatal stage impact later neurobehavioral outcomes of the preterm infants, we innovate a sparse log-contrast regression with functional compositional predictors. The functional simplex structure is strictly preserved, and the functional compositional predictors are allowed to have sparse, smoothly varying, and accumulating effects on the outcome through time. Through a pragmatic basis expansion step, the problem boils down to a linearly constrained sparse group regression, for which we develop an efficient algorithm and obtain theoretical performance guarantees. Our approach yields insightful results in the preterm infant study. The identified microbiome markers and the estimated time dynamics of their impact on the neurobehavioral outcome shed light on the linkage between stress accumulation in early postnatal stage and neurodevelopmental process of infants.\\


\end{abstract}
\doublespace

%\clearpage
% >>> Introduction <<<
\section{Introduction}

Over the past decade, advances in neonatal care have contributed to a dramatic increase in survival among very preterm birth infants (born before 32 weeks' gestation) from 15\% to over 90\% \citep{Fanaroff2003, Stoll2010}. With this cheerful gain in survival, recent research has shifted focus to the investigation of the increase in neurological morbidity and long-term adverse outcomes related to immature neuro-immune systems and stressful early life experience \citep{Mwaniki2012}. In particular, the neonatal intensive care unit (NICU) experience is found to be one of the most crucial factors that drive preterm infant neurodevelopmental and health outcomes. Accumulated infant stress at NICU arises from numerous causes, such as repeated painful procedures, daily clustered care, maternal separation, among others. \citet{Mwaniki2012} showed that these neonatal insults were associated with a much escalated risk of long-term neurological morbidity, e.g., 39.4\% of NICU survivors had at least one neurodevelopmental deficit. However, the onset of the altered neuro-immune progress induced by infant stress/pain is often insidious, and the mechanism of this association, which holds the key for reducing costly health consequences of prematurity, remain largely unclear. Expanding research evidence supports that a functional communication exists between the central nervous system and gastrointestinal tract, the brain-gut axis, in which the gut microbiome plays a key role in early programming and later responsivity of the stress system \citep{Dinan2012}.

%during critical neurodevelopmental windows
%in the exploration of the mechanisms that alter biological neuro-development process,
%Caporaso2012

As such, a central hypothesis is that the stressful early life experience of very preterm neonates is imprinting gut microbiome by the regulation of the brain-gut axis, and consequently, certain microbiome markers are predictive of later infant neurodevelopment. To investigate, a study was conducted in a NICU in the northeast of the U.S., where stable preterm infants were recruited. Infant fecal samples were collected daily when available, during the infant's first month of postnatal age. 
%Bacterial DNA were isolated and extracted from each stool sample, and through sequencing and processing, resulted in microbiome compositional data \citep{Bomar2011,xiaomei2017}. 
Bacterial DNA were isolated and extracted from each stool sample, and through sequencing and processing, resulted in gut microbiome data.
Gender, delivery type, birth weight, feeding type, among others, were also recorded for each infant. Infant neurobehavioral outcomes were measured when the infant reached 36--38 weeks of post-menstrual age, using the NICU Network Neurobehavioral Scale (NNNS). More details on the study and the data are provided in Section \ref{sec:data}. The above scientific hypothesis can then be approached through a statistical analysis, by examining how the microbiome compositions collected over the early postnatal period predict or impact on the later NNNS score, after adjusting for the effects of relevant infant characteristics.

The gut microbiome data were processed and operationalized as compositions, as commonly done in the microbiome literature \citep{Bomar2011,xiaomei2017}. Compositional data analysis is not an unfamiliar territory to statisticians. Data consisting of  percentages or proportions of certain composition are commonly encountered in various scientific fields including ecology, biology and geology. One unique attribute of compositional data is the unit-sum constraint, i.e., the components of a composition are non-negative and always sum up to one; this entails that the data live in a simplex and thus renders many statistical methods that comply with Euclidean geometry inapplicable. Much foundational work on the statistical treatment of compositional data was done by John Aitchison \citep{Aitchison1982, Aitchison1984}; see \citet{AITCHISON2003} for a thorough survey on the subject. Of particular interest to us is regression with compositional predictors, for which the \textit{log-contrast models} \citep{Aitchison1984} have been very popular. A prominent feature of the model is that it enables the regression analysis to obey the so-called principle of subcompositional coherence, i.e., the compositional data should be analyzed in a way that the same results can be obtained regardless of whether we analyze the entire composition or only a subcomposition \citep{Aitchison2005}. Recently, \citet{Lin2014} studied a sparse linear regression model with compositional covariates, extending the log-contrast model to high dimensions. The problem was nicely formulated as a constrained lasso regression \citep{tib1996}, with a zero-sum linear constraint on the regression coefficients. \citet{Shi2016} further extended the sparse regression model to the case of multiple linear constraints for the analysis of microbiome subcompositions, and a de-biased procedure was adopted to obtain an asymptotically unbiased estimator of the regression coefficients and its asymptotic distribution. See \citet{LiH2015} for a recent comprehensive review on microbiome compositional data analysis. However, to our knowledge, regression method on handling high-dimensional compositional trajectories or series is still lacking.

%for making statistical inference.

%Since the seminal work of Aitchison (1982), methodolog- ical developments for compositional data analysis have given rise to fruitful research, thoroughly surveyed by Aitchison (2003).c The increasing availability of large compositional datasets, whose dimensionality is comparable to or much larger than the sample size, poses new challenges to existing methodology. However, little formal effort has been made to develop principled tools of analysis for such data. A typical example arises in metagenomic studies of microbial com- munities based on 16S rRNA gene sequencing, where the relative abundances of hundreds to thousands of bacterial taxa on a few tens to hundreds of individuals are available for analysis; see, for example, Chen & Li (2013).

%With our application, a useful setup is to let the psychological test score be the outcome variable, and let the compositions of bacteria in stool samples be the predictor. Other variables such as gender, race, feed type, etc would serve as control variables.

%The problem we are facing, however, can not be solved by existing approaches. In our problem, the main complication is that the microbiome compositions were observed over time, producing high-dimensional compositional trajectories or series.

Motivated by the needs in identifying potential microbiome markers and estimating how the trajectories of microbiome compositions along early postnatal stage impact later neurobehavioral outcome, we propose a \textit{sparse log-contrast regression model with functional compositional predictors}. 
%In our approach detailed in Section \ref{sec:method}, the compositional predictors are allowed to have smoothly varying, accumulating effects on the outcome through certain continuous domain, e.g., time. A component of the composition is deemed irrelevant when it does not impact the outcome at all throughout the time window, i.e., its corresponding coefficient curve is a zero line. Sparsity-inducing regularized estimation is thus adopted as it is expected that many compositional components are irrelevant or having negligible effects on the outcome. Through a simple yet effective basis expansion step, the proposed setup reduces to a linearly-constrained sparse group regression. 
In our approach detailed in Section \ref{sec:method}, longitudinal microbial compositions are treated as functional compositional predictors, with time-varying effects on the outcome. We build a scalar-on-function regression model for the log-transformed predictors, which naturally connects to the log-contrast regression. We particularly focus on the identification of important microbes using a sparsity-inducing regularization method.
Section \ref{sec:est} concerns the computational issues. Some theoretical properties of the proposed estimator that are of practical concern are discussed in Section \ref{sec:th}. In Section \ref{sec:sim}, simulation studies showcase the superior performance of the proposed approach over several competing methods. The data analysis of the preterm infant study is presented in Section \ref{sec:app}. The identified microbiome markers are justifiable based on existing literature, and the estimated dynamic trajectories of their impact on the outcome shed new lights on the functional linkage between the accumulation of prenatal stress and neurodevelpoment of infants. Some concluding remarks are given in Section \ref{sec:dis}.

\section{Preterm Infant Study and Problem Setup}\label{sec:data}

%\subsection{Data Description}
%Caporaso2012
Data were collected at a Level IV NICU in the northeast region of the U.S. (Level IV NICUs provide the highest level, the most acute care.) Fecal samples of preterm infants were collected daily when available, mainly during the infant’s postnatal age (PNA) of 5 to 28 days ($t\in [5,28]$). Bacterial DNA were isolated and extracted from each stool sample \citep{Bomar2011,xiaomei2017}; the V4 regions of the 16S rRNA gene were sequenced using the Illumina platform and clustered and analyzed using QIIME \citep{xiaomei2017}, resulting in microbiome count data.
Since the number of sequencing reads varied a lot across samples, we further normalize the data by calculating the ratio of each microbe in each sample. As a result, we obtain a compositional data matrix. To conduct log transformation in our model, following the convention in the literature, we replace zeros by the maximum rounding error (i.e., 0.5) to avoid singularity \citep{AITCHISON2003, Lin2014}.
% As the number of sequencing reads varied a lot across samples, the count data were transformed into compositional data with zero count replaced by 0.5, the maximum rounding error \citep{Aitchison1986, Lin2014}. 
%The compositional data consisted of $p=22$ categories at the order level of the taxonomic ranks. 
Due to the limited sample size, we mainly focus on $p=22$ categories at the order level of the taxonomic ranks as a proof of concept. (We also perform a confirmative analysis at the genus level which has more than 60 categories.) Taxonomic rank is the relative level of a group of organisms in a taxonomic hierarchy in biological classification; the major ranks are species, genus, family, order, class, phylum, kingdom, and domain. In this study, infants with less than 5 fecal samples were excluded, which resulted in $n=34$ infants. There were totally 414 fecal samples, so the average number of daily fecal samples collected for each infant was 12.2. Figure \ref{fig:data}(a) shows the histogram of the number of samples collected from each infant, and Figure \ref{fig:data}(b)--(d) show some examples of the observed profile of the time-varying compositions along the postnatal age.

% \begin{figure}[htp]
%   \captionsetup[subfigure]{singlelinecheck=false}
%   \centering
%   %\begin{subfigure}%{0.45\textwidth}
%       %\vspace{0.05\textheight}
%       \includegraphics[width=0.4\linewidth%keepaspectratio
%       ]{graphs/hist_obs}%\caption{}   
%       %\vspace{0.03\textheight}
%   %\end{subfigure}
%   % %\hspace{0.2\textwidth}
%   % \begin{subfigure}{0.45\textwidth}
%   %   \includegraphics[width=1\linewidth,%keepaspectratio
%   %   ]{graphs/sub10_3}\caption{}
%   % \end{subfigure}
%   \caption{Histogram of the number of samples collected from each infant.} \label{fig:hist} 
% \end{figure}

\begin{figure}[h!]
    \captionsetup[subfigure]{singlelinecheck=false}
    \centering
    \begin{subfigure}{0.42\textwidth}
    %   %\vspace{0.05\textheight}
        \includegraphics[width=0.9\linewidth%keepaspectratio
        ]{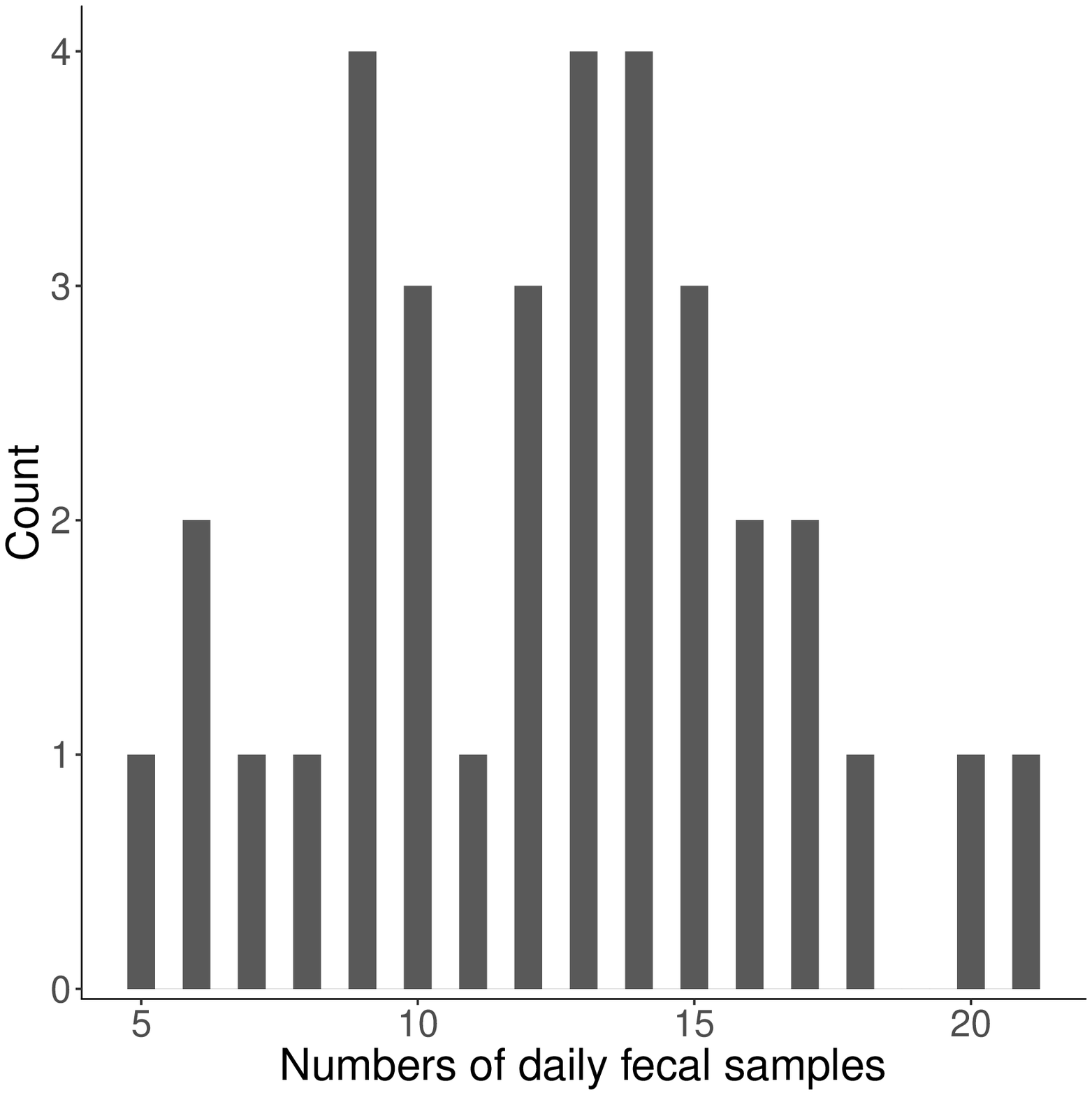}\caption{} 
    %   %\vspace{0.03\textheight}
    \end{subfigure}
    \begin{subfigure}{0.42\textwidth}
        %\caption{}
        %\vspace{0.05\textheight}
        \includegraphics[width=1\linewidth%keepaspectratio
        ]{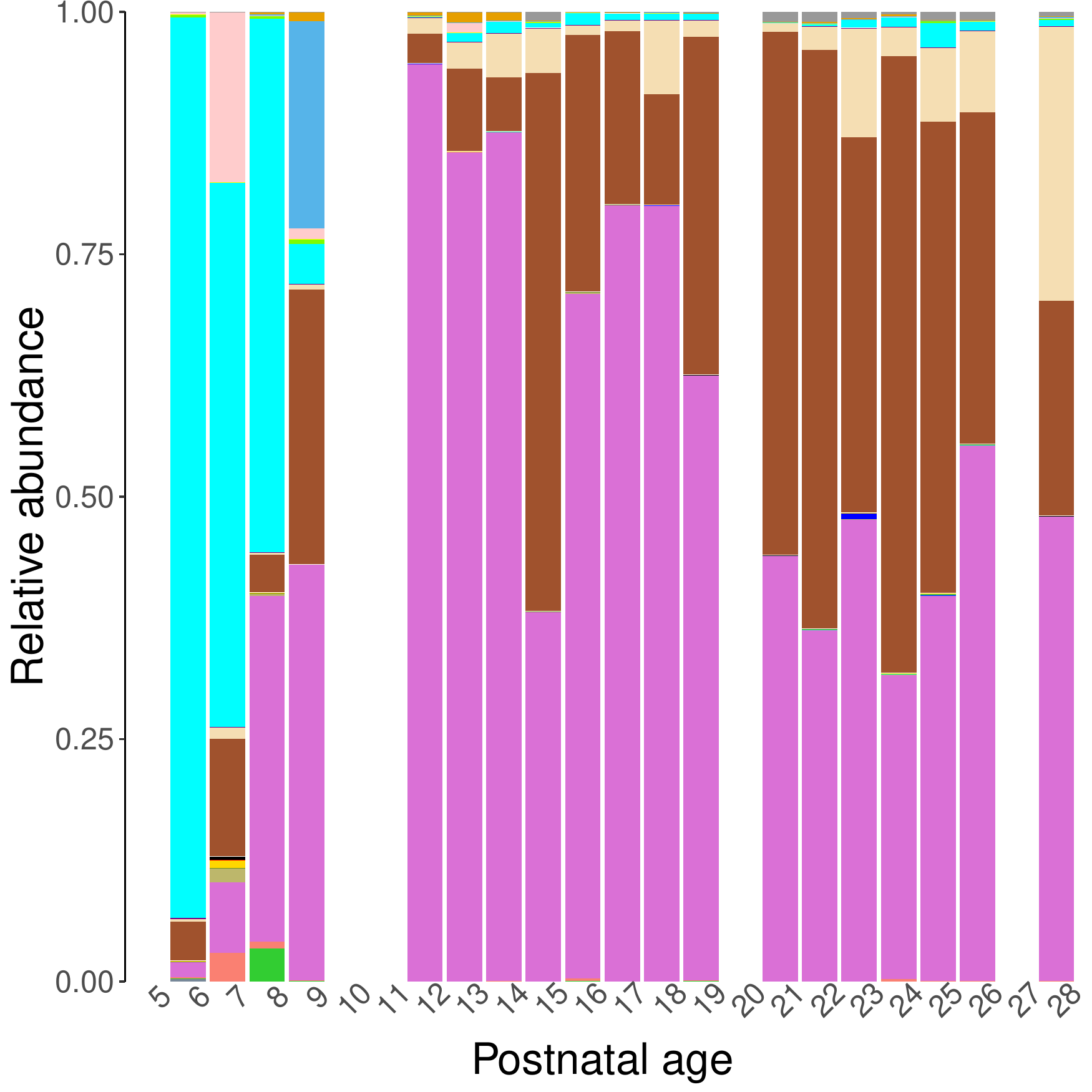}\caption{} %9,28
        %\vspace{0.03\textheight}
    \end{subfigure}\\
    %\hspace{0.2\textwidth}
    \begin{subfigure}{0.42\textwidth}
        %\caption{}
        \includegraphics[width=1\linewidth%keepaspectratio
        ]{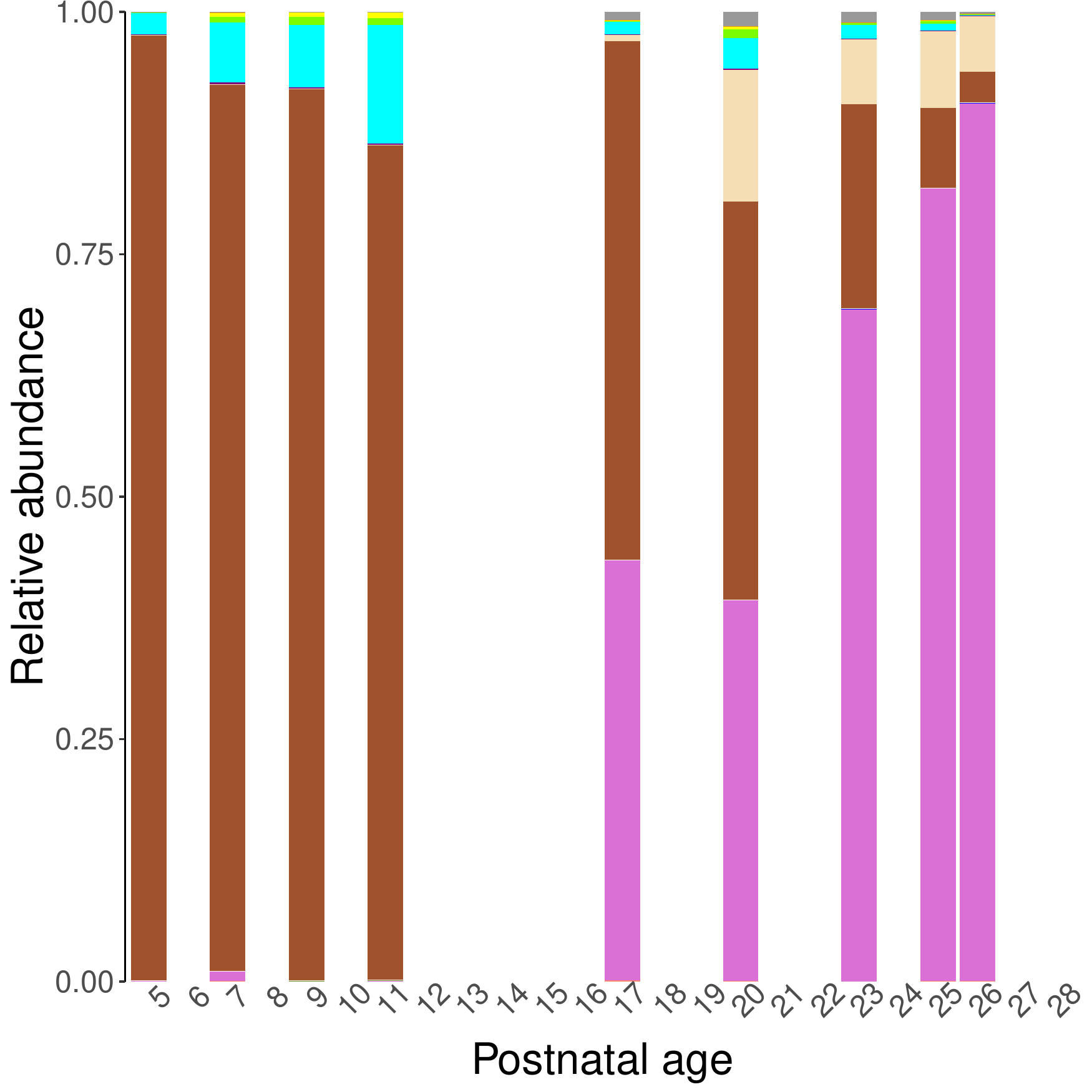}\caption{}%21
              \end{subfigure}
        \begin{subfigure}{0.42\textwidth}
        %\caption{}
        \includegraphics[width=1\linewidth%keepaspectratio
        ]{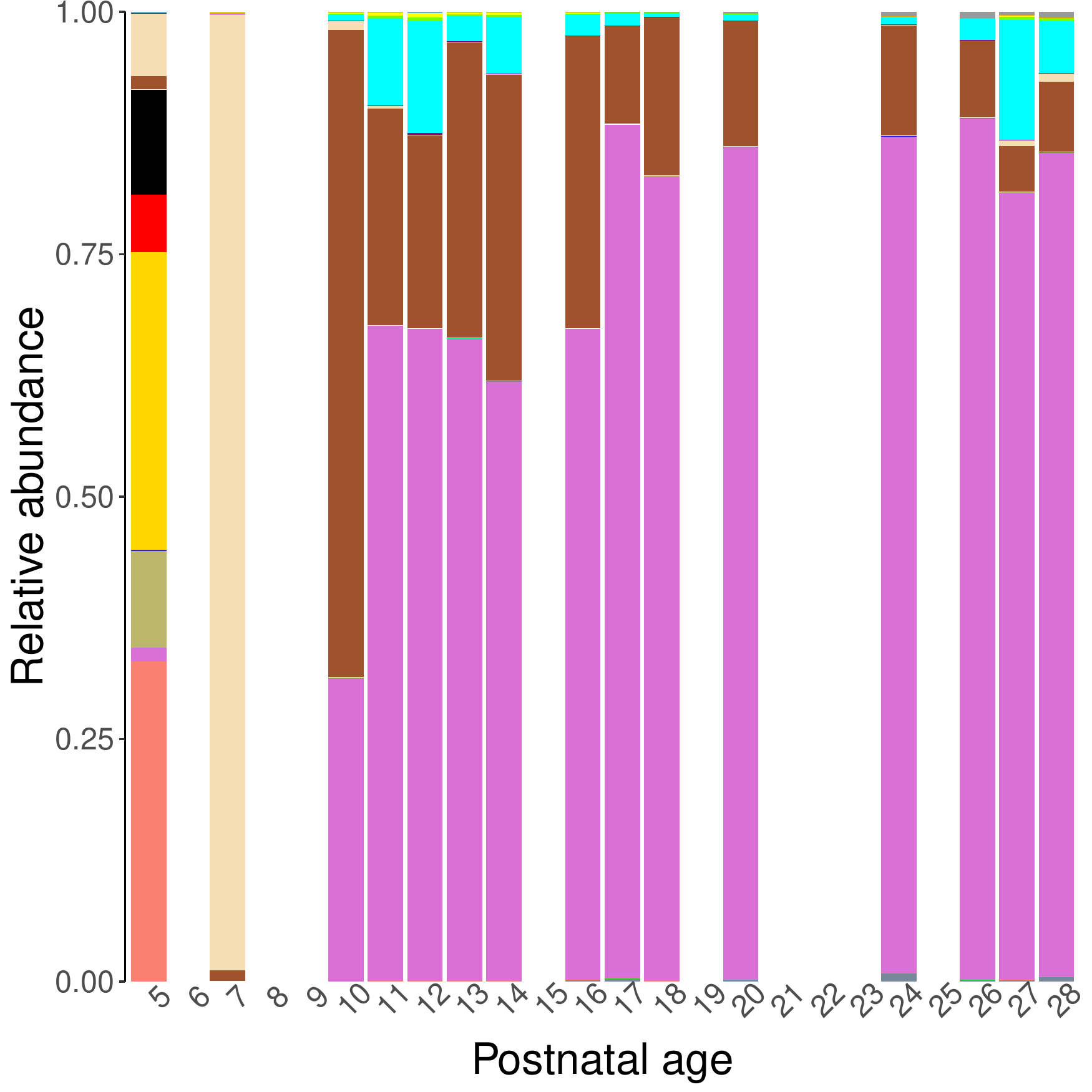}\caption{}%21
    \end{subfigure}
    \begin{minipage}[t]{1\textwidth}
        \includegraphics[width=\linewidth]{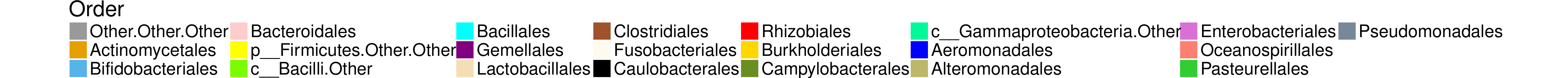}
    \end{minipage}
    \caption{(a) Histogram of the number of samples collected from each infant. (b)--(d) Example profiles of time-varying compositional data along postnatal age.} \label{fig:data} 
\end{figure}

%we focus on PNA ranging between 5 and 30 and

Infant neurobehavioral outcomes were measured when the infant reached 36--38 weeks of post-menstrual age or prior to hospital discharge, using the NICU Network Neurobehavioral Scale (NNNS). The NNNS is a standardized assessment of neonatal neurobehavioral outcomes that provides an appraisal of neurological integrity and behavioral function of the normal and at-risk/preterm infant. In particular, the Stress/Abstinence subscale (NSTRESS) measures signs of stress and includes 50 items. Each sign of stress/abstinence is scored as present or absent, and the composite NSTRESS score ranges between 0 and 1. A higher NSTRESS score demonstrates a more stressful behavioral performance. \citet{xiaomei2017} showed that the composite NSTRESS score is positively associated with painful/stressful experience in preterm infants. Other variables about birth and characteristics of infant included gender, delivery type, premature rupture of membranes (PROM), score for Neonatal Acute Physiology--Perinatal Extension-II (SNAPPE-II), birth weight, and percentage of feeding with mother's breast milk (\%MBM).

%measured by New Injury Severity Score (NISS). 

To formulate the statistical problem, let $\y = [y_1,\ldots, y_n]\trans \in \mathbb{R}^n$ be consisting of the observed neurobehavioral outcomes of the preterm infants, i.e., their NNNS scores. Let $\x_i(t)= [x_{i1}(t),\ldots,x_{ip}(t)]\trans \in \mathbb{S}^{p-1}$ be the gut microbiome compositions from the $i$th infant at time $t$. Here we let $\mathbb{S}^{p-1} = \{[x_1,\ldots,x_p]\trans\in\mathbb{R}^p; x_{j}>0, \sum_{j=1}^{p}x_j=1.\}$, to denote the $(p-1)$-dimensional positive simplex lying in $\mathbb{R}^p$. Let $\X(t) = [\x_1(t),\ldots, \x_n(t)]\trans\in \mathbb{R}^{n\times p}$ be the matrix of the functional predictors at time $t$. The observed gut microbiome compositions during the early postnatal period can then be viewed as discrete observations from $\X(t)$. Also define  $\Z_c\in \mathbb{R}^{n\times p_c}$, formed by data from the aforementioned time-invariant infant characteristics, e.g., gender, delivery type, among others. 

As the main objective is to identify the microbiome markers that are predictive of later infant neurodevelopment, we need to perform a regression analysis to examining how the outcome $y$, the NNNS score, is associated with $\X(t)$, the gut micorbiome trajectories, while controlling for the infant characteristics collected in $\Z_c$. The fact that $\X(t)$ is both functional and compositional makes the problem very challenging.

% Most of the infants included is female (64.7\%), least in delivered by vaginal section (55.9\%), and fair in PROM (50\%);
% Average birth weight was 1451.7 g (SD = 479.3); average SNAPEII score was 9.3 (SD = 10.6);
% average mean percentage of MBM, which is calculated for the first 30 days of life, is 58.6\% (SD = 27.8\%).
% On average, days of fecal samples collected for each infants was 12.4 (SD = 4.0). 57.7\% of compositions from 423 fecal samples with 0 entries, with largest count 9999.

%Besides gut microbiota, there might be other important demographic variables functioning as control variables.

%There were totally 56 preterm infants, one of who did not have neither NSTRESS score or SNAPEII. 38 out of 55 infants had at least one fecal sample collected during their first 40 postnatal age (PNA) day, with total 455 fecal samples. One fecal sample was collected before PNA=4; 20 belonging to 4 infants were collected after PNA=30. Then we focus on PNA range from 5 to 30. Since Our method need integration of longitudinal compositional data with basis, which is attained by summation of step functions, we excluded subjects with less than 5 fecal samples. Finally, our data set includes 34 subject with 423 fecal samples.

% <<< Data >>>

% >>> Model <<<
\section{Regression with Functional Compositional Predictors}\label{sec:method}

\subsection{Linear Log-Contrast Model}\label{sec:method:llc}

%Log-ratio transformation technique is widespreadly used when deal with compositional data. For instance, ilr and clr.

%These transformations covert data lying in a hyperplane $\mathbb{S}^P$ into the Euclidean space $\mathbb{R}^{p-1}$. We're more familiar with the latter space and benefits us with plenty of well-established models.

We first briefly review the existing regression approaches for dealing with a single set of compositional predictors. Suppose we observed $n$ independent observations of a response variable $y_i\in \mathbb{R}$ and a compositional predictor $\x_i= [x_{i1},\ldots,x_{ip}]\trans$ such that $\x_i \in \mathbb{S}^{p-1}$. Denote $\y = [y_1,\ldots, y_n]\trans \in \mathbb{R}^n$ as the response vector and
$\X = [\x_1,\ldots, \x_n]\trans \in \mathbb{R}^{n\times p}$%$\X = [\x_1\trans,\ldots, \x_n\trans]\trans \in \mathbb{R}^{n\times p}$
as the design matrix.

%It is apparent that 
Ignoring the simplex structure of $\X$ would lead to parameter identifiablity issue in the linear regression of $\y$ on $\X$. One naive ``remedy'' is to exclude an arbitrary component of the compositional vector in the regression, which, however, leads to a method that is not invariant to the choice of the removed component since it affects both of prediction and selection and consequently makes proper model interpretation and inference difficult. Ever since the pioneer work by John Aitchison \citep{Aitchison1982,Aitchison1984,AITCHISON2003} on the statistical treatments of compositional data, the so-called \textit{log-contrast model} has gained much popularity in a variety of regression problems with compositional predictors. The main idea is to perform a log-ratio transformation of the compositional data, such that the transformed data admit the familiar Euclidean geometry in $\mathbb{R}^{p-1}$. Specifically, for each $i=1,\ldots, n$, let $\widetilde{z}_{ij} = \log(x_{ij}/x_{ir})$, where $r \in \{1,\ldots, p\}$ is a chosen reference level, and $j = 1,\ldots,r-1,r+1,\ldots, p$, resulting in $\widetilde{\Z}_{\bar{r}} = [\widetilde{z}_{ij}]\in \mathbb{R}^{n\times (p-1)}$. Also define $z_{ij} = \log(x_{ij})$ and $\Z = [z_{ij}] \in \mathbb{R}^{n\times p}$. The linear log-contrast regression model is expressed as
\begin{align}\label{eq:model1}
\y =\beta_0^*\1_n + \widetilde{\Z}_{\bar{r}}\bbeta_{\bar{r}}^*  + \be,
\end{align}
where $\beta_0^*$ is the intercept, $\bbeta_{\bar{r}}^*\in \mathbb{R}^{p-1}$ is the regression coefficient vector,
%$\bbeta_{\bar{r}}^*\in \mathbb{R}^{p-1}$ is a subvector of a regression coefficient vector $\bbeta^*\in\mathbb{R}^p$ by removing its $r$th component $\beta_r^*$, 
and $\be\in \mathbb{R}^n$ is the random error vector with zero mean. Interestingly, although it appears that the model in \eqref{eq:model1} depends on the choice of the reference level, it in fact admits a symmetric form. 
%To see this, let $z_{ij} = \log(x_{ij})$ and $\Z = [z_{ij}] \in \mathbb{R}^{n\times p}$. 
By simple algebra, model \eqref{eq:model1} can be equivalently expressed as
\begin{align}\label{eq:model2}
\y =\beta_0^*\1_n + \Z\bbeta^*  + \be, \qquad \mbox{s.t.} \sum_{j=1}^{p}\beta_j^*=0,
\end{align}
where $\bbeta^*$ is the regression coefficient vector for design matrix $\Z$, and $\e$ and $\beta_0^*$ are the same as in model \eqref{eq:model1}. It can be showed that $\bbeta_{\bar{r}}^*\in \mathbb{R}^{p-1}$ is a subvector of a regression coefficient vector $\bbeta^*\in\mathbb{R}^p$ by removing its $r$th component $\beta_r^*$.

Consequently, in classical regression setups, the least squares estimation under model \eqref{eq:model1} is equivalent to the constrained least squares estimation under model \eqref{eq:model2}. However, in high dimensional scenarios, i.e., when $p$ is much larger than $n$, the two model formulations could lead to discrepancies in regularized estimation. For example, the two corresponding lasso criteria \citep{tib1996} are no longer equivalent:
\begin{align}
\min_{\beta_0,\bbeta_{\bar{r}} } & \left\{\frac{1}{2n}\|\y -\beta_0\1_n - \widetilde{\Z}_{\bar{r}}\bbeta_{\bar{r}}\|^2  + \lambda \|\bbeta_{\bar{r}}\|_1\right\},\label{eq:model1lasso}\\
\min_{\beta_0,\bbeta}& \left\{\frac{1}{2n}\|\y -\beta_0\1_n - \Z\bbeta\|^2  + \lambda \|\bbeta\|_1\right\}, \qquad \mbox{s.t.} \sum_{j=1}^{p}\beta_j=0,\label{eq:model2lasso}
\end{align}
where $\|\cdot\|$, $\|\cdot\|_1$ denote the $\ell_2$, $\ell_1$ norms, respectively, and $\lambda$ is a tuning parameter controlling the amount of regularization. Although \eqref{eq:model1lasso} is simpler to compute, clearly its solution and hence its variable selection depend on the choice of the reference component. In contrast, \eqref{eq:model2lasso} remains to be symmetric in all the $p$ components. \citet{Lin2014} proposed and studied \eqref{eq:model2lasso} and showed that the estimator admits many desirable properties \citep{AITCHISON2003}. %\textcolor{red}{discussed by \citet{AITCHISON2003}.} %\citep{AITCHISON2003}. %including scale invariance, permutation invariance and selection invariance.

% By introducing a new coefficient $\beta^*_p = -\sum_{1}^{p-1} \beta^*_i$, the above model can be turned into more conventional form.
% \begin{equation}\label{eq:logcons}
% \y = \mathbf{1}_n\beta_0^* + \Z \bm{\beta} + \bm{\epsilon}, \quad \sum_{1}^{p} \beta_i =0,
% \end{equation}
% where $\Z=\{\log x_{ij}\}_{ij} \in \mathbb{R}^{n \times p}$ is the design matrix and $\bm{\beta} \in \mathbb{R}^p$ is the coefficient vector. Since $\forall i, i=1,\ldots,n$

% \begin{equation*}
% \begin{aligned}
% \sum_{j=1}^{p-1} \log (x_{ij}/x_{ip})\beta^*_j
% &= \sum_{j=1}^{p-1}\log{x_{ij}}\beta^*_j - \sum_{1}^{p-1} \log(x_{ip})\beta^*_j \\
% &= \sum_{j=1}^{p-1}\log{x_{ij}}\beta^*_j + \log(x_{ip}) (-\sum_{1}^{p-1} \beta^*_j) \\
% &= \sum_{j=1}^{p} \Z_{ij}\beta_j.
% \end{aligned}
% \end{equation*}
% Then centralize $\y$ and $\Z$ w.r.t each column, we can get rid of the intercept $\beta_0^*$, i.e.,
% \begin{equation}
% \y = \Z \bm{\beta} + \bm{\epsilon}, \quad \mathbf{1}_p \trans \bm{\beta} =0,
% \end{equation}

\subsection{Sparse Functional Log-Contrast Regression}\label{Covariate Model}

%Still let $\y = [y_1,\ldots, y_n]\trans \in \mathbb{R}^n$ be the observed outcome/response vector, which is time invariant. For any $t \in \mathbb{T}$, let $\x_i(t) \in \mathbb{S}^{p-1}$ be the compositional vector for the $i$th subject; let $\X(t) = [\x_1(t),\ldots, \x_n(t)]\trans\in \mathbb{R}^{n\times p}$ %$\X(t) = [\x_1\trans(t),\ldots, \x_n\trans(t)]\in \mathbb{R}^{n\times p}$ be the matrix of the functional predictors at $t$. 

In the preterm infant study, the compositional predictors are observed over a continuous domain, i.e., time, and thus they should be treated as functional compositional data. Recall from Section \ref{sec:data} that $\y \in \mathbb{R}^n$ is the response vector, $\X(t) \in \mathbb{R}^{n\times p}$ the matrix of the functional and compositional predictors at $t$, and $\Z_c\in \mathbb{R}^{n\times p_c}$ the matrix of time-invariant control variables. Here to focus on the main idea, we assume $\X(t)$ is completely observed for $t\in \mathbb{T}$, and the discussion about handling discrete time data is deferred to Section \ref{sec:est:discrete}. Similar as in Section \ref{sec:method:llc}, we define $\widetilde{\Z}_{\bar{r}}(t) \in \mathbb{R}^{n\times (p-1)}$, for $r=1,\ldots, p$, and $\Z(t) = \log(\X(t)) \in \mathbb{R}^{n\times p}$. %In the sequel we shall use general terminologies to present the proposed approach, as it is widely applicable.

%Some time-invariant control variables may also be available, e.g., gender, delivery type, among others, which form $\Z_c\in \mathbb{R}^{n\times p_c}$.

%Assuming the availability of $\X(t)$, $t\in \mathbb{T}$ and
Motivated by model \eqref{eq:model2}, we propose a \textit{functional log-contrast regression model},
\begin{equation}
\y = \beta_0^*\1_n+ \Z_c\bbeta_c^* + \int_{t\in \mathbb{T}} \Z(t)\bbeta^*(t) dt + \be,\qquad  \mbox{s.t. } \mathbf{1}_p\trans\bbeta^*(t) = 0,\,\forall t\in \mathbb{T}, \label{eq:fmodel}
\end{equation}
where $\beta_0^*$ is the intercept, $\bbeta_c^* \in \mathbb{R}^{p_c}$ is the regression coefficient vector corresponding to the control variables, $\bbeta^*(t)=[\beta_1^*(t),\ldots,\beta_p^*(t)] \trans \in \mathbb{R}^p$ is the functional regression coefficient vector as a function of $t$, and the remaining terms are defined the same as in model \eqref{eq:model2}. The proposed model allows the compositional predictors to have potentially different effects on the response through $\bbeta^*(t)$, and their aggregated effects on the response is then given by the integral of $\Z(t)$ weighted by $\bbeta^*(t)$ over time. Following \citet{Lin2014}, here we adopt the symmetric form of the log-contrast model, in which the zero-sum constraints preserve the simplex structure over time while all the compositional components are treated equally. %We thus term \eqref{eq:fmodel} as the \textit{functional log-contrast regression model}.

%The merit of the above model lies in imposing some meaningful low-dimensional structures on the coefficient curves $\bbeta^*(t)$. 

To address the problems in the preterm infant study, we consider both sparsity and smoothness of $\bbeta^*(t)$. First, as it is believed that only a few compositional components are relevant to the prediction of the outcome, we assume the true coefficient curves are sparse, i.e., $s^* =|\mathcal{S}| \ll p$, where $\mathcal{S}$ is the index set of the non-zero coefficient curves
\begin{align*}
\mathcal{S} = \{j; \beta_j^*(t)\neq 0 \mbox{ for some } t\in\mathbb{T}, j = 1,\ldots,p.\}.
\end{align*}
This sparsity assumption is the basis of component selection and is widely applicable, especially when $p$, the number of compositional components, is large. Second, since the effects of gut microbiome compositions on preterm infant's neurodevelopment evolves gradually over the postnatal period, we assume the coefficient curves are smooth over $t$, and adopt a truncated basis expansion approach \citep{Ramsay2005} to bring the infinite dimensional problem to finite dimensions. Specifically, we assume
    \begin{align}
    \bbeta^*(t)=\B{B}^*\bPhi(t),\label{eq:basis}
    \end{align}
    where $\B{B}^*=[\bbeta_1^*,\ldots,\bbeta_p^*]\trans \in \mathbb{R}^{p \times k}$ is a coefficient matrix, and $\bPhi(t)= [\phi _1(t),\ldots, \phi_k(t)]\trans \in \mathbb{R}^k$ consists of basis with $\J_{\phi\phi}=\int_{t\in \mathbb{T}} \bPhi(t)\bPhi\trans(t)dt$ being a positive definite (p.d.) matrix. Here for simplicity the same set of basis functions is used in the expansion of each $\beta_j(t)$, $j=1,\ldots,p$, which usually suffices in practice, and the extension to use different basis for different $\beta_j(t)$ is straightforward. There are many choices of the basis functions, e.g., Fourier basis, wavelet basis, and spline basis; see \citet{Ramsay2005} for a detailed account on the truncated basis expansion approaches in functional regression. 

Some discussions on the number of basis functions are in order. In classical least squares types of estimation, the choice of $k$ usually boils down to a bias-and-variance tradeoff. That is, while larger values of $k$ can lead to a better in-sample estimation at the risk of potential overfitting, smaller values of $k$ result in simpler estimators at the expense of missing interesting local oscillations. The issue can be resolved by echoing regularization, i.e., taking a sufficiently large $k$ to ensure the flexibility of the model and performing regularized estimation to avoid overfitting. From a theoretical perspective, we allow $k$ to grow with the sample size $n$, that is, the complexity of the functional curves that the method can potentially capture may increase when more data become available; see Section \ref{sec:th} for details. We also remark that for a non-parametric treatment, one can assume $\bbeta^*(t)$ satisfies certain H{\"o}lder condition \citep{Tsybakov2008} to control the approximate error induced by the basis truncation.  

%Henceforth we treat $k$ as a fixed and known quantity in the derivation of the proposed methodology.

%The proposed model in \eqref{eq:fmodel} is simplified under Assumptions \ref{as:0}--\ref{as:1}. 
%With \eqref{eq:basis}, t
The functional sparsity in $\bbeta^*(
t)$ now amounts to the row-sparsity of the coefficient matrix $\B{B}^*$ in \eqref{eq:basis}. The zero-sum constraint on $\bbeta^*(t)$, i.e., $\1_p\trans\bbeta^*(t)= 0$ for all $t \in \mathbb{T}$, is now equivalent to $\B{B}\strans \1_p = \0$. To see this, note that $\mathbf{1}_p\trans \bbeta^*(t) =0$ leads to $\int_{t\in \mathbb{T}} \1_p \trans \B{B}^* \bPhi(t) \bPhi(t)\trans (\1_p\trans \B{B}^*)\trans dt= \1_p\trans \B{B}^* \J_{\phi\phi} (\1_p\trans \B{B}^*)\trans=0$; it follows that $\B{B}\strans \1_p = \0$ as $\J_{\phi\phi}$ is p.d.. (The other direction holds trivially.) Further, the integral part in the model becomes
\begin{align*} %\label{eq:Zbeta} %\label{eq:Z}
\int_{t\in \mathbb{T}} \Z(t)\bbeta^*(t) dt
 &= \int_{t\in \mathbb{T}} \Z(t)\B{B}^*\bPhi(t) dt \\ %\nonumber \\
%& = \int_{t\in \mathbb{T}} \Z(t)\B{B}^*\bPhi(t) dt \nonumber \\
%& =\int_{t\in \mathbb{T}} \Z(t)(\I_p \otimes \bPhi(t)\trans) \mbox{vec}(\B{B}^*\trans) dt \nonumber \\
 &=\left\{\int_{t\in \mathbb{T}} \Z(t)(\I_p \otimes \bPhi(t)\trans) dt \right\} \mbox{vec}(\B{B}\strans) %\nonumber\\
 = \Z\bbeta^*,
\end{align*}
where, with some abuse of notations, we redefine $\bbeta^*   = [\bbeta_1\strans,\ldots, \bbeta_p\strans]\trans = \mbox{vec}(\B{B}\strans) \in \mathbb{R}^{pk}$ and
\begin{align}
\Z & = \int_{t\in \mathbb{T}} \Z(t)(\I_p \otimes \bPhi(t)\trans) dt = [\Z_1, \ldots, \Z_p] \in \mathbb{R}^{n\times (pk)}.\label{eq:Z}
\end{align}
Each $\bbeta_j^* \in \mathbb{R}^k$ and $\Z_j\in \mathbb{R}^{n\times k}$ correspond to the coefficient vector and the covariate matrix for the $j$th compositional component, respectively. We remark that $\Z$ is usually not exactly computed since $\Z(t)$ may not be fully observed; we defer the discussion to Section \ref{sec:est:discrete}. 

The functional model in \eqref{eq:fmodel} then becomes a constrained sparse linear regression model
\begin{equation}
\y = \beta_0^*\mathbf{1}_n + \Z_c\bbeta_c^* + \Z\bbeta^* + \be, \qquad \mbox{s.t. } \sum_{j=1}^{p}\bbeta_j^* = \0,\label{eq:fmodel2}
\end{equation}
where $\bbeta^*$ is expected to be sparse accordingly to the row-sparsity of $\B{B}^*$. To enable the selection of the compositional components, we therefore propose to conduct model estimation by minimizing a linearly constrained group lasso criterion \citep{yuan2006},
\begin{align}\label{eq:glasso}
\min_{\beta_0,\bbeta_c,\bbeta}\left\{\frac{1}{2n}\| \y -\beta_0\1_n - \Z_c\bbeta_c - \Z\bbeta\|^2  + \lambda \sum_{j=1}^{p} \|\bbeta_j\|\right\}, \qquad \mbox{s.t.} \sum_{j=1}^{p}\bbeta_j=\0,
\end{align}
where $\lambda$ is a tuning parameter controlling the amount of regularization. We remark that the group lasso penalty is imposed on the coefficients for each microbiome category to encourage microbe selection.

The proposed estimator possesses several desirable invariance properties \citep
{AITCHISON2003, Lin2014}:\\
%\begin{itemize}

\noindent (I) Scale invariance: the estimator is invariant to the transformation $\X(t)\rightarrow \S\X(t)$ where $\S = \mbox{diag}(\s)$ is a diagonal matrix with diagonal elements $\s = [s_1,\ldots, s_n]\trans$ and all $s_i > 0$. That is, it does not matter whether the data vectors are scaled to have a unit sum; the method only cares about the relative proportions.
        This is simply because $\Z(t)\bbeta(t) = \{\log(\X(t)) + \log(\s)\1_p\trans \}\bbeta(t) = \log(\X(t))\bbeta(t)$, due to the zero-sum constraints. In fact, this scale invariance continues to hold when the scaling factor $\s$ changes in time.\\

\noindent (II) Permutation invariance: results of the analysis do not depend on the sequence by which the components are given or labeled.\\

\noindent (III) Subcomposition coherence: if we know in advance that some $\beta_j(t)$ curves are zero, the analysis is unchanged if we apply the procedure to the subcompositions formed by the components of $\X(t)$ corresponding to the other $\beta_j(t)$ curves. To see this, suppose $\beta_j(t)\equiv 0$ for $j\in \mathcal{S}^c$, where $\mathcal{S}^c$ is the complement of a set $\mathcal{S}$ on $\{1,\ldots,p\}$. Let $\s(t) = \{\X_{\mathcal{S}}(t)\1_{|\mathcal{S}|} \}^{-1} \in \mathbb{R}^n$ be a scaling factor in which the inversion is entrywisely applied, so that $\mbox{diag}(\s(t))\X_{\mathcal{S}}(t)$ gives the subcompositions formed by the components in $\mathcal{S}$. Then we have
\begin{align*}
\log(\X(t))\bbeta(t) %= & \log(\X_{\mathcal{S}}(t))\bbeta_{\mathcal{S}}(t)\\
= & \{\log(\X_{\mathcal{S}}(t)) + \log(\s(t)) \1_{|\mathcal{S}|} \trans\}\bbeta_{\mathcal{S}}(t)\\
%= & \{\log(\X_{\mathcal{S}}(t)) + \log(s(t))\1\trans\}\bbeta_{\mathcal{S}}(t)\\
= & \log(\mbox{diag}(\s(t))\X_{\mathcal{S}}(t))\bbeta_{\mathcal{S}}(t).
\end{align*}
%Similarly, if we fit the proposed model to a $p$-part compositional dataset, the results do change if we include a new non-informative component and work with the resulting $p+1$-part composition.
%\end{itemize}
In particular, when there are only two non-zero components, e.g., $\bbeta_1(t)\neq 0$, $\bbeta_2(t)\neq 0$ and $\bbeta_j(t) = 0$ for $j=3,\ldots,p$, it is necessarily true that $\bbeta_1(t) = - \bbeta_2(t)$ due to the zero-sum constraint. This is neither an unpleasant artifact nor a limitation of the proposed method. This special case can be understood from the above property of subcomposition coherence: the analysis becomes the same as using the subcompositions formed from the first two components of $\X(t)$; consequently, the two possible log-ratios are exactly opposite to each other, so do their corresponding coefficient curves. Therefore, this feature is consistent with the data structure, as in two-part componsitional data, either part carries exactly the same information.

\section{Computation}\label{sec:est}

\subsection{Solving and Tuning Constrained Group Lasso}

The problem in \eqref{eq:glasso} is convex, and we solve it by an augmented Lagrangian algorithm \citep{Boyd2011}. To save space, details are provided in Section \ref{sec:supp:comp} of Supplementary Materials. 

A general way to select the tuning parameters, i.e., the basis dimension $k$ and the group penalty level $\lambda$, is the $K$-fold cross validation \citep{stone1974}, which is based on the predictive performance of the models. However, it is well known that the best model for prediction may not coincide with that for variable selection, and in fact, the former often leads to overselection. This phenomenon under our model is revealed in Section \ref{sec:th}, where it is shown that consistent component selection shall be based on the zero pattern of a thresholded estimator. Following \citet{FanTang2013} and \citet{Lin2014}, we thus also experiment with minimizing a generalized information criterion (GIC) for model selection which favors more sparse models,
$$
\textrm{GIC}(\lambda, k) = \log \big( \widehat{\sigma}^2(\lambda, k) \big) + \big(s(\lambda, k) - 1 \big) k
\log \big( \max\{pk +1 + p_c, n\}\big)
\frac{\log(\log n)}{n},
$$
where $\widehat{\sigma}^2(\lambda, k)$ is the mean squared error define as $\|\y - \widehat{\beta}_0(\lambda, k)\1_n - \Z_c\widehat{\bbeta}_c(\lambda,k)- \Z\widehat{\bbeta}(\lambda, k)  \|^2/n$ with $\widehat{\beta}_0(\lambda,k)$, $\widehat{\bbeta}_c(\lambda,k)$ and $\widehat{\bbeta}(\lambda, k)$ being the regularized estimators of regression coefficients, and $s(\lambda, k)$ is the number of nonzero coefficient groups in $\widehat{\bbeta}(\lambda, k)$.
% \textcolor{red}{
% where $\widehat{\sigma}^2(\lambda, k)=\|\y - (\1_n, \Z_c) \widehat{\bbeta}_{0c}(\lambda,k)- \Z\widehat{\bbeta}(\lambda, k) \|^2/n$,
% $\widehat{\bbeta}_{0c}(\lambda,k)$ and $\widehat{\bbeta}(\lambda, k)$ are the regularized estimators,
% and $s(\lambda, k)$ is the number of nonzero coefficient groups in $\widehat{\bbeta}(\lambda, k)$.
% }

%And then optimal pair of $(\lambda, k)$ is selected by minimizing $\textrm{GIC}(\lambda, k)$.

% Model selection is achieved through turning parameter $\lambda$ as well as degree freedom (df)  $k$ of basis $\bPhi(t)$.
% To choose $\lambda$ and $k$, we conduct cross validation and adopt the pair $(\widehat{k},\widehat{\lambda})$ that give the minimum out-of-sample prediction error.
% %For each k,
% Let $\lambda_{\textrm{max}}=\lambda_1 <\cdots<\lambda_n=\lambda_{\textrm{min}}$. By warm-start strategy, conduct algorithm
% %$\ref{alg:CCGL}$
% starting from $\lambda_1$ and set
% %the result for $\lambda_m$,
% $\widehat{\beta}_{0}(\lambda_m)$, $\widehat{\bbeta}_{c}(\lambda_m)$ and $\widehat{\bbeta}(\lambda_m)$ as initial values for $\lambda_{m+1}$.

\subsection{On Discrete Time Observations} \label{sec:est:discrete}
\label{subsec:Esit}

So far we have treated the integrated design matrix $\Z$ defined in \eqref{eq:Z} as given. In practical situations, however, the functional compositional predictors are most often not observed continuously but at discrete points, so $\Z$ can not be computed exactly. It is preferable that the induced uncertainty is considered in statistical modeling. In functional regression with a scalar response, \citet{Ramsay2005} discussed using truncated basis expansions for both the functional predictor and the functional coefficient curve to convert the infinite dimensional problem to finite dimensional, where truncation can be viewed as a type of regularization. Integrals were approximated by finite Riemann sums with discrete observations. The subsequent methodological development in functional regression has mainly followed along this general strategy, with various choices of basis functions and associated regularization approaches \citep{Morris2015}. For example, a functional predictor could be expanded by its eigenbasis via a functional principal component analysis, and the coefficient function could be expanded either by the same eigenbasis or by other basis such as wavelet or spline.

Due to the nature of the compositional data, ideally the functional compositions shall be expanded by a multivariate basis that preserves the simplex structure under truncation or other types of regularization, which however, to the best of our knowledge, is not yet available. In essence, a multivariate functional principal component analysis for compositional data, or a joint modeling approach of both the functional compositions and the regression, is needed, which is beyond the scope of the current work. %These extensions are discussed in Section \ref{sec:dis}.

%As pointed out by XX, basis functions are the building blocks of functional data analysis and typically determine the mechanism by which regularization is done.

%The functional basis used to expand the functional predictor and the coefficient function may or may not be the same.

%There has been a rich literature on functional regression; see XX for a recent review. One approach is based on estimating $\X(t)$, e.g., by functional principal component analysis; this, together with a basis expansion of $\bbeta(t)$, turns the problem to low dimensional space. Another approach is to expand the covariate X and the coefficient function beta in the same functional basis, such as the B-spline basis or eigenbasis in

%Each basis function defines a linear combination among the locations within the function, which in effect induces a correlation among those functional regions sharing loadings of high magnitude, and thus establishing a specific framework for borrowing strength. Their use also allows finite numbers of basis coefficients to yield estimates and inference in an infinite-dimensional function space. Some of the most commonly used basis functions in FDA are splines, Fourier series, wavelets, and prin- cipal components, and each is suited for functions with certain characteristics.

%So estimators of one has to replace integrals by summations.

%We now discuss model estimation based on discrete time observations on $\X(t)$. Suppose we observe
% $$
% \y=[y_1,\cdots, y_n] \trans, \Z_c=[\z_{c1}\trans,\cdots,\z_{cn}\trans],
% $$

For the preterm infant study, we take a pragmatic way of lifting the discrete-time data to continuous time. In this study, stool sample of each baby was collected daily whenever available; this resulted in a good coverage rate, with on average 12.2 daily samples for each infant over a 24-day study period. Also, biologists believe that the gut microbiome compositions change continuously over time. As such, we simply apply linear interpolation to obtain continuous time compositional curves. %which automatically preserves the simplex structure of the data at any time point within the domain. 
It can be readily seen that the linear interpolation approach amounts to compute $\Z$ defined in \eqref{eq:Z} using the trapezoid rule. 

Specifically, suppose for each $i=1,\cdots,n$, we observe
$\x_i(t)=[x_{i1}(t),\cdots,\allowbreak x_{ip}(t)]\trans$ at discrete time points $t_{i,v} \in \mathbb{T} = [T_1, T_2]$, for $v=1,\cdots,m_i$. That is, different subjects may be observed at different sets of time points in $\mathbb{T}$. Correspondingly, we have
$$
%\z_i(t) = \log(\x_i(t)) =[z_{i1}(t),\cdots, z_{ip}(t)]\trans, 
\z_i(t) = [z_{i1}(t),\cdots, z_{ip}(t)]\trans,
\quad t=t_{i,1},\cdots,t_{i,m_i},\; i=1,\cdots,n.
$$
Recall that $\Z = \int_{t\in \mathbb{T}} \Z(t)(\I_p \otimes \bPhi(t)\trans) dt \in \mathbb{R}^{n\times (pk)}$. Let $\Z=[\Z_1,\cdots,\Z_p] \in \mathbb{R}^{n \times (pk)}$ with $\Z_j=[z_{ijl}]_{n \times k} \in \mathbb{R}^{n\times k}$ for $j=1,\cdots,p$. Adopting linear interpolation, the entries of $\Z$ are computed using the trapezoid rule as follows,
\begin{align}
z_{ijl}  = &\sum_{v=2}^{m_i} \big(\phi_l(t_{i,v-1})z_{ij}(t_{i,v-1}) + \phi_l(t_{i,v})z_{ij}(t_{i,v}) \big) \frac{t_{i,v}-t_{i,v-1}}{2} \notag\\
                 & + \phi_l(t_{i,1})z_{ij}(t_{i,1})(t_{i,1}-T_0) + \phi_l(t_{i,m_i})z_{ij}(t_{i,m_i})(T_1 - t_{i,m_i}) ,\label{eq:int}
\end{align}
for $l=1,\cdots,k$. In what follows, unless otherwise noted, the integrals in the case of discrete data are computed using the above trapezoid rule.

\section{Theoretical Perspectives}\label{sec:th}

Here we attempt to provide some theoretical perspectives of two questions of practical concerns: (1) whether it is indeed beneficial to use the linearly constrained formulation rather than a naive baseline formulation, which chooses an arbitrary reference component to perform the log-ratio transformation of the compositional predictors and then proceeds with an unconstrained group lasso regression, and (2) whether the proposed method can accurately identify the relevant compositional predictors.

We first describe the setup. Our analysis is under the setting when the basis expansion in \eqref{eq:basis} holds and the integrated design matrix $\Z$ is available. The results are non-asymptotic, where both the number of functional predictors $p$ and the degrees of freedom of the basis functions $k$ are allowed to grow with the sample size $n$. For any $\bbeta = [\bbeta_1\trans,\ldots, \bbeta_p\trans]\trans \in \mathbb{R}^{pk}$, define $\bbeta_{\bar{r}}\in \mathbb{R}^{(p-1)k}$ as a subvector of $\bbeta$ by removing its $r$th component $\bbeta_r$, for each $r=1,\ldots, p$. Let $\mathcal{J} \subset \{1,\ldots,p\}$ be an index set, and denote $\bbeta_{\mathcal{J}}$ be a subvector of $\bbeta$ consisting of $\bbeta_j$, $j \in \mathcal{J}$. Denote $\mathcal{J}^c$ as the complement of $\mathcal{J}$. 
Recall that $\X(t) = [\x_1(t),\ldots, \x_n(t)]\trans \in \mathbb{R}^{n\times p}$ %$\X(t) = [\x_1\trans(t),\ldots, \x_n\trans(t)]\in \mathbb{R}^{n\times p}$, 
$\Z(t)=[z_{ij}(t)] \in \mathbb{R}^{n\times p}$ with $z_{ij}(t) = \log(x_{ij}(t))$, and $\widetilde{\Z}_{\bar{r}}(t)=[\widetilde{z}_{ij}(t)] \in \mathbb{R}^{n\times (p-1)}$ with $\widetilde{z}_{ij}(t) = \log(x_{ij}(t)/x_{ir}(t))$ for each $r=1,\ldots, p$. Moreover, due to \eqref{eq:basis}, we define
$\widetilde{\Z}_{\bar{r}} = \int_{t\in \mathbb{T}} \widetilde{\Z}_{\bar{r}}(t)(\I_p \otimes \bPhi(t)\trans) dt \in \mathbb{R}^{n\times (p-1)k}$
and $\Z = \int_{t\in \mathbb{T}} \Z(t)(\I_p \otimes \bPhi(t) \trans) dt \in \mathbb{R}^{n\times (pk)}$
as in \eqref{eq:Z}. Write $\widetilde{\Z}_{\bar{r}} = [\widetilde{\Z}_{\bar{r},1}, \ldots, \widetilde{\Z}_{\bar{r},r-1}, \widetilde{\Z}_{\bar{r},r+1}, \ldots, \widetilde{\Z}_{\bar{r},p}]$ with each $\widetilde{\Z}_{\bar{r},j} \in \mathbb{R}^{n\times k}$. Write $\Z = [\Z_1, \ldots, \Z_p]$ with each $\Z_{j} \in \mathbb{R}^{n\times k}$. Let $\bPsi_{\bar{r},j} = \widetilde{\Z}_{\bar{r},j}\trans \widetilde{\Z}_{\bar{r},j}/n$, for $r=1,\ldots, p$, $j = 1,\ldots, p$ and $j\neq r$. It boils down to analyze the constrained linear model with grouped predictors in \eqref{eq:fmodel2}. For simplicity, we omit the intercept and the control variables, and write the model as
\begin{align*}
\y = \Z\bbeta^* + \be, \qquad \mbox{s.t. } \sum_{j=1}^{p}\bbeta_j^* = \0, %\label{eq:simplefmodel2} %\label{eq:fmodel2}
\end{align*}
where $\bbeta^* = [\bbeta_1\strans,\ldots, \bbeta_p\strans]\trans \in \mathbb{R}^{pk}$. Recall that $\mathcal{S} = \{j;\bbeta_j^*(t)\neq \0, j = 1,\ldots,p.\} = \{j;\bbeta_j^*\neq \0, j = 1,\ldots,p.\}$, and $s^* = |\mathcal{S}| \ll p$.

The proposed constrained group lasso
estimator is,
\begin{align}\label{eq:glasso2}
\widehat{\bbeta} = \arg\min_{\bbeta}\left\{\frac{1}{2n}\| \y - \Z\bbeta\|^2  + \lambda \sum_{j=1}^{p} \|\bbeta_j\|\right\}, \qquad \mbox{s.t.} \sum_{j=1}^{p}\bbeta_j=\0.
\end{align}
This estimator satisfies that $\widehat{\bbeta}_r = -\sum_{j\neq r}^{p}\widehat{\bbeta}_j$. Therefore, it holds true that for any $r = 1,\ldots, p$,
\begin{align*} %\label{eq:glasso2r}%\label{eq:glasso2}
%\widehat{\bbeta}_{\bar{r}} = \arg\min_{\bbeta_{\bar{r}}} \frac{1}{2n}\| \y - \widetilde{\Z}_{\bar{r}}\bbeta_{\bar{r}}\|^2  + \lambda \sum_{j\neq r}^{p} \|\bbeta_j\| + \lambda \|\sum_{j\neq r}^p \bbeta_j\| .
\widehat{\bbeta}_{\bar{r}} = \arg\min_{\bbeta_{\bar{r}}}  \left\{  \frac{1}{2n}\| \y - \widetilde{\Z}_{\bar{r}}\bbeta_{\bar{r}}\|^2  + \lambda \sum_{j\neq r}^{p} \|\bbeta_j\| + \lambda \|\sum_{j\neq r}^p \bbeta_j\| \right\}.
\end{align*}
%On the other hand, one naive method is the so-called baseline method, which chooses an arbitrary reference component to perform the log-ratio transformation of the compositional predictors and then proceeds with an unconstrained group lasso regression.
On the other hand, as to the baseline method, when the $r$th component is choosing as the reference level, the estimator is given by
\begin{align}\label{eq:glasso3}
%\widetilde{\bbeta}_{\bar{r}} = \arg\min_{\bbeta_{\bar{r}}} \frac{1}{2n}\| \y - \widetilde{\Z}_{\bar{r}}\bbeta_{\bar{r}}\|^2  + \lambda \sum_{j\neq r}^{p} \|\bbeta_j\|.
\widetilde{\bbeta}_{\bar{r}} = \arg\min_{\bbeta_{\bar{r}}}  \left\{  \frac{1}{2n}\| \y - \widetilde{\Z}_{\bar{r}}\bbeta_{\bar{r}}\|^2  + \lambda \sum_{j\neq r}^{p} \|\bbeta_j\|  \right\}.
\end{align}

Our analysis follows and extends the work by \citet{Lounici2011} on group lasso to the case of constrained group lasso in \eqref{eq:glasso2} arising from functional compositional data analysis. All the proofs are provided in Section \ref{sec:supp:th} of Supplementary Materials.

%Our theoretical analysis mainly concerns the properties of the constrained group lasso estimator, including its finite sample error bound and its performance on recovering the sparsity pattern.

\begin{assumption}\label{as:2}
The error terms $e_1,\ldots, e_n$ are independently and identically distributed as $\mbox{N}(0,1)$ random variables.
\end{assumption}

\begin{assumption}[Restricted Eigenvalue Condition (RE)]\label{as:3}
There exists $\kappa>0$, such that
\begin{align*}
\min\left\{ \frac{\|\Z\bDelta\|}{\sqrt{n}\|\bDelta_{\mathcal{J}}\|}:\right.
& |\mathcal{J}| \leq s^*, \bDelta\in \mathbb{R}^{pk} \neq \0, \sum_{j=1}^{p}\bDelta_j=\0,\\
&\left. \sum_{j\in \mathcal{J}^c}\|\bDelta_j\| + \min_j\|\bDelta_j\| \leq 3 \sum_{j\in \mathcal{J}}\|\bDelta_j\|.\right\} \geq \kappa.
\end{align*}
\end{assumption}

% For sparse models, it is most likely true that $\min_{j}\|\bDelta_j\| = 0$, so the inequality regarding $\bDelta_j$s becomes the same as that in the regular group lasso model. We can compare with the above RE condition with that for the estimator in \eqref{eq:glasso3}. There are two cases. First, consider the case when $\bbeta_r^* = \0$. Then the RE is expressed as
% \begin{align*}
% \min\left\{ \frac{\|\widetilde{\Z}_{\bar{r}}\bDelta_{\bar{r}}\|}{\sqrt{n}\|\bDelta_{\mathcal{J}}\|}:\right.
%  |\mathcal{J}| \leq s^*, \bDelta_{\bar{r}}\in \mathbb{R}^{(p-1)k} \neq \0,
% \left. \sum_{j\in \mathcal{J}^c, j \neq r}\|\bDelta_j\| \leq 3 \sum_{j\in \mathcal{J}}\|\bDelta_j\|.\right\} \geq \kappa.
% \end{align*}

\begin{theorem}[Error Bounds]\label{th:1}
Suppose Assumptions \ref{as:2}--\ref{as:3} hold. Choose
$$
\lambda \geq \min_{r}\max_{j\neq r}\frac{2\sigma}{\sqrt{n}} \sqrt{\mbox{tr}(\bPsi_{\bar{r},j}) + 2\sigma_{\max}(\bPsi_{\bar{r},j})(2q\log(p-1)+\sqrt{kq\log(p-1)})}.
$$
Then, with probability at least $1-2(p-1)^{1-q}$, the constrained group lasso estimator $\widehat{\bbeta}$ in \eqref{eq:glasso2} satisfies that
\begin{align*}
\frac{1}{n}\|\Z(\widehat{\bbeta} - \bbeta^*)\|^2 \leq \frac{16\lambda^2s^*}{\kappa^2},\\%\label{eq:th11}\\
\sum_{j=1}^{p}\|\widehat{\bbeta}_j - \bbeta_j^*\| + \min_j \|\widehat{\bbeta}_j - \bbeta_j^*\| \leq \frac{16\lambda s^*}{\kappa^2}.%\label{eq:th12}
\end{align*}
\end{theorem}

%One naive method is the so-called baseline method, which chooses an arbitrary %reference component to perform the log-ratio transformation of the compositional predictors and then proceeds with an unconstrained group lasso regression.
It is interesting to compare with the baseline approach in \eqref{eq:glasso3}, for which once a baseline $r$ is chosen, its theoretical property mimics that of the regular group lasso model with $p-1$ groups \citep{Lounici2011}. Due to the linear constraints, the restricted set of $\bDelta$ in Assumption \ref{as:3} for which the minimum is taken is smaller than that of the regular group lasso estimator. As such, the condition for the constrained model becomes weaker in general. Also, in Theorem \ref{th:1}, the choice of $\lambda$, which directly impacts the final estimation error rate, is taken as a minimal value over $r$, the choice of the baseline. Therefore, in view of the RE condition and the choice of $\lambda$, our results reveal that the proposed method is capable of achieving the best possible performance of the baseline method under a possibly weaker condition.

%{\color{red} (Can we make this statement in a more precise way? This is trivial when an irrelevant component is choosing as the baseline, but may not so otherwise?)}

\begin{assumption}[$\beta$-min Condition]\label{as:4}
Choose the same $\lambda$ as in Theorem \ref{th:1}. Assume that
$$
\min_{j\in \mathcal{S}}\|\bbeta_j^*\| > \frac{16\lambda s^*}{\kappa^2}.
$$
\end{assumption}

\begin{corollary}[Selection Consistency]\label{th:2}
Suppose Assumptions \ref{as:2}--\ref{as:4} hold. Let
$$
\widehat{\mathcal{S}} = \{j: \|\widehat{\bbeta}_j\|> \frac{8\lambda s^*}{\kappa^2}\}.
$$
Then, with probability at least $1-2(p-1)^{1-q}$, we have that $\widehat{\mathcal{S}} = \mathcal{S}$.
\end{corollary}

Corollary \ref{th:2} reveals the ``overselection'' phenomenon due to convex penalization; see, e.g., \citet{WeiHuang2010}. That is, the proposed constrained group lasso estimator in general does not miss important variable groups/components, albeit overselecting some irrelevant ones. As such, a thresholding operation is preferred in order to recovery the correct sparsity pattern exactly. However, the theoretical threshold is not available in practice, as it involves unknown quantities such as $\sigma^2$ and $\kappa$. Nevertheless, the results provide guarantee that using the original estimator can avoid false negatives at the expense of some false positives, which is acceptable in many applications.

% <<< Theory >>>

% >>> Simulation <<<
\section{Simulation}\label{sec:sim}
We conduct simulation studies to compare the performance of our proposed sparse functional log-contrast regression via constrained group lasso (CGL) in \eqref{eq:glasso}, the baseline approach in the form of \eqref{eq:glasso3} via group lasso (BGL) in which the reference level is chosen randomly, and the naive approach of group lasso (GL) in which the zero-sum constraints are ignored in \eqref{eq:glasso}, and cross sectional method (I) of taking average of observations along time (Average), and cross sectional method (II) of considering the snapshot of the most significant time point (Snapshot).
 %i.e.,
%\begin{enumerate}
    %\item[(i)] naive method
    % $$
    % \min_{\beta_0,\bbeta_c,\bbeta} \frac{1}{2n}\| \y -  \beta_0 \1_n- \Z_c \bbeta_c -  \Z\bbeta\|^2 +  \lambda \sum_{j=1}^{p} \| \bbeta_j\|.
    % $$
%    which simply ignores the zero-sum constraints in problem (\ref{eq:fmodel}).
    %     \item[(ii)] BGL method
    %     $$
    %     \min_{\beta_0,\bbeta_c,\bbeta_{\bar{r}}} \frac{1}{2n}\| \y -  \beta_0 \1_n- \Z_c \bbeta_c -  \widetilde{\Z}_{\bar{r}} \bbeta_{\bar{r}}\|^2 +  \lambda \sum_{j \neq r}^{p} \| \bbeta_j\|,
    %     $$
    % where $\widetilde{\Z}_{\bar{r}}$ is the matrix after taking integration of $ \widetilde{\Z}_{\bar{r}}(t)$ with basis. Baseline $r$ is random selected.
%\end{enumerate}

The compositional data are generated as follows. We first generate $M$ time points within the interval $[0,1]$, i.e., $0=t_1<\cdots<t_M=1$. For inducing dependence between time points, we consider an autoregressive correlation structure, $\bSig_{T} = [\rho_{T}^{|\mu - \nu|}]_{M\times M}$, where $1\leq \mu,\nu\leq M$; for inducing dependence between compositions, we consider a compound symmetry correlation structure, $\bSig_{X} =[\rho_{X}^{I(j=j')}]_{p\times p}$, where $1\leq j,j'\leq p$ and $I(\cdot)$ is the indicator function. The ``non-normalized'' data for each subject $i$, $i=1,\ldots, n$, are then generated from multivariate normal distribution as $\w_i = [\w_i({t_1})\trans, \cdots, \w_i(t_M)\trans]\trans \sim \mbox{N}(\0, \sigma_{X}^2(\bSig_{T} \otimes \bSig_{X}))$,  where each $\w_i(t_{\nu}) \in \mathbb{R}^p$ for $\nu= 1,\ldots, M$. Finally, the compositional data are obtained as $x_{ij}(t_{\nu}) = \exp(w_{ij}(t_{\nu})) /\sum_{j=1}^p \exp( w_{ij}(t_{\nu}))$, for $i=1,\ldots,n$, $j=1,\ldots,p$ and $\nu = 1,\ldots, M$. The regression curves $\bbeta^*(t)$ are generated as $\B{B}^*\bPhi(t)$, where $\bPsi(t)$ is from a set of cubic spline basis computed using the \verb|bs| %\textrm{bs} 
function in the R package \verb|splines| %\textrm{splines} 
with $t \in \{t_1,\ldots, t_M\}$ and degrees of freedom set to 5. The first three rows of $\B{B}^*$ are set as $[1, 0, 1 , 0, -0.5]$, $[0, 0,  -1,  0,  1]$ and $[-1, 0, 0, 0, -0.5]$, respectively, and the rest are set to zero. The intercept is set to be $\beta_0^* =1$ and for simplicity we do not consider additional control. The error terms are generated as independent $\mbox{N}(0, \sigma^2)$ random variables where $\sigma^2$ is set to control the signal to noise ratio (SNR). Finally, the response $\y$ is generated by model \eqref{eq:fmodel}, where the integral is computed as in \eqref{eq:int}. We have experimented with $(n,p)\in \{(50,30), (100, 30), (100,100),(100,200)\}$ and parameter settings $M = 20$, $\sigma_X^2 = 9$, $\rho_T \in \{0, 0.6\}$, $\rho_X = \{0, 0.6\}$ and $\mbox{SNR} = \{2, 4\}$. The simulation is repeated 100 times under each setting.

%The covariates generated thus follow a logistic normal distribution at each time point \citet{Aitchison1980}.

%% Then let $\w_i(t_{\mu})= \big(\I_p \otimes \bPsi(t_{\mu}\big) \trans)\Q_i$, where $\Q_i \sim N_{pq}(\0, \sigma^2 \I_{pq})$. Last we get covariates $x_{ij}(t_{\mu})$ by the transformation $x_{ij}(t_{\mu}) = \exp(w_{ij}(t_{\mu}))/\sum_{j=1}^p \exp(w_{ij}(t_{\mu}))$. At each time point $\t_{\mu}$, the covariates generated thus follows a logistic normal distribution \citep{Aitchison1980}. Set $q=5, ns=100$ and $\sigma = 3$.

%% with df=5.
%% Let $\bPsi(t_{\mu})=(\psi_1(t_{\mu}), \cdots,\psi_q(t_{\mu})) \trans \in \mathbb{R}^q$ as
%%B-splines
%% cubic splines with degree of freedom (df) $q$ estimated at $t_{\mu}$, $\mu = 1, \cdots, \textrm{ns}$.

%% To reflect the fact that observations is finite, we let compositional covariates at each time point has probability $\%obs=0.1, 0.2$ or 0.6 to be

%We used five performance measures for our comparisons.

The prediction error (Pred) is measured by $\|\y_{te} - \Z_{te}\widehat{\bbeta} \|^2/n_{te}$,  computed from an independently generated test sample $(\y_{te}; \X_{te}(t), t\in \{t_1,\ldots, t_M\})$ of size $n_{te} = 500$.
The estimation error (Est) is measured by $\sum_{j=1}^p (\int_{[0,1]} |\widehat{\beta}_j(t) - \beta_j^*(t) |^2\, dt)^{1/2}/p$. %$\sum_{j=1}^p (\int_{[0,1]} |\widehat{\bbeta}_j(t) - \bbeta_j^*(t)|^2\, dt)^{1/2}/p$.
%and $\ell_{2}.\textrm{const} = \|\hat{\beta_0} - \beta_0\|_2$.
%% $\ell_{1,\infty} = \max_{j} \int |\hat{\bbeta}_j(t) - \bbeta_j^*(t)|\, dt$,  $\ell_{1,1} = \sum_{j=1}^p \int |\hat{\bbeta}_j(t) - \bbeta_j^*(t)|\, dt$,  $\ell_{2,\infty} = \max_{j} (\int |\hat{\bbeta}_j(t) - \bbeta_j^*(t)|^2\, dt)^{1/2}$ and  $\ell_{2,1} = \sum_{j=1}^p (\int |\hat{\bbeta}_j(t) - \bbeta_j^*(t)|^2\, dt)^{1/2}$.
For variable selection of the compositional components, we report the false positive rate (FPR) and the false negative rate (FNR), based on the sparsity patterns of $\widehat{\bbeta}(t)$ and $\bbeta^*(t)$. We have experimented with both 10-fold cross validation (CV) and GIC for selecting tuning parameters $k$ and $\lambda$. As shown in Corollary \ref{th:2}, a thresholding of the estimator is preferred for the purpose of variable selection, although the ideal threshold is not available in practice. Here with the same spirit and based on empirical evidence, we define the selected index set $\widehat{\mathcal{S}}$ based on the relative magnitudes of the $p$ estimated coefficient curves:
$$
%\tilde{s} =\{j \in \{1,\ldots,p\};
%           \big(\int_{[0,1]} \widehat{\beta}_j^2(t)\, dt \big)^{1/2} / \{\sum_{j=1}^{p} \big(\int_{[0,1]} \widehat{\bbeta}_j^2(t)\, dt \big)^{1/2}\}
%           \ge 1/p
%           \}
\widehat{\mathcal{S}} = \{j;
\big(\int_{[0,1]} \widehat{\beta}_{j}^2(t)\, dt \big)^{1/2} / \{\sum_{j=1}^{p} \big(\int_{[0,1]} \widehat{\beta}_{j}^2(t)\, dt \big)^{1/2}\} \ge 1/p,
j = 1, \cdots, p
\}.
$$
That is, we only count the components whose relative ``energy'' exceeds the average $1/p$ as selected.

% which excluding groups with small magnitude in term of function $\ell_2$ norm of $\hat{\bbeta}_i(t)$. Then measurements FPR and FNR in CV criterion were caculated based on $\tilde{s}$.
% Following \citet{FanTang2013}, we also propose GIC criterion, referred as GIC, for model \eqref{eq:glasso} as
% $$
% \textrm{GIC}(\lambda, k) = \log \big( \hat{\sigma}^2(\lambda, k) \big) + \big(s(\lambda, k) - 1 \big) k
% \log \big( \max\{pk +1, n\}\big)
% \frac{\log\log n}{n},
% $$
% %% \vee
% where $\hat{\sigma}^2(\lambda, k) = \|\y - \Z\hat{\bbeta}(\lambda, k) - \hat{\beta}_0(\lambda, k) \1_n \|^2_2/n$, $\hat{\bbeta}(\lambda, k)$ is the regularized estimator for $\bbeta^*$ and $s(\lambda, k)$ is the number of nonzero coefficients groups in $\hat{\bbeta}(\lambda, k)$. And then optimal pair of $(\lambda, k)$ is selected by minimizing $\textrm{GIC}(\lambda, k)$.

%For measurements PE, $\ell_{2,1}/p$ and $\ell_{2}.\textrm{const}$, results were enlarged by multiplying 10.

%With prefixed $\textrm{df} = k$ $\in \{4, 5\}$, the tuning parameter pairs $\lambda$ and $k$ were selected by 10 folds cross-validation, referred as CV, as described in \ref{subsec:Esit}. Since cross-validation regularization method tends to select larger model to trade off between model selection consistency and prediction accuracy, in result showed we re-filter the nonezero groups. Once $(\lambda, k)$ were selected, $\hat{\bbeta}$ and $\hat{\bbeta}(t)$ were determined. Then

The simulation results for $(n,p) = (50,30)$ and $(n,p) = (100,200)$ with $\mbox{SNR} = 4$ are reported in Tables \ref{tbl:n50p30snr4} -- \ref{tbl:n100p200snr4}. The two naive methods, Average and Snapshot perform much worse in prediction than other methods; they tend to miss important variables as seen from their high FNR values. (To save space, we do not show the results of Average and Snapshot with GIC tuning.) In general, CGL shows better predictive and selection performance than both GL and BGL, and in some cases the improvement can even be substantial;  (We have also tried the unpenalized least squares estimator, which fails miserably in prediction and hence is omitted.) The BGL method performs the worst among the three. The two tuning methods, CV and GIC, show quite difference behaviors: the former generally yields larger false positive rates and much smaller false negative rates than the latter. Indeed, this is consistent with the theoretical results in Section \ref{sec:th} that the proposed convex regularized estimation approach has a tendency of over-selection when tuned based on optimizing predictive performance. Nevertheless, the CV-tuned estimators rarely miss important components and performs much better in prediction comparing to their GIC tuned counterparts. Therefore, CV may be preferable in practice when one cares more about prediction and can afford some false alarms for the capture of all the relevant signals.

% n50p30snr4
\begin{table}[htp!]
    \centering
    \caption{Simulation results for $(n,p)= (50,30)$ and $\mbox{SNR} = 4$. Reported are the average values over 100 simulation runs, with the standard deviations in parentheses. For better presentation, the values of Est and Pred are multiplied by 10.}%Minimum values are boldfaced for each comparison.
    \label{tbl:n50p30snr4}
    %\footnotesize
    \begin{tabular}{lllccccc}
        \toprule
        $ (\rho_{X}, \rho_{T})$ %$(\rho_x, \rho_w)$ 
        & Criterion & Method & Est & Pred & FPR (\%) & FNR (\%) \\%& Est & $L_{2}.\textrm{const}$ \\
        \toprule
        $ (0, 0)$ & CV & BGL & 0.25 (0.01) & 0.39 (0.01) & 28.85 (1.28) &  0.00 (0.00)  \\%& 0.25 (0.01) & 0.22 (0.02) \\
        &    & GL &  0.23 (0.01)  & 0.39 (0.01) & 27.48 (1.35) &  0.00 (0.00)  \\%& 0.23 (0.01) & 4.19 (0.28) \\
        &    & CGL & 0.23 (0.01) &  0.34 (0.01) & 29.22 (1.43) &  0.00 (0.00)  \\%& 0.23 (0.01) & 0.22 (0.02) \\
        &    & Average &         & 2.03 (0.03) &  12.00 (1.37) & 66.00 (3.45) \\
        &    & Snapshot &        & 2.11 (0.05) & 16.74 (1.77) & 53.00 (3.45) \\
        [1mm]
        & GIC & BGL & 0.33 (0.01) & 1.46 (0.06) & 4.04 (0.19)  & 48.00 (3.33) \\%& 0.33 (0.01) & 0.47 (0.04) \\
        &     & GL & 0.31 (0.01)  & 1.44 (0.05) & 0.19 (0.08)  & 52.67 (2.69) \\%& 0.31 (0.01) & 5.71 (0.41) \\
        &     & CGL &  0.29 (0.01) & 1.24 (0.05) & 1.63 (0.24)  & 20.00 (2.37) \\%& 0.29 (0.01) & 0.43 (0.03) \\
        \midrule 
        $ (0, 0.6)$ & CV & BGL & 0.28 (0.01) & 1.27 (0.04) & 30.70 (1.48) & 0.33 (0.33)  \\%& 0.28 (0.01) & 0.37 (0.03) \\
        &    & GL & 0.26 (0.01)  & 1.21 (0.03) &  29.04 (1.40) &  0.00 (0.00)  \\%& 0.26 (0.01) & 5.12 (0.38) \\
        &    & CGL &  0.25 (0.01) &  1.13 (0.03) & 29.67 (1.43) &  0.00 (0.00)  \\%& 0.25 (0.01) & 0.34 (0.03) \\
        &    & Average &         & 5.58 (0.14) &  19.67 (1.96) & 34.00 (3.45) \\
        &    & Snapshot &        & 5.31 (0.10) & 23.44 (1.60) & 22.67 (1.83) \\   
        [1mm]   
        & GIC & BGL & 0.34 (0.01) & 4.61 (0.16) & 3.74 (0.16)  & 52.67 (2.60) \\%& 0.34 (0.01) & 0.66 (0.05) \\
        &     & GL &  0.31 (0.00)  & 3.93 (0.12) & 0.11 (0.06)  & 51.67 (2.39) \\%& 0.31 (0.00) & 7.54 (0.41) \\
        &     & CGL & 0.31 (0.01) & 3.91 (0.17) & 1.52 (0.24)  & 23.67 (2.19) \\%& 0.31 (0.01) & 0.63 (0.05) \\
        \midrule
        $ (0.6, 0)$ & CV & BGL & 0.25 (0.01) & 0.15 (0.01) & 29.26 (1.35) &  0.00 (0.00)  \\%& 0.25 (0.01) & 0.13 (0.01) \\
        &    & GL & 0.24 (0.01)  & 0.16 (0.00) & 29.93 (1.42) &  0.00 (0.00)  \\%& 0.24 (0.01) & 4.28 (0.39) \\
        &    & CGL &  0.23 (0.01) &  0.14 (0.00) &  29.07 (1.22) &  0.00 (0.00)  \\%& 0.23 (0.01) & 0.12 (0.01) \\
        &    & Average &         & 0.80 (0.01) &  14.63 (1.70) & 57.33 (3.52) \\
        &    & Snapshot &        & 0.85 (0.02) & 16.70 (1.73) & 57.00 (3.29) \\
        [1mm]
        & GIC & BGL & 0.34 (0.01) & 0.65 (0.02) & 3.81 (0.19)  & 56.33 (2.67) \\%& 0.34 (0.01) & 0.29 (0.02) \\
        &     & GL & 0.32 (0.01)  & 0.62 (0.02) &  0.19 (0.08)  & 59.33 (2.25) \\%& 0.32 (0.01) & 4.74 (0.30) \\
        &     & CGL &  0.30 (0.01) &  0.54 (0.02) & 1.63 (0.22)  & 22.67 (2.22) \\%& 0.30 (0.01) & 0.25 (0.02) \\
        \midrule
        $ (0.6, 0.6)$ & CV & BGL & 0.29 (0.01) & 0.53 (0.02) & 33.52 (1.38) & 0.33 (0.33)  \\%& 0.29 (0.01) & 0.25 (0.02) \\
        &    & GL & 0.26 (0.01)  & 0.49 (0.02) &  30.22 (1.31) &  0.00 (0.00)  \\%& 0.26 (0.01) & 4.82 (0.37) \\
        &    & CGL & 0.25 (0.01) &  0.45 (0.01) & 30.37 (1.44) &  0.00 (0.00)  \\%& 0.25 (0.01) & 0.22 (0.02) \\
        &    & Average &         & 2.02 (0.04) &  22.81 (1.89) & 26.33 (2.81) \\
        &    & Snapshot &        & 2.10 (0.03) & 22.85 (1.60) & 25.67 (1.76) \\
        [1mm]
        & GIC & BGL & 0.35 (0.01) & 1.85 (0.06) & 3.81 (0.15) & 53.67 (2.59)  \\%& 0.35 (0.01) & 0.47 (0.04) \\
        &     & GL & 0.32 (0.00)  & 1.69 (0.05) &  0.11 (0.06) & 57.67 (2.00)  \\%& 0.32 (0.00) & 5.54 (0.40) \\
        &     & CGL &  0.31 (0.01) &  1.52 (0.06) & 1.74 (0.23) & 25.00 (2.24)  \\%& 0.31 (0.01) & 0.41 (0.03) \\
        \bottomrule
    \end{tabular}
\end{table}

\begin{table}[htp!]
    \centering
    \caption{Simulation results for $(n,p)= (100,200)$ and $\mbox{SNR} = 4$. The layout is the same as in Table \ref{tbl:n50p30snr4}.}  \label{tbl:n100p200snr4}
    %\footnotesize
    \begin{tabular}{lllccccc}
        \toprule
        $ (\rho_{X}, \rho_{T})$ %$(\rho_x, \rho_w)$
        & Criterion & Method & Est & Pred & FPR (\%) & FNR (\%) \\%& Est & $L_{2}.\textrm{const}$ \\
        \toprule
        $ (0,0)$ & CV & BGL &  0.04 (0.00) & 0.31 (0.01) & 15.28 (0.48) &  0.00 (0.00) \\%& 0.04 (0.00) & 0.14 (0.01) \\
        &   & GL  &  0.04 (0.00) & 0.31 (0.01) & 15.27 (0.48) &  0.00 (0.00) \\%& 0.04 (0.00) & 5.07 (0.35) \\
        &   & CGL &  0.04 (0.00) &  0.29 (0.00) & 15.57 (0.51) &  0.00 (0.00) \\%& 0.04 (0.00) & 0.13 (0.01) \\
        &   & Average &         & 1.98 (0.03) &  3.04 (0.41)  & 73.33 (3.11)\\
        &   & Snapshot &        & 1.99 (0.03) & 4.82 (0.64)  & 59.33 (3.20)\\ 
        [1mm]
        & GIC & BGL & 0.05 (0.00) & 1.45 (0.05) & 0.51 (0.01) & 44.00 (3.07) \\%& 0.05 (0.00) & 0.26 (0.02) \\
        &     & GL  &  0.04 (0.00) & 1.33 (0.05) & 0.01 (0.01) & 46.33 (2.88) \\%& 0.04 (0.00) & 8.22 (0.55) \\
        &     & CGL &  0.04 (0.00) &  1.13 (0.05) & 0.19 (0.03) &  11.67 (1.73) \\%& 0.04 (0.00) & 0.23 (0.02) \\
        \midrule
        $ (0,0.6)$ & CV & BGL &  0.04 (0.00) & 1.02 (0.02) & 16.26 (0.51) & 0.00 (0.00) \\%& 0.04 (0.00) & 0.25 (0.02) \\
        &    & GL  &  0.04 (0.00) & 0.97 (0.02) & 15.62 (0.52) &  0.00 (0.00) \\%& 0.04 (0.00) & 5.70 (0.43) \\
        &    & CGL &  0.04 (0.00) &  0.94 (0.02) & 16.32 (0.50) &  0.00 (0.00) \\%& 0.04 (0.00) & 0.24 (0.02) \\
        &    & Average &         & 5.41 (0.10) & 7.17 (0.70)  & 27.00 (2.71)\\
        &    & Snapshot &        & 5.14 (0.10) &  6.57 (0.56)  & 27.67 (1.26)\\
        [1mm]
        & GIC & BGL & 0.05 (0.00) & 4.15 (0.16) & 0.51 (0.01) & 43.00 (3.01) \\%& 0.05 (0.00) & 0.50 (0.04) \\
        &     & GL  &  0.04 (0.00) &  3.44 (0.12) &  0.01 (0.01) & 42.67 (2.92) \\%& 0.04 (0.00) & 10.37 (0.60) \\
        &     & CGL &  0.04 (0.00) & 3.57 (0.15) & 0.10 (0.02) &  16.00 (1.92) \\%& 0.04 (0.00) & 0.47 (0.04) \\
        \midrule
        $ (0.6,0)$ & CV & BGL &  0.04 (0.00) &  0.12 (0.00) & 14.78 (0.49) &  0.00 (0.00) \\%& 0.04 (0.00) & 0.09 (0.01) \\
        &    & GL  &  0.04 (0.00) &  0.12 (0.00) & 15.44 (0.61) &  0.00 (0.00) \\%& 0.04 (0.00) & 3.93 (0.27) \\
        &    & CGL &  0.04 (0.00) &  0.12 (0.00) & 15.07 (0.55) &  0.00 (0.00) \\%& 0.04 (0.00) & 0.09 (0.01) \\
        &    & Average &         & 0.80 (0.01) & 5.14 (0.68)  & 58.33 (3.80) \\
        &    & Snapshot &        & 0.81 (0.01) &  4.30 (0.49)  & 61.67 (2.82) \\
        [1mm]
        & GIC & BGL & 0.05 (0.00) & 0.55 (0.02) & 0.53 (0.01) & 39.00 (3.39) \\%& 0.05 (0.00) & 0.20 (0.01) \\
        &     & GL  &  0.04 (0.00) & 0.47 (0.02) &  0.02 (0.01) & 36.33 (3.22) \\%& 0.04 (0.00) & 5.72 (0.43) \\
        &     & CGL &  0.04 (0.00) &  0.41 (0.02) & 0.15 (0.03) &  9.67 (1.79)  \\%& 0.04 (0.00) & 0.17 (0.01) \\
        \midrule
        $ (0.6,0.6)$ & CV & BGL &  0.04 (0.00) & 0.41 (0.01) & 16.21 (0.50) &  0.00 (0.00) \\%& 0.04 (0.00) & 0.16 (0.01) \\
        &    & GL  &  0.04 (0.00) & 0.40 (0.01) & 15.30 (0.55) &  0.00 (0.00) \\%& 0.04 (0.00) & 4.47 (0.32) \\
        &    & CGL &  0.04 (0.00) &  0.39 (0.01) & 15.59 (0.50) &  0.00 (0.00) \\%& 0.04 (0.00) & 0.16 (0.01) \\
        &    & Average &         & 2.03 (0.03) &  7.35 (0.63)  & 16.67 (2.25) \\
        &    & Snapshot &        & 2.09 (0.04) & 7.62 (0.67)  & 27.33 (1.29) \\
        [1mm]
        & GIC & BGL & 0.05 (0.00) & 1.76 (0.06) & 0.52 (0.01) & 48.33 (2.93) \\%& 0.05 (0.00) & 0.34 (0.02) \\
        &     & GL  &  0.04 (0.00) & 1.46 (0.05) &  0.01 (0.01) & 47.33 (2.73) \\%& 0.04 (0.00) & 7.89 (0.42) \\
        &     & CGL &  0.04 (0.00) &  1.40 (0.06) & 0.15 (0.03) &  17.00 (1.98) \\%& 0.04 (0.00) & 0.29 (0.02) \\
        \bottomrule
    \end{tabular}
\end{table}

Figure \ref{fig:sim1} show boxplots of prediction errors from CV tuning for $\mbox{SNR}=4$. (The case of $\mbox{SNR}=2$ is reported in Section \ref{sec:supp:sim} of Supplementary Materials; we do note include Average and Snapshot methods as they perform much worse.) The performance of all methods deteriorates when the SNR becomes smaller, the between-component correlation becomes smaller, or the between-time correlation becomes stronger. Small between-component correlation causes the presence of a few dominating compositional components due to the unit-sum constraints, while large between-time correlation makes the functional compositions smooth over time and consequently makes it hard to distinguish the relevant components from the others. %To save space, more simulation results are shown in Supplemental Materials.

\begin{figure}[!htbp]
    \captionsetup[subfigure]{singlelinecheck=true}
    \centering
    \subcaptionbox[singlelinecheck=true]{$n = 50, p = 30$}{
        \includegraphics[ width = 5.5cm, angle = -90]{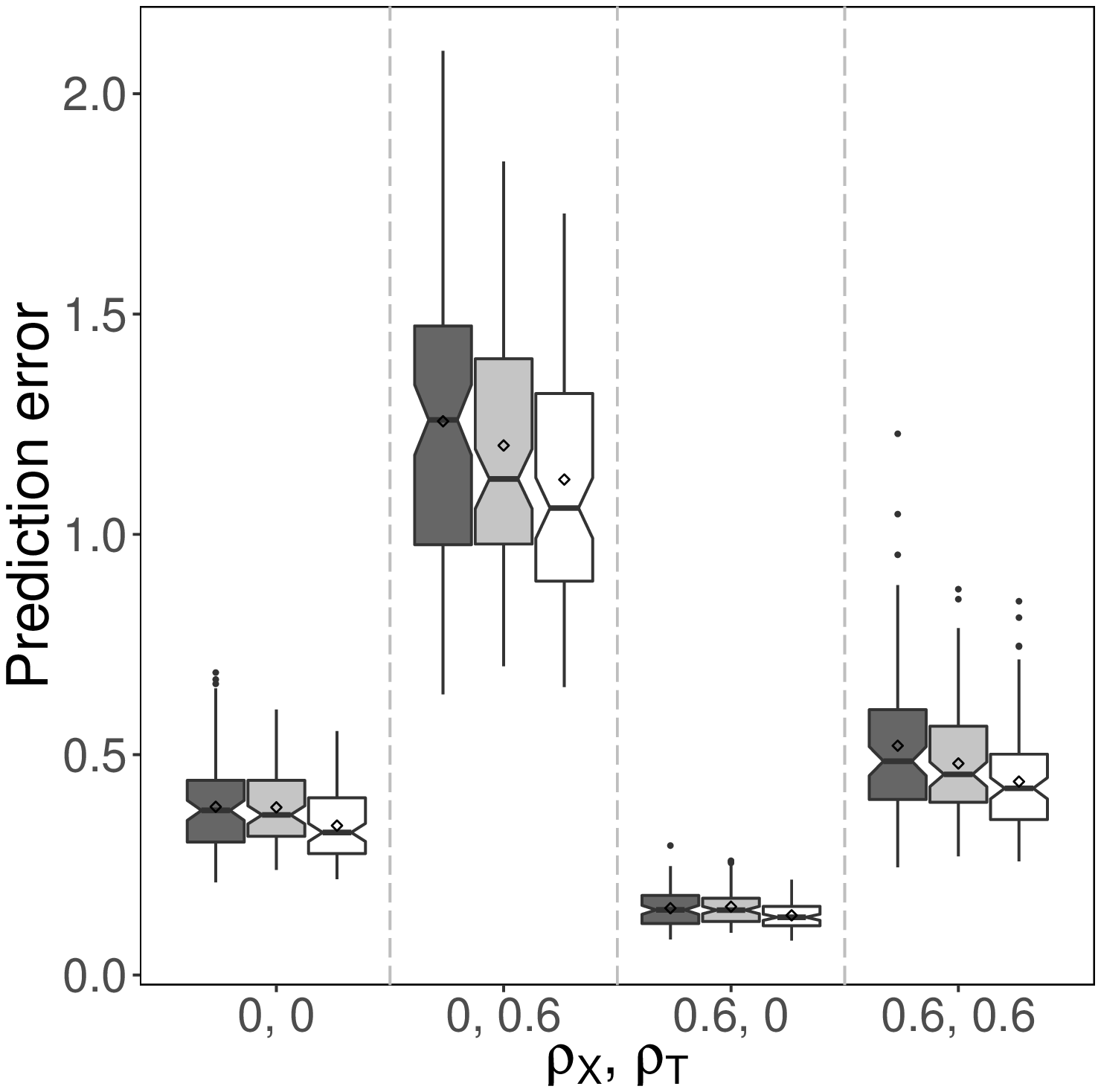}
        %\makebox[0pt][l]{%
        %   \hspace{0.04\columnwidth}
    }\hspace{0.5cm}%
    \subcaptionbox{$n = 100, p = 30$}{
        \includegraphics[ width = 5.5cm, angle = -90]{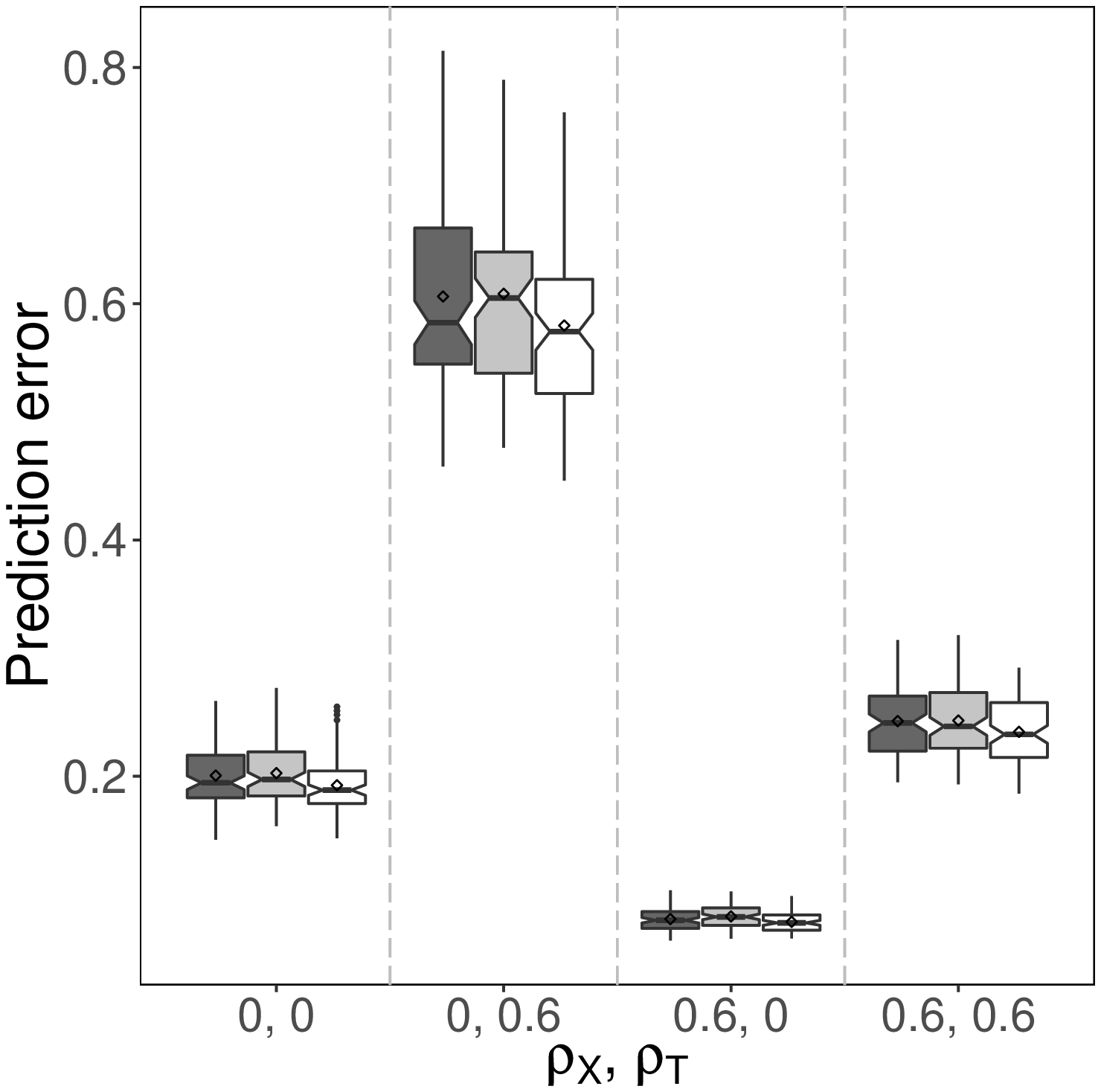}
    }\hspace{0.5cm}%
    \subcaptionbox{$n = 100, p = 100$}{
        \includegraphics[ width = 5.5cm, angle = -90]{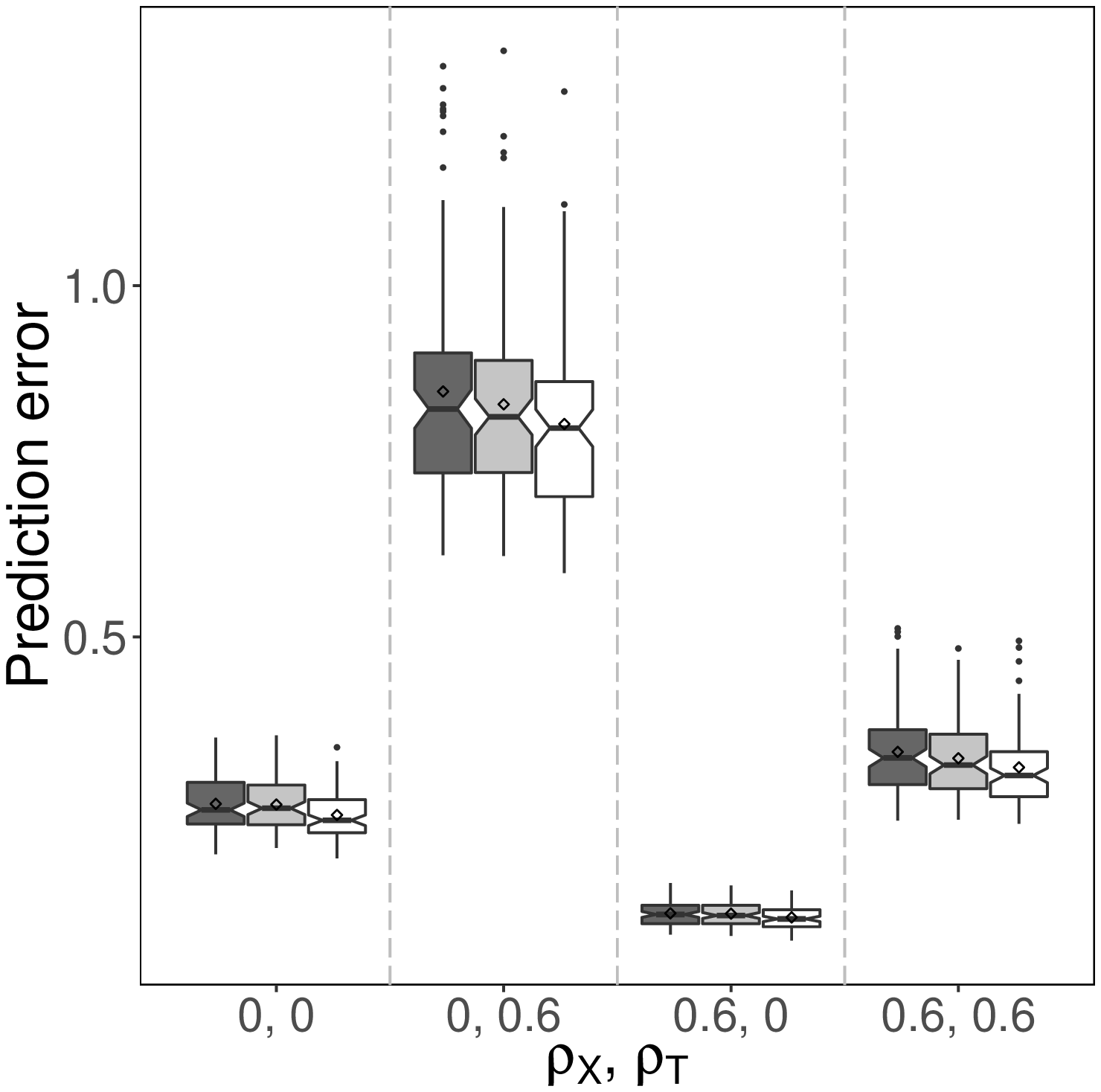}
    }\hspace{0.5cm}%
    \subcaptionbox{$n = 100, p = 200$}{
        \includegraphics[ width = 5.5cm, angle = -90]{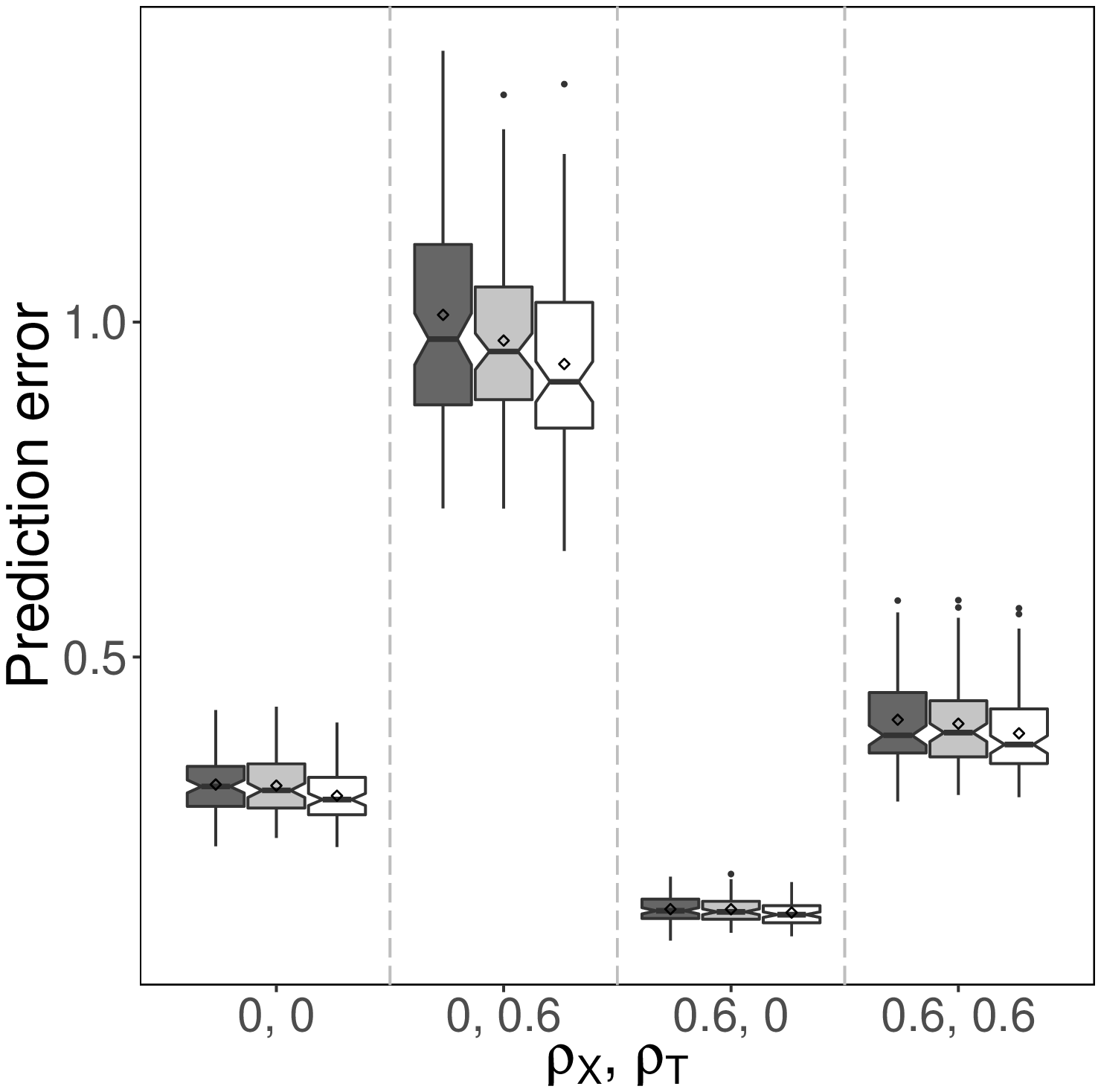}
    }
    \caption{Boxplots of prediction errors for various simulation settings with $\mbox{SNR} = 4$. The dark grey, light grey and white colors correspond to three different estimation methods BGL, GL and CGL, respectively.
%100 replications and 2\% were truncated for both sides.
    } \label{fig:sim1}
\end{figure}

% <<< Simulation >>>

% >>> Application <<<

\section{Linking Microbiome Trajectories to Neurobehavioral Outcomes}\label{sec:app}

Recall that our main objective is to identify the microbiome markers that are predictive of later infant neurodevelopment as measured by NNNS. This predictive association, if proven true, can provide supporting evidence to the claim that the stressful early life experience of preterm infants is imprinting gut microbiome by the regulation of the brain-gut axis. We tackle the problem with the functional log-contrast regression model in \eqref{eq:fmodel}, in which the composite NSTRESS score serves as the response variable, the gut microbiome observed during the early postnatal period serves as the functional compositional predictors, and the infant characteristics listed in Table \ref{tab:dem} below serve as the time-invariate control variables. We apply the proposed CGL approach for model estimation and compositional component selection. The cubic spline basis is used, and the tuning of the degrees of freedom $k$ as well as the sparsity parameter $\lambda$ is done using cross validation.

Our approach is able to identify four bacteria categories at the order level that are associated with the neurobehavioral outcome of infant, after controlling for the effects of several infant characteristics. Before we discuss the selected microbiome markers, let's first focus on the effects of the control variables. Table \ref{tab:dem} shows the estimated coefficients of the control variables along with some descriptive statistics. It is seen that the neurobehavioral outcome is better (i.e., NSTRESS is small) for infants with larger birth weight, smaller SNAPE-II score and more mother's breast milk for feeding. Regarding the delivery of infant, vaginal delivery and the absence of premature rupture of membranes are associated with better neurobehavioral development. These interesting and intuitive results are consistent with existing literature \citep{Neu2011CesareanVV, Feldman2003}. The analysis also shows that female infants tend to perform slightly better than male, after accounting for other effects. %of other variables including the integrated effects of the gut microbiome.

%coefficient estimation for PNA[5,28] mapped to [0,1]
\begin{table}[htp!]
    \centering
    \caption{Descriptive statistics of infant characteristics and their estimated coefficients from fitting the sparse functional log-contrast regression. Values of estimated coefficient are multiplied by 100.}\label{tab:dem}
    \begin{tabular}{lcc}
        %\toprule
                \hline
        Numerical variable & Mean (sd) & Estimated coefficient \\
        \hline
        Birth weight (in gram) & 1451.7 (479.3) & $-0.003$\\
        SNAPE-II               & 9.3 (10.6)     & $0.122$\\
        \%MBM                  & 61.8 (29.9)    & $-9.79$\\
                \cmidrule(lr){1-3}
        Binary variable & Percentage of ones &  \\
        %\cmidrule(lr){2-3}
        Gender (female = 1)        & 50.0\%        & $-0.065$\\
        PROM (yes = 1)             & 44.1\%         & $3.11$\\
        Delivery type (vaginal =1) & 35.3\%         & $-5.43$\\
        %\bottomrule

                \hline
          %22,19,17
    \end{tabular}
\end{table}

The estimated functional effects of the four selected bacteria categories are shown in the four panels of Figure \ref{fig:phyla}, respectively. In each panel, the lower part shows the estimated functional effects of a category over time (between 5 and 28 days of postnatal age), and the upper part attempts to show directly from raw data how this category changes over time for infants with high, medium, or low ``adjusted'' NSTRESS score, obtained by subtracting the estimated effects of the control variables and other selected bacteria categories from the observed NSTRESS scores.
Specifically, we construct smoothed curves of log-compositions of each selected category for three clusters of infants (using locally weighted scatterplot smoothing). 
For each category, the clusters are based on the percentiles of its ``adjusted'' NSTRESS score. %partial residuals, obtained by subtracting the estimated effects of the control variables and other selected bacteria categories from the observed NSTRESS scores. 
The curve with its 90\% confidence band is shown in red for the high group, i.e., infants with the upper one third of the adjusted scores, in blue for the medium group, i.e., infants with the middle one third of the adjusted scores, and in green for the low group, i.e., infants with the lower one third of the adjusted scores. As an example, for category 1, the red curve increases in the beginning to be above the other two curves and then becomes mostly below them in the later stage. This suggests that the time-varying effect of category 1 on the NSTRESS score is first positive and then negative, which is clearly reflected by the estimated functional effects. Similarly for the other three selected categories, the patterns of the estimated effects agree well with those of the observed data. This verifies visually that our proposed model and the estimation approach yield sensible results.

%%%Leaveone result
\begin{figure}[htp!]
    \captionsetup[subfigure]{singlelinecheck=true}
    \centering
    \subcaptionbox[singlelinecheck=true]{Category 1\label{fig:3_1}}{
        \includegraphics[ width = 5.5cm, height = 5cm]{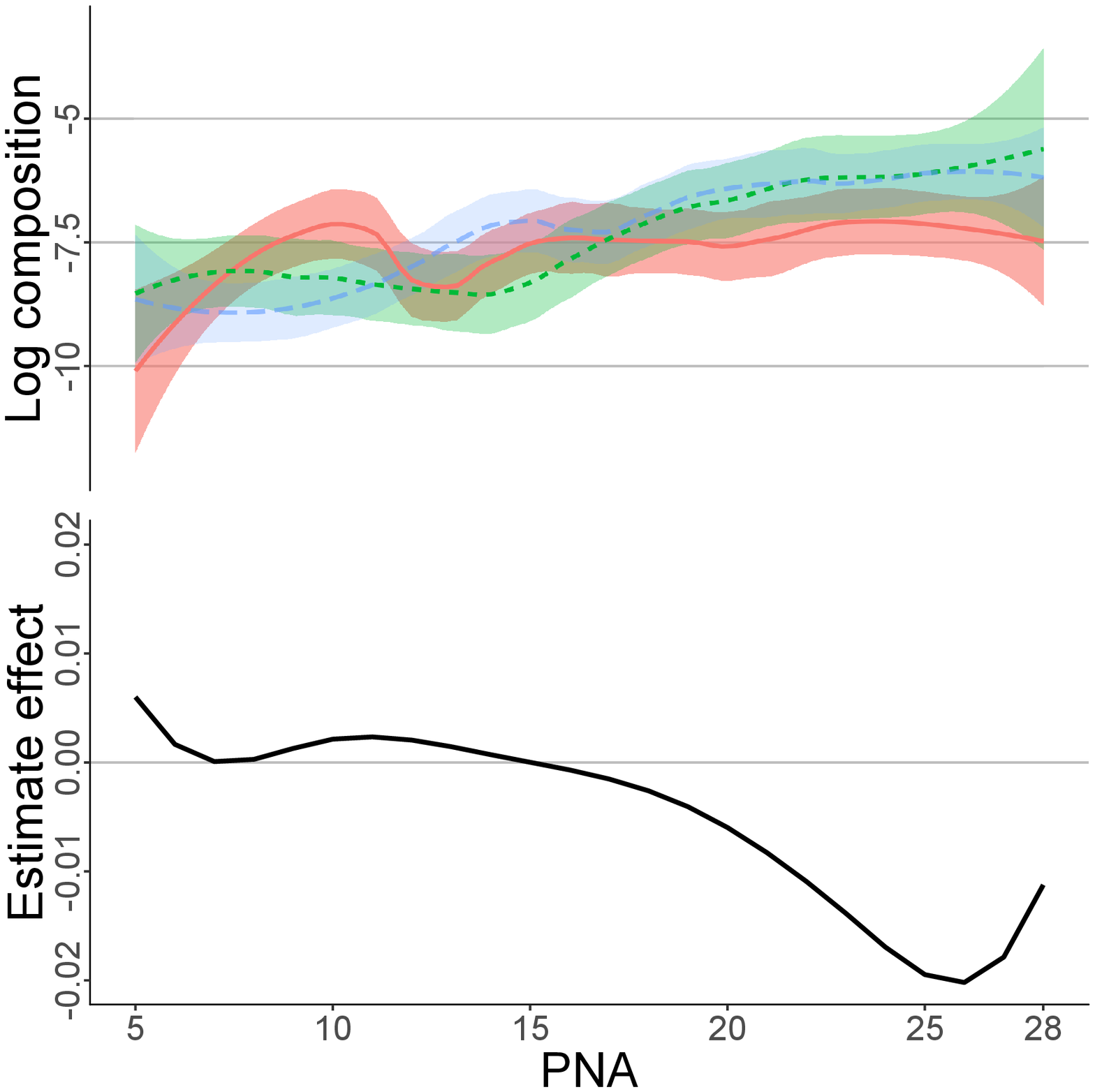}
        %\makebox[0pt][l]{%
        %   \hspace{0.04\columnwidth}
    }\hspace{0.5cm}%
    \subcaptionbox{Category 9\label{fig:3_9}}{
        \includegraphics[ width = 5.5cm, height = 5cm]{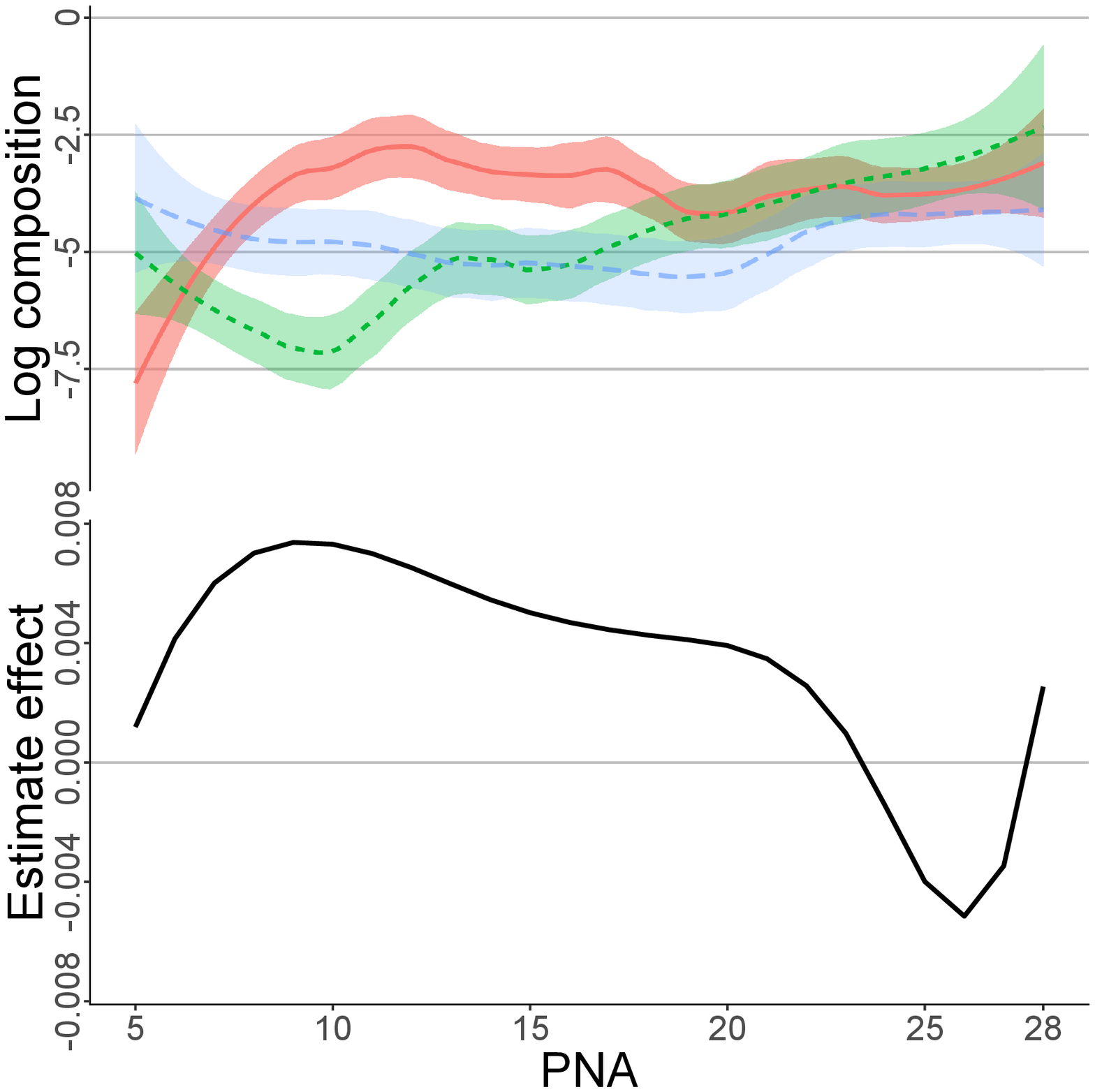}
    }\hspace{0.5cm}%
    \subcaptionbox{Category 10\label{fig:3_10}}{
        \includegraphics[ width = 5.5cm, height = 5cm]{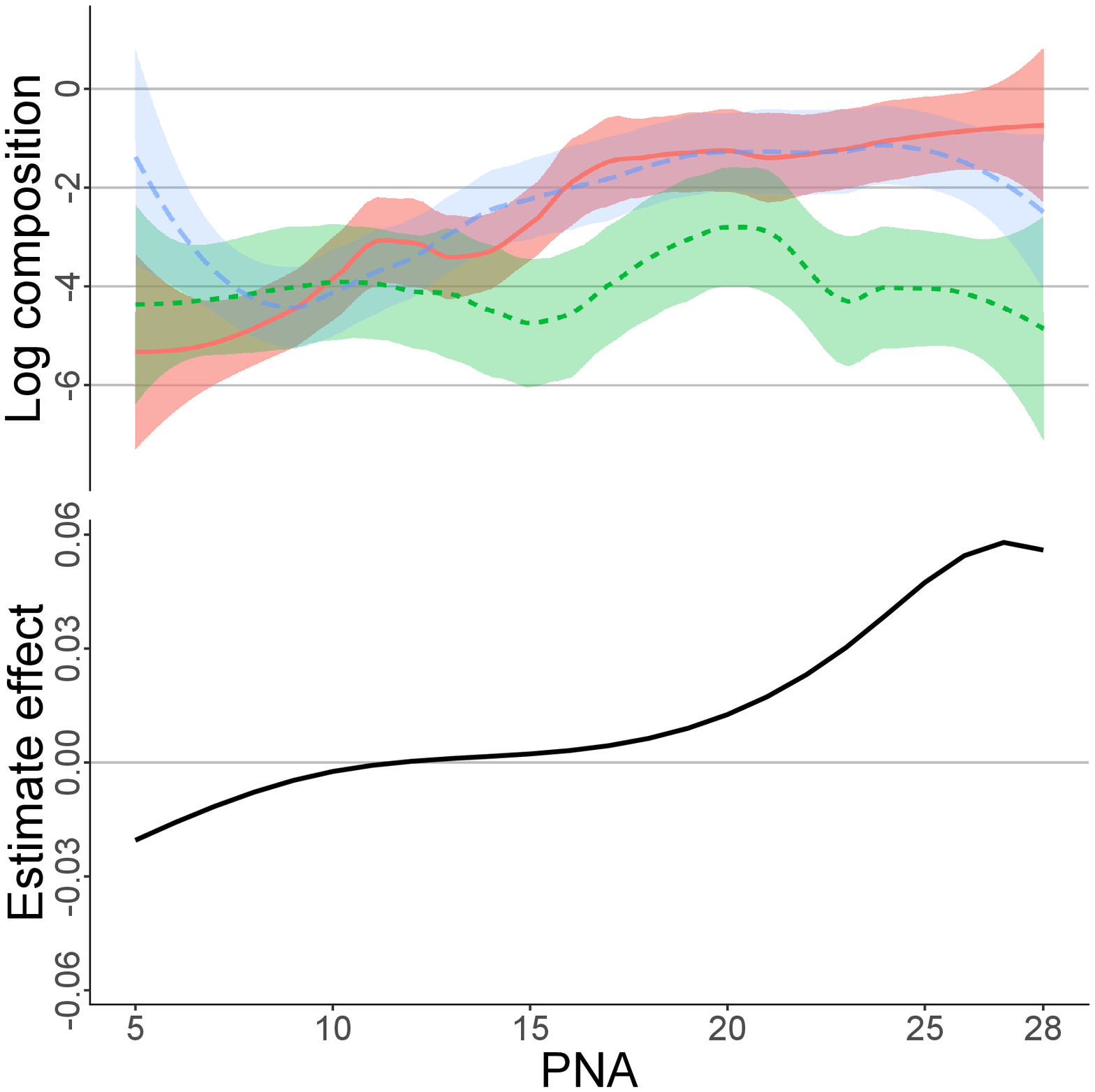}
    }\hspace{0.5cm}%
    \subcaptionbox{Category 19\label{fig:3_19}}{
        \includegraphics[ width = 5.5cm, height = 5cm]{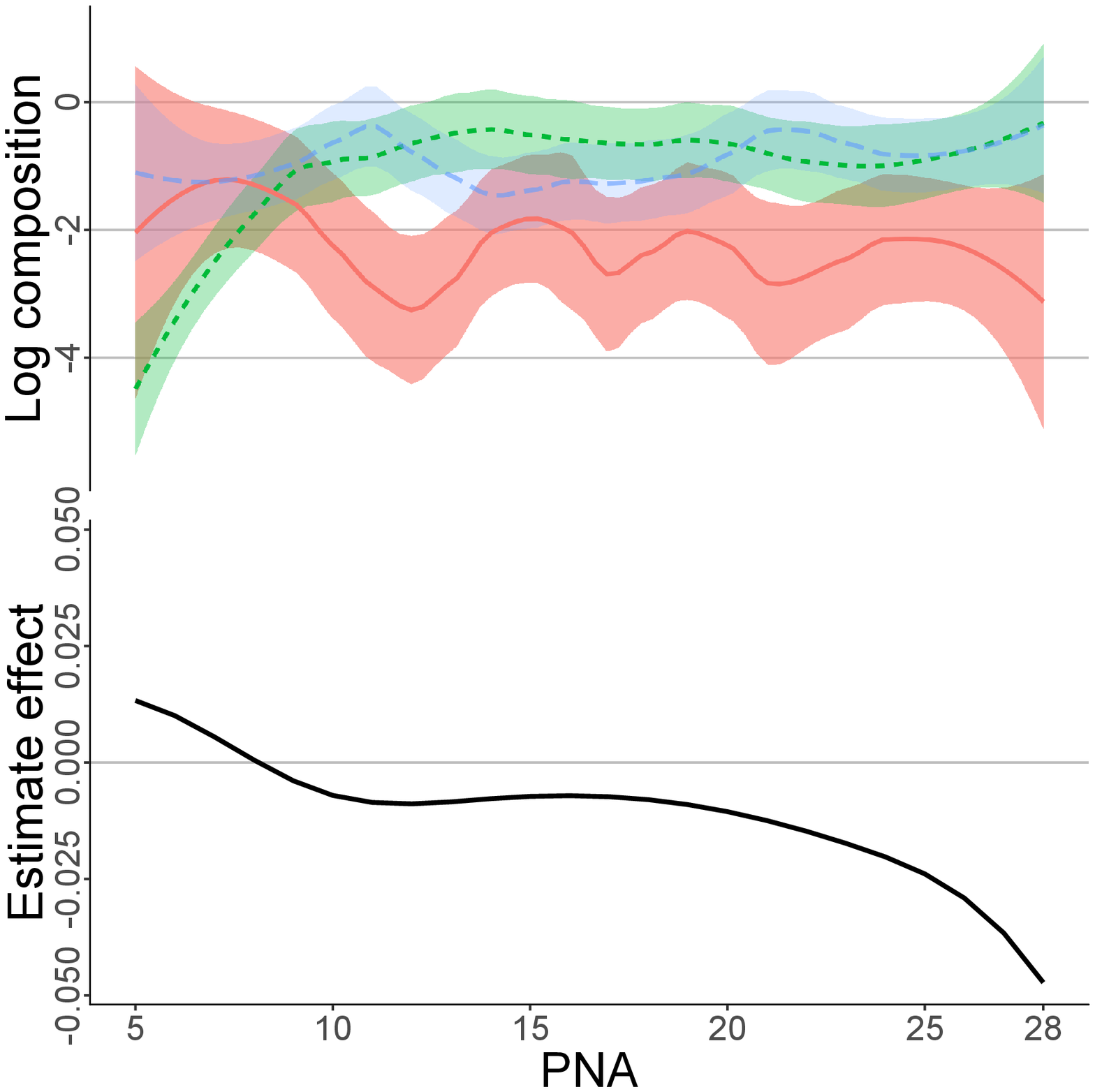}
    }
    \caption{Estimated effects of the four selected bacteria categories at the order level over infant's postnatal age (PNA) of 5 to 28 days. In each sub-graph, the upper panel shows how this category changes over time for three clusters of infants. 
For each category, the clusters are based on the percentiles of its partial residuals, obtained by subtracting the estimated effects of the control variables and other selected bacteria categories from the observed NSTRESS scores. The curve with its 90\% confidence band is shown in red for the high group, in blue for the medium group, and in green for the low group.
    }
    \label{fig:phyla}
\end{figure}

To access the stability of the results, we have generated 100 bootstrap samples and used the same cross validation procedure to select the best models. The results are show in Figure \ref{fig:boot}. The signs of the coefficients of the control variables are quite stable, except for the gender and SNAPE-II variables; this shows that these two variables may not have much effect on the outcome when conditioning on other terms in the model. For each control variable, the sign with the higher proportion among its 100 bootstrap estimates agrees with that of the estimate from fitting the original data, except for the gender. Furthermore, the top four categories with the highest proportions of being selected in bootstrap coincide with the categories selected from fitting the original data. Categories 10 and 19 are selected about 90\% of the times, while 9 and 1 are selected more than 70\% and 60\% of the times, respectively.

%%%selection proportion
\begin{figure}[htp!]
    \captionsetup[subfigure]{singlelinecheck=false}
    \centering
    \begin{subfigure}{0.45\linewidth}
        \caption{}
        \includegraphics[width=1\linewidth%keepaspectratio
        ]{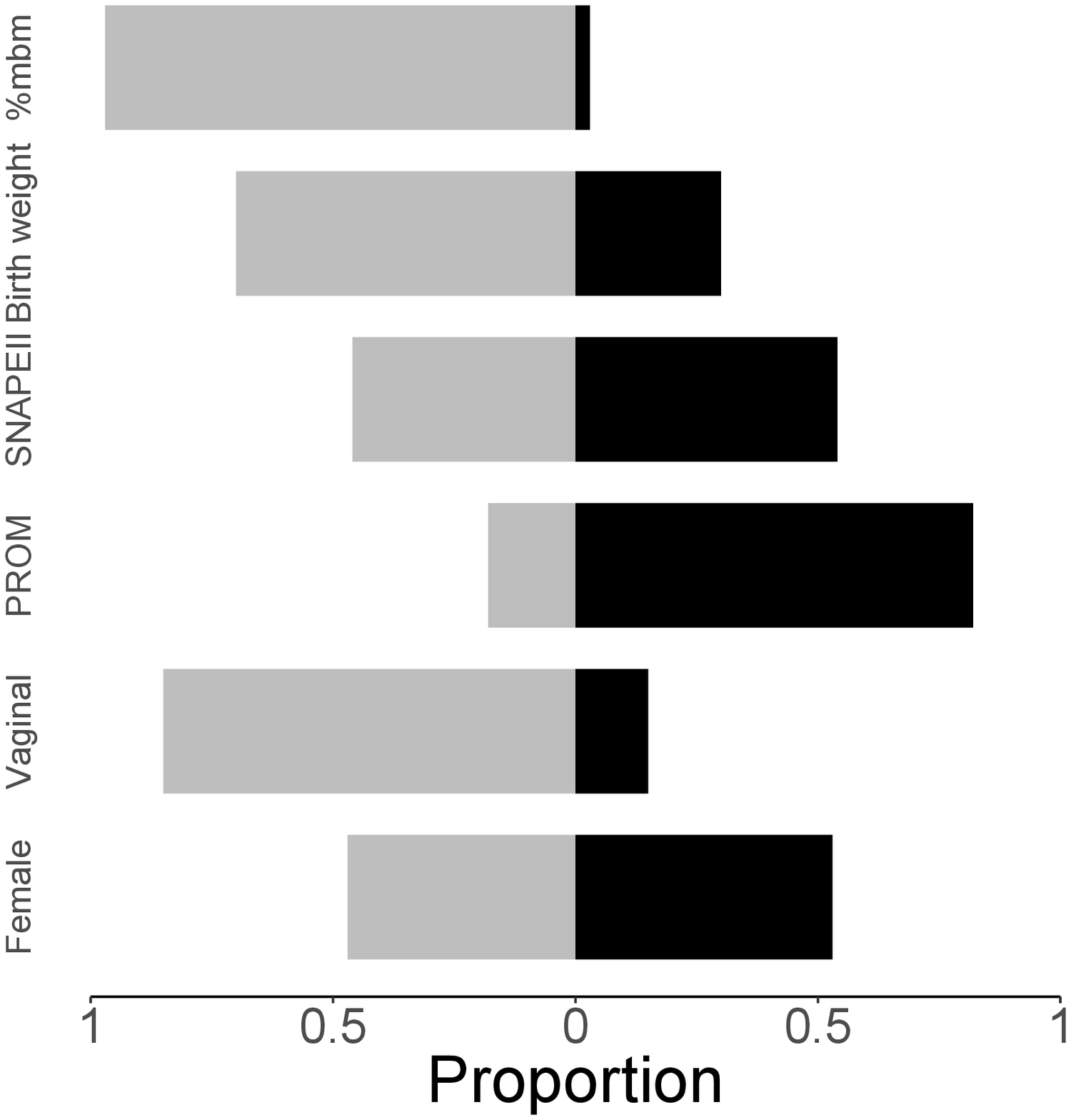}
    \end{subfigure}
%\hspace{0.2\textwidth}
    \begin{subfigure}{0.45\textwidth}
        \caption{}
        \includegraphics[width=1\linewidth%keepaspectratio
        ]{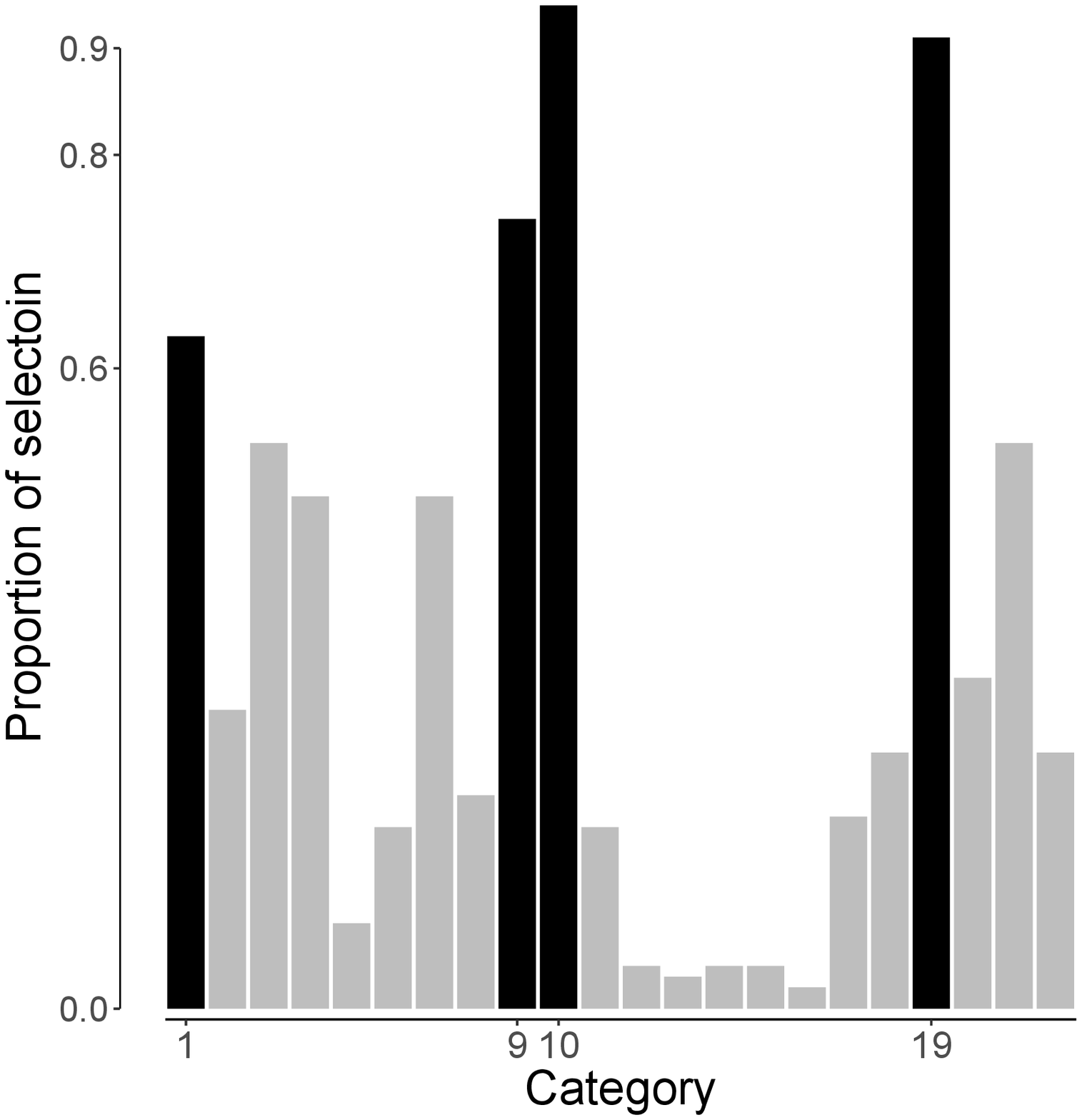}
    \end{subfigure}
\hspace{0.2\textwidth}
    \caption[]{Selection results from 100 bootstrap samples. (a) Proportions of the signs of the estimated coefficients of the control variables. Proportions of positive signs are shown as black blocks to the right, and those of negative signs are shown as light gray blocks to the left. (b) Proportions of selecting the 22 bacteria categories at the order level. The bars of the four selected categories from fitting the original data are colored in black. }
    \label{fig:boot}
\end{figure}

Category 10 consists of Clostridiales, which are an order of bacteria belonging to the phylum Firmicutes. Studies showed that infants fed with mother's milk had significantly higher abundance in Clostridiales \citep{Cong2014}. Clostridiales are generally regarded as hallmarks of a healthy gut; it can be a sign of infection when their subtypes such as Eubacteria die off in the large intestine. Our results show that controlling for other effects in the model, the effect of Clostridiales on the stress score switches from negative to positive during the %first thirty
postnatal days from 5 to 28. Category 9 consists of Lactobacillales, or lactic acid bacteria (LAB), another order of bacteria belonging to the phylum Firmicutes. These bacteria are usually found in decomposing plants and milk products; they are considered beneficial and produce organic acids such as lactic acid from carbohydrates. %such as lactose or glucose. However, due to stresses, irregular eating habits, unbalanced diet, westernized foods and lifestyles and many environmental factors, good lactic acid bacteria decrease as you get older.
Our analysis shows that controlling for the other effects in the model, higher LAB proportions are associated with higher stress scores for a period of time during the early postnatal days. Both Clostridiales and LAB belong to phylum Firmicutes, which make up the largest portion of the human gut microbiome, and the abundance of Firmicutes has been shown to be associated with inflammation and obesity \citet{Clarke2012,Boulange2016}. Category 19 consists of Enterobacteriales, an order of gram-negative bacteria. They are responsible for various infections such as bacteremia, lower respiratory tract infections, skin infections, etc. Category 1 consists of other unclassified bacteria. The functional regression analysis presented here may lead to a better understanding of how the trajectories of gut microbiome during early postnatal stage impact neurobehavioral outcomes of infants through the gut-brain axis.

\begin{figure}[htp!]
    \captionsetup[subfigure]{singlelinecheck=true}
    \centering
    \subcaptionbox[singlelinecheck=true]{Genus 1\label{fig:4_1}}{
        \includegraphics[ width = 5.5cm, height = 5cm]{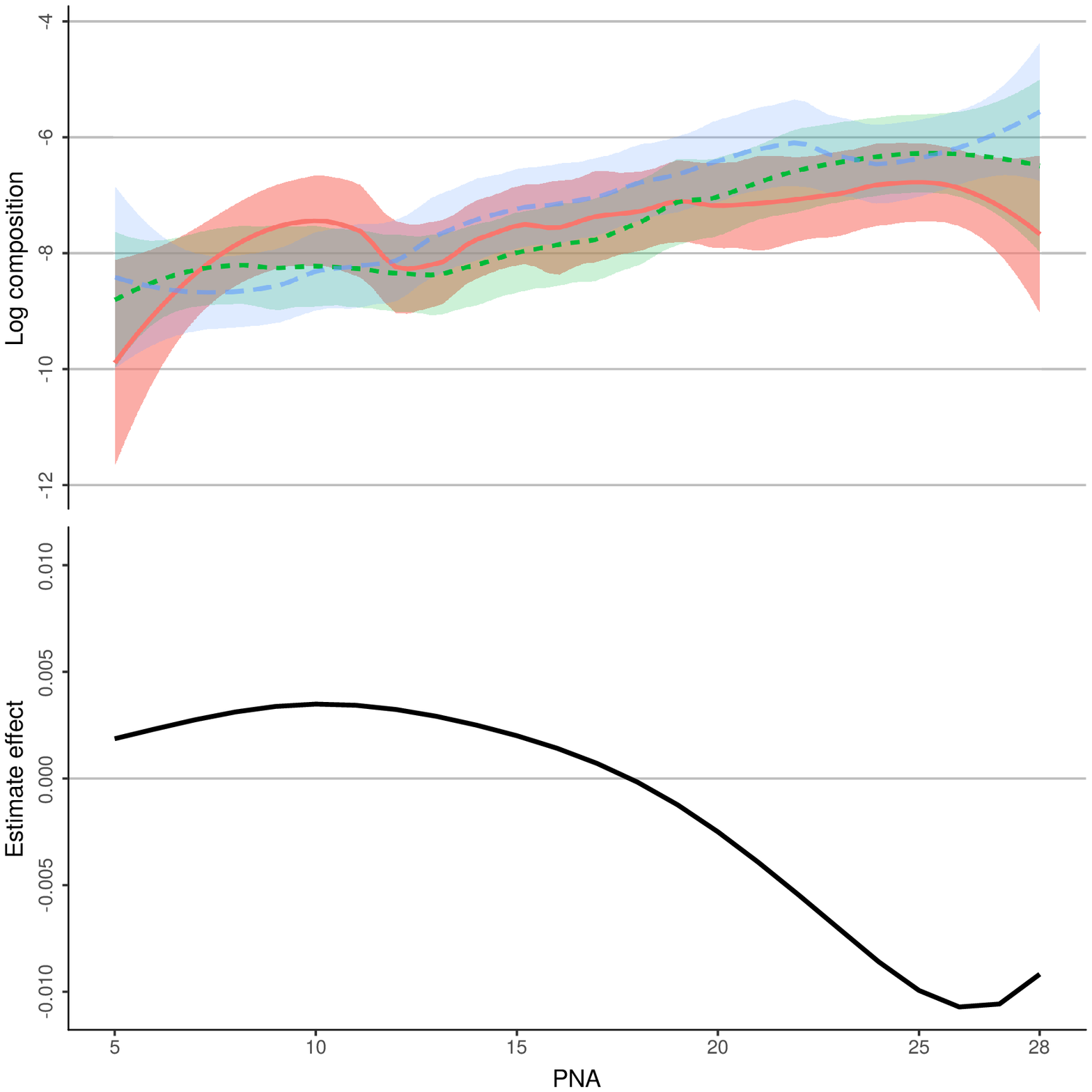}
        %\makebox[0pt][l]{%
        %   \hspace{0.04\columnwidth}
    }\hspace{0.5cm}%
    \subcaptionbox{Genus 20\label{fig:4_20}}{
        \includegraphics[ width = 5.5cm, height = 5cm]{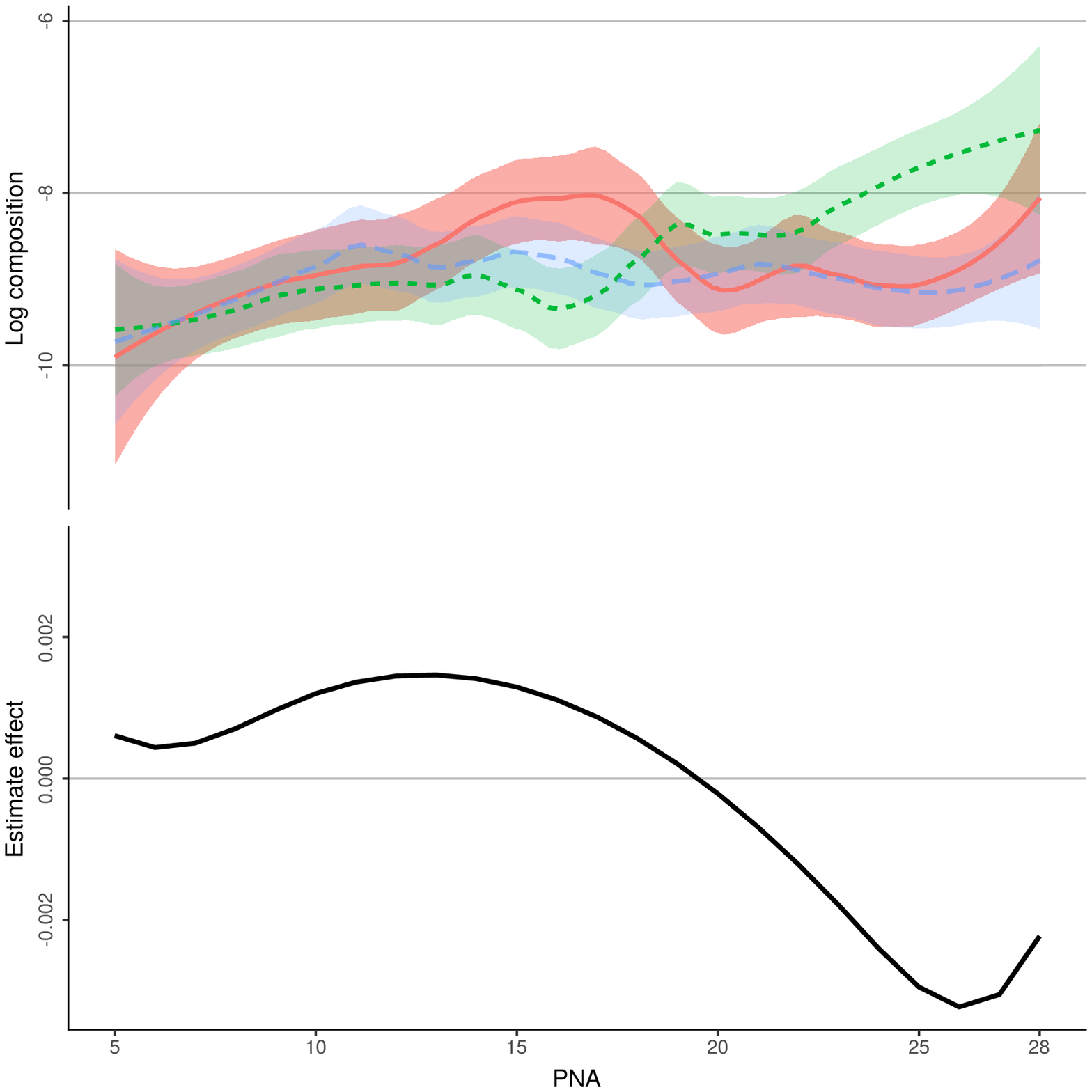}
    }\hspace{0.5cm}%
    \subcaptionbox{Genus 38\label{fig:4_38}}{
        \includegraphics[ width = 5.5cm, height = 5cm]{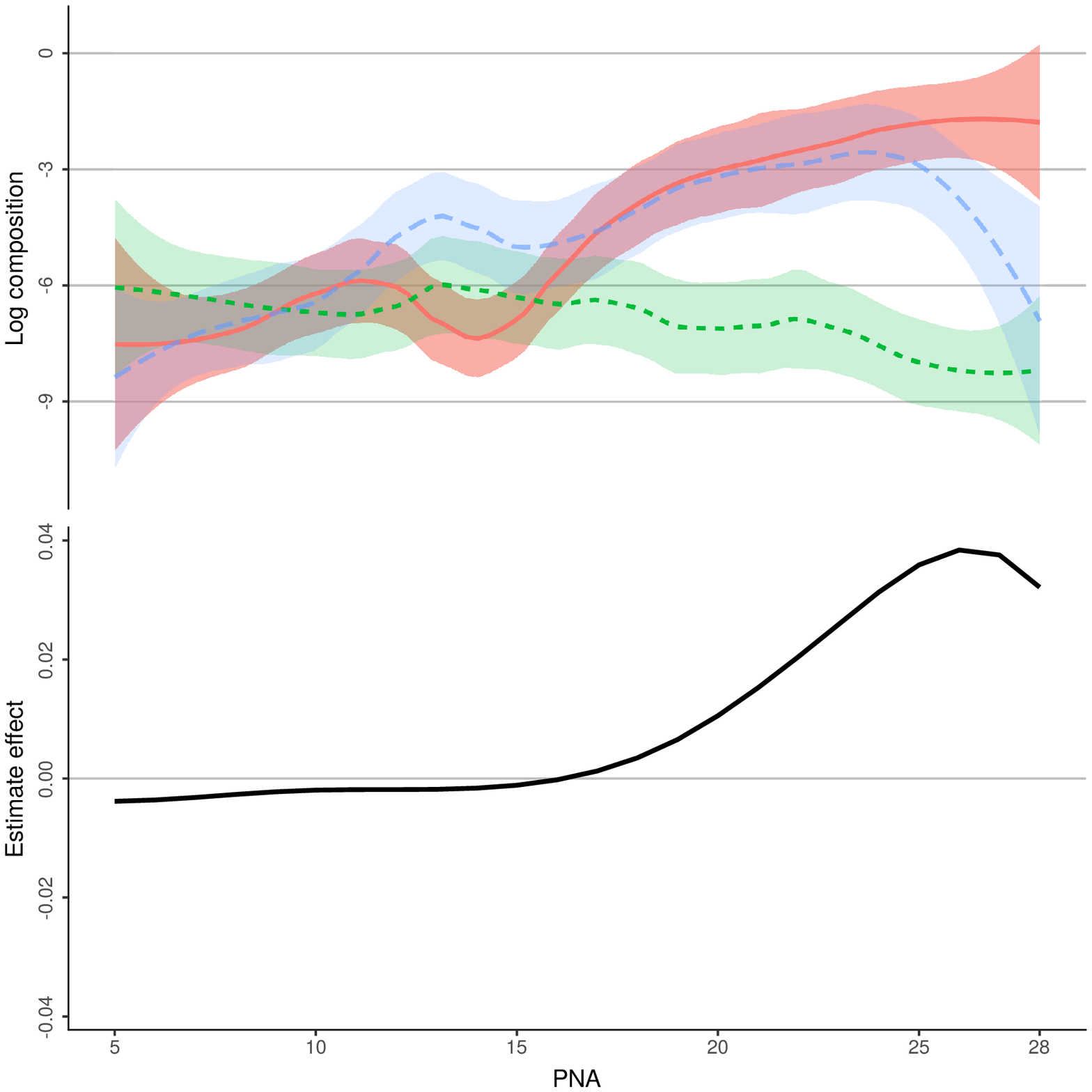}
    }\hspace{0.5cm}%
    \subcaptionbox{Genus 48\label{fig:4_48}}{
        \includegraphics[ width = 5.5cm, height = 5cm]{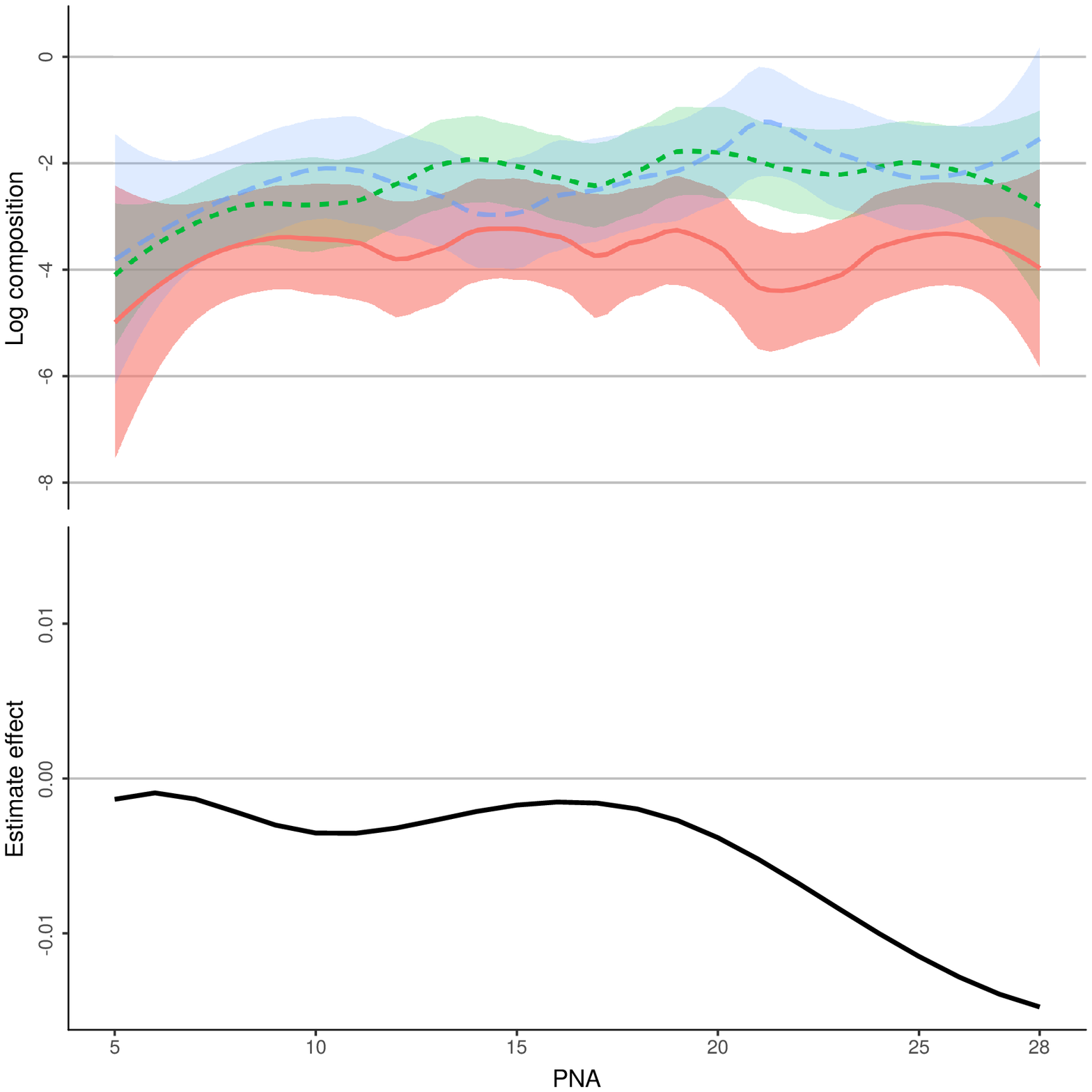}
    }
    \subcaptionbox{Genus 55\label{fig:4_5}}{
    \includegraphics[ width = 5.5cm, height = 5cm]{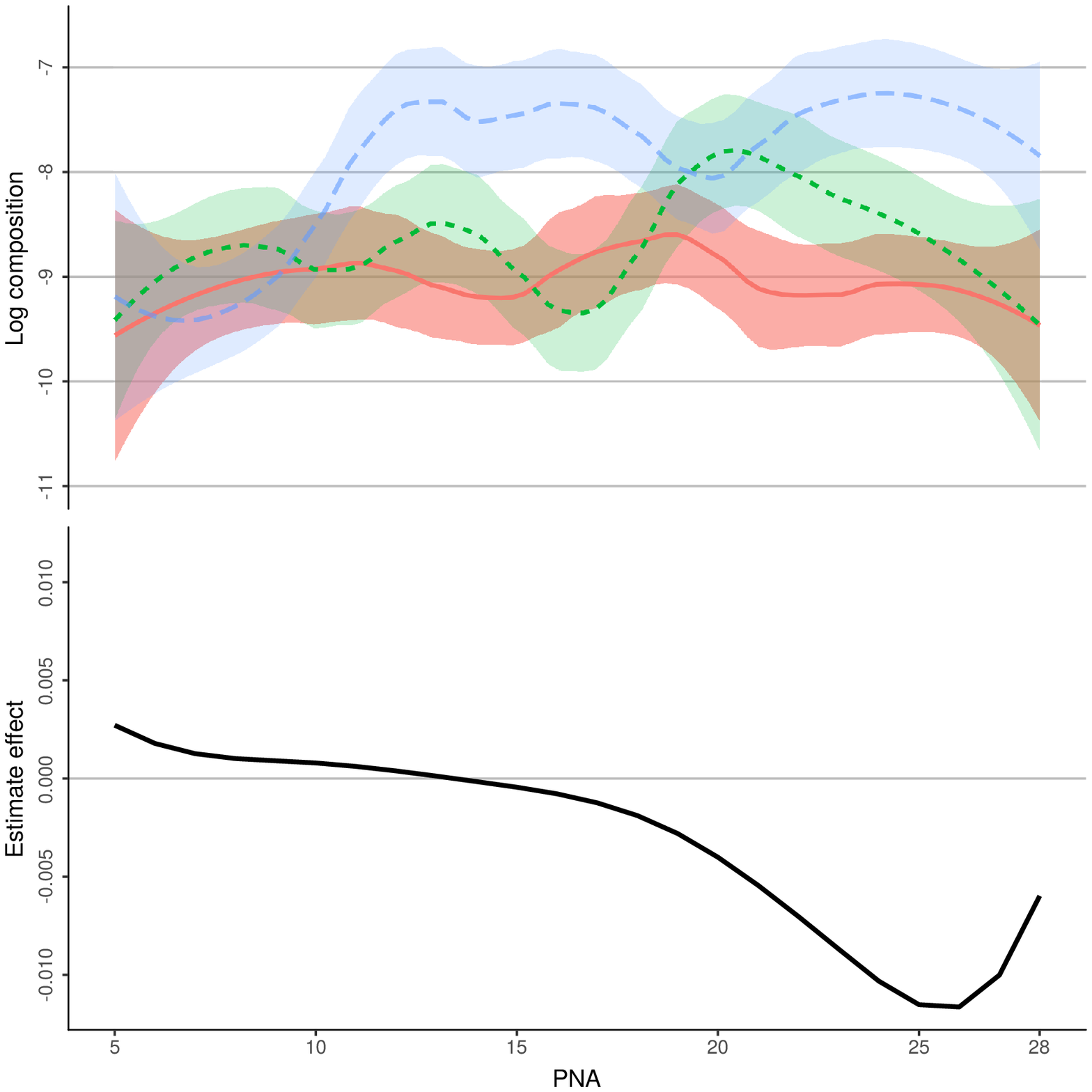}
    %\makebox[0pt][l]{%
    %   \hspace{0.04\columnwidth}
    }\hspace{6.1cm} \hfill %\hspace{6.63cm} \hfill
    \caption{Estimated effects of the five selected bacteria categories at the genus level over infant's postnatal age (PNA) of 5 to 28 days. The layout is the same as in Figure \ref{fig:phyla}.
    }
    \label{fig:genus}
\end{figure}

We also repeat the analysis on a lower level taxon, i.e., the genus level. Five out of $p=62$ genera are selected, and their estimated functional effects are shown in the five panels of Figure \ref{fig:genus}. The tendency of the estimated effects adequately reflects those of the observed data and the results are consistent with previous study on the order level. In particular, the five selected genera all belong to the four selected order categories; see Table \ref{tbl:sel}. Genus 38 comprises genus Veillonella, belonging to the order Clostridiales. Veillonella have been implicated as pathogens; they are often associated with oral, central nervous system and various soft tissue infections. Our results show that controlling for the other effect in the model, the effect of Veillonella on the stress score works similarly to that of Clostridiales, switching from negative to positive. Genus 20 consists of Enterococcus, which is an large genus of bacteria belonging to the order LAB. In humans, E. faecalis and E. faecium are the most abundant species of this genus found in fecal content, comprising up to 1\% of the adult intestinal microbiota. Although utilization of Enterococci as probiotics has been under controversial discussion, enterococcal strains such as E. faecium SF68 and E. faecalis Symbio-flor have been marketed as probiotics for decades without incidence and with very few reported adverse events \citep{Franz2011}. On the other aspect, Enterococci is also important nosocomial pathogens that cause bacteraemia, endocarditis and other infections. Same as LAB, controlling for the other effect in the model, our results show that higher Enterococcus proportions are associated with higher stress scores for a period of time during the early postnatal days. Genus 55 is Shigella, belonging to the order Enterobacteriales. Shigella is considered as pathogen causing shigellosis. The main sign of shigella infection is diarrhea, which often is bloody. However, shigellosis rarely affects infants during the first month of life. Even in highly endemic areas neonatal shigellosis is exceedingly uncommon \citep{Haltalin1967}. Our analysis shows that controlling for the other effects in the model, the effect of Shigella changes from positive to negative during early postnatal days. Genus 48 consists of other unclassified genera of bacteria that belongs to the order Enterobacteriales. Genus 1 consists of other unclassified bacteria.

\begin{table}[htp!]
\tiny
    \centering
    \caption{Comparison of selection of microbiome markers between order level and genus level.} \label{tbl:sel}
    \begin{tabular}{ll}
        \toprule
         \multicolumn{1}{c}{Order level}     
                     & \multicolumn{1}{c}{Genus level} \\
        \toprule          
        1: \bftab Others       
                     & 1: \bftab Others \\
        \midrule          
        \begin{tabular}[c]{@{}l@{}} 9: \bftab Lactobacillales\\ 
             %considered beneficial;\\ 
             Produce organic acids such as lactic acid from  \\
             carbohydrates.\\
             \\
            \end{tabular} 
                     & \begin{tabular}[c]{@{}l@{}} 20: \bftab Enterococcus  \\
                         It's used as probiotics in humans; \\
                         It's considered as pathogens that cause bacteraemia,   \\
                         endocarditis and other infections.             \\
                       \end{tabular} \\
        \midrule
        \begin{tabular}[c]{@{}l@{}} 10: \bftab Clostridiales \\
            It's generally regarded as hallmarks of a healthy gut; \\ 
            It's a sign of infection when their subtypes such as   \\
            Eubacteria die off in the large intestine.
        \end{tabular}                                                                                                                                       &   \begin{tabular}[c]{@{}l@{}} 38: \bftab Veillonella    \\
                         It's implicated as pathogens;                  \\
                         It's associated with oral infections and various soft                   \\
                         tissue infections.                  \\
                        \end{tabular}   \\ 
        \midrule
    
            \begin{tabular}[c]{@{}l@{}} 19: \bftab Enterobacteriales \\
                % gram-negative;\\
                It's responsible for various infections such as  \\
                bacteremia, lower respiratory tract infections, skin \\
                infections, etc.\\
        \end{tabular}
                                                                                                                                                                                      
        & \begin{tabular}[c]{@{}l@{}} 55: \bftab Shigella       \\
            It's considered as pathogen causing shigellosis;\\
            Shigellosis is exceedingly uncommon for infants \\
            during the first month of life\\
            48: \bftab Others\\
        \end{tabular} \\
        \bottomrule
             
    \end{tabular}
\end{table}

% <<< Application >>>

% >>> Discussion <<<
\section{Discussion}\label{sec:dis}

%Therefore, a functional regression analysis presented here may lead to a better understanding of how the trajectories of gut microbiota during early postnatal stage impact neurobehavioral outcomes of infants through the gut-brain axis. 

We have attempted a functional log-contrast regression approach to identify trajectories of gut microbiome components during early postnatal stage that are associated with later neurobehavioral outcomes of pre-term infants. There are several directions for future research to address the limitations of the current work. 
The results on order and genus levels only give a general idea of how the microbial communities effect health outcomes, to fully decipher their roles further analysis on species level or even operational taxonomic unit (OTU) is needed.
The data analysis can benefit from extending the model to consider potential interactions between the control variables and the gut microbiome, as it is possible, for example, that the effects of certain microbiome markers differ for male and female infants. Extensions to binary outcome or mixture model setup are interesting and could be widely applicable; indeed, it is of interest to see whether there exists a subgroup structure among the preterm infants. To take into account the uncertainty due to discrete observations, it is urgent to develop smoothing or dimension reduction methods such as multivariate functional principal component analysis for compositional data observed discretely over time. A joint modeling approach of both the regression and the functional compositions themselves may also be fruitful.

%In this work, we essentially assume that the integrals involving the trajectories of the compositional data can be well approximated via simple interpolation. Consequently, our theoretical analysis has focused on the linearly constrained group lasso models.

%The key may be to expand the functional compositions by a multivariate basis that preserves the simplex structure under certain type of regularization. %Finally, for the constrained lasso estimator, statistical inference can be developed through the de-biased procedure. 

%The functional log-contrast model presented in this paper provides

%may lead to a better understanding of how the trajectories of gut microbiota during early postnatal stage impact neurobehavioral outcomes of infants through the gut-brain axis. 

% <<< Discussion >>>

% \section{Software}

% Software in the form of an R package is available from the corresponding author (\href{https://kun-chen.uconn.edu/code/}{https://kun-chen.uconn.edu/code/}).

% \section*{Supplementary Materials}

% Supplementary materials are available at \href{http://biostatistics.oxfordjournals.org}{http://biostatistics.oxfordjournals.org}. 

\section*{Acknowledgments}

Cong's research is supported by by U.S. National Institutes of Health grants NINR K23NR014674 and R01NR016928. Li's research is supported by U.S. National Institutes of Health grant NIDCR R03DE027773. Chen's research is partially supported by U.S. National Science Foundation grants DMS-1613295 and IIS-1718798. The authors thank the medical and nursing staff in the NICUs of Connecticut Children’s Medical Center at Hartford and Farmington, CT for their support and assistance. 

%\bibliographystyle{rss}
% \bibliographystyle{rss}
% \bibliography{compositional-bibtex}{}

\begin{thebibliography}{31}
\expandafter\ifx\csname natexlab\endcsname\relax\def\natexlab#1{#1}\fi
\expandafter\ifx\csname url\endcsname\relax
  \def\url#1{\texttt{#1}}\fi
\expandafter\ifx\csname urlprefix\endcsname\relax\def\urlprefix{URL }\fi

\bibitem[{Aitchison(1982)}]{Aitchison1982}
Aitchison, J. (1982) The statistical analysis of compositional data.
\newblock \emph{Journal of the Royal Statistical Society: Series B},
  \textbf{44}, 139--177.

\bibitem[{Aitchison(2003)}]{AITCHISON2003}
Aitchison, J. (2003) \emph{The Statistical Analysis of Compositional Data.}
\newblock New Jersey, US: Blackburn Press.

\bibitem[{Aitchison and Bacon-Shone(1984)}]{Aitchison1984}
Aitchison, J. and Bacon-Shone, J. (1984) Log-contrast models for experiments
  with mixtures.
\newblock \emph{Biometrika}, \textbf{71}, 323--330.

\bibitem[{Aitchison and J.~Egozcue(2005)}]{Aitchison2005}
Aitchison, J. and J.~Egozcue, J. (2005) Compositional data analysis: Where are
  we and where should we be heading?
\newblock \emph{Mathematical Geology}, \textbf{37}, 829--850.

\bibitem[{Bomar \emph{et~al.}(2011)Bomar, Maltz, Colston and Graf}]{Bomar2011}
Bomar, L., Maltz, M., Colston, S. and Graf, J. (2011) Directed culturing of
  microorganisms using metatranscriptomics.
\newblock \textbf{2}, e00012--00011.

\bibitem[{Boulang{\'e} \emph{et~al.}(2016)Boulang{\'e}, Neves, Chilloux,
  Nicholson and Dumas}]{Boulange2016}
Boulang{\'e}, C.~L., Neves, A.~L., Chilloux, J., Nicholson, J.~K. and Dumas,
  M.-E. (2016) Impact of the gut microbiota on inflammation, obesity, and
  metabolic disease.
\newblock \emph{Genome Medicine}, \textbf{8}, 42.

\bibitem[{Boyd \emph{et~al.}(2011)Boyd, Parikh, Chu, Peleato and
  Eckstein}]{Boyd2011}
Boyd, S., Parikh, N., Chu, E., Peleato, B. and Eckstein, J. (2011)
  \emph{{Distributed Optimization and Statistical Learning via the Alternating
  Direction Method of Multipliers}}, vol.~3.

\bibitem[{Clarke \emph{et~al.}(2012)Clarke, Murphy, Nilaweera, Ross, Shanahan,
  O'Toole and Cotter}]{Clarke2012}
Clarke, S.~F., Murphy, E.~F., Nilaweera, K., Ross, P.~R., Shanahan, F.,
  O'Toole, P.~W. and Cotter, P.~D. (2012) The gut microbiota and its
  relationship to diet and obesity: New insights.
\newblock \emph{Gut Microbes}, \textbf{3}, 186--202.

\bibitem[{Cong \emph{et~al.}(2017)Cong, Judge, Xu and Diallo}]{xiaomei2017}
Cong, X., Judge, M., Xu, W. and Diallo, A. (2017) Influence of feeding type on
  gut microbiome development in hospitalized preterm infants.
\newblock \emph{Nursing Research}, \textbf{66}, 123--133.

\bibitem[{Cong \emph{et~al.}(2016)Cong, Xu, Janton, Henderson, Matson, McGrath,
  Maas and Graf}]{Cong2014}
Cong, X., Xu, W., Janton, S., Henderson, W.~A., Matson, A., McGrath, J.~M.,
  Maas, K. and Graf, J. (2016) Gut microbiome developmental patterns in early
  life of preterm infants: Impacts of feeding and gender.
\newblock \emph{PLOS ONE}, \textbf{11}, 1--19.

\bibitem[{Dinan and Cryan(2012)}]{Dinan2012}
Dinan, T. and Cryan, J. (2012) Regulation of the stress response by the gut
  microbiota: implications for psychoneuroendocrinology.
\newblock \emph{Psychoneuroendocrinology}, \textbf{37}, 1369--1378.

\bibitem[{Fan and Tang(2013)}]{FanTang2013}
Fan, Y. and Tang, C.~Y. (2013) Tuning parameter selection in high dimensional
  penalized likelihood.
\newblock \emph{Journal of the Royal Statistical Society: Series B},
  \textbf{75}, 531--552.

\bibitem[{Fanaroff \emph{et~al.}(2003)Fanaroff, Hack and MC}]{Fanaroff2003}
Fanaroff, A., Hack, M. and MC, W. (2003) The nichd neonatal research network:
  changes in practice and outcomes during the first 15 years.
\newblock \emph{Seminars in Perinatology}, \textbf{27}, 281--287.

\bibitem[{Feldman and Eidelman(2003)}]{Feldman2003}
Feldman, R. and Eidelman, A.~I. (2003) Direct and indirect effects of breast
  milk on the neurobehavioral and cognitive development of premature infants.
\newblock \emph{Developmental Psychobiology}, \textbf{43}, 109--119.

\bibitem[{Franz \emph{et~al.}(2011)Franz, Huch, Abriouel, Holzapfel and
  G{\'a}lvez}]{Franz2011}
Franz, C.~M., Huch, M., Abriouel, H., Holzapfel, W.~H. and G{\'a}lvez, A.~V.
  (2011) Enterococci as probiotics and their implications in food safety.
\newblock \emph{International journal of food microbiology}, \textbf{151 2},
  125--40.

\bibitem[{Haltalin(1967)}]{Haltalin1967}
Haltalin, K.~C. (1967) {Neonatal Shigellosis: Report of 16 Cases and Review of
  the Literature}.
\newblock \emph{JAMA Pediatrics}, \textbf{114}, 603--611.

\bibitem[{Huang \emph{et~al.}(2012)Huang, Breheny and Ma}]{huang2012}
Huang, J., Breheny, P. and Ma, S. (2012) A selective review of group selection
  in high dimensional models.
\newblock \emph{Statist. Sci.}, \textbf{27}, 481--499.

\bibitem[{Li(2015)}]{LiH2015}
Li, H. (2015) Microbiome, metagenomics, and high-dimensional compositional data
  analysis.
\newblock \emph{Annual Review of Statistics and Its Application}, \textbf{2},
  73--94.

\bibitem[{Lin \emph{et~al.}(2014)Lin, Shi, Feng and Li}]{Lin2014}
Lin, W., Shi, P., Feng, R. and Li, H. (2014) Variable selection in regression
  with compositional covariates.
\newblock \emph{Biometrika}, \textbf{101}, 785--797.

\bibitem[{Lounici \emph{et~al.}(2011)Lounici, Pontil, van~de Geer and
  Tsybakov}]{Lounici2011}
Lounici, K., Pontil, M., van~de Geer, S. and Tsybakov, A.~B. (2011) Oracle
  inequalities and optimal inference under group sparsity.
\newblock \emph{Ann. Statist.}, \textbf{39}, 2164--2204.

\bibitem[{Morris(2015)}]{Morris2015}
Morris, J.~S. (2015) Functional regression.
\newblock \emph{Annual Review of Statistics and Its Application}, \textbf{2},
  321--359.

\bibitem[{Mwaniki \emph{et~al.}(2012)Mwaniki, Atieno, Lawn and
  Newton}]{Mwaniki2012}
Mwaniki, M., Atieno, M., Lawn, J. and Newton, C. (2012) Long-term
  neurodevelopmental outcomes after intrauterine and neonatal insults: a
  systematic review.
\newblock \emph{Lancet}, \textbf{379}, 445--452.

\bibitem[{Neu and Rushing(2011)}]{Neu2011CesareanVV}
Neu, J. and Rushing, J.~M. (2011) Cesarean versus vaginal delivery: long-term
  infant outcomes and the hygiene hypothesis.
\newblock \emph{Clinics in perinatology}, \textbf{38 2}, 321--31.

\bibitem[{Ramsay and Silverman(2005)}]{Ramsay2005}
Ramsay, J.~O. and Silverman, B.~W. (2005) \emph{{Functional Data Analysis}}.
\newblock Springer Series in Statistics. Springer, 2nd edn.

\bibitem[{Shi \emph{et~al.}(2016)Shi, Zhang and Li}]{Shi2016}
Shi, P., Zhang, A. and Li, H. (2016) Regression analysis for microbiome
  compositional data.
\newblock \emph{Ann. Appl. Stat.}, \textbf{10}, 1019--1040.

\bibitem[{Stoll \emph{et~al.}(2010)Stoll, Hansen and Bell}]{Stoll2010}
Stoll, B., Hansen, N. and Bell, E. (2010) Neonatal outcomes of extremely
  preterm infants from the nichd neonatal research network.
\newblock \textbf{126}, 443--456.

\bibitem[{Stone(1974)}]{stone1974}
Stone, M. (1974) Cross-validation and multinomial prediction.
\newblock \emph{Biometrika}, \textbf{61}, 509--515.

\bibitem[{Tibshirani(1996)}]{tib1996}
Tibshirani, R.~J. (1996) Regression shrinkage and selection via the lasso.
\newblock \emph{Journal of the Royal Statistical Society: Series B},
  \textbf{58}, 267--288.

\bibitem[{Tsybakov(2008)}]{Tsybakov2008}
Tsybakov, A.~B. (2008) \emph{Introduction to Nonparametric Estimation}.
\newblock Springer Publishing Company, Incorporated, 1st edn.

\bibitem[{Wei and Huang(2010)}]{WeiHuang2010}
Wei, F. and Huang, J. (2010) Consistent group selection in high-dimensional
  linear regression.
\newblock \emph{Bernoulli}, \textbf{16}, 1369--1384.

\bibitem[{Yuan and Lin(2006)}]{yuan2006}
Yuan, M. and Lin, Y. (2006) {Model selection and estimation in regression with
  grouped variables}.
\newblock \emph{Journal of the Royal Statistical Society: Series B},
  \textbf{68}, 49--67.

\end{thebibliography}

\clearpage
\appendix

\section*{Supplemental Materials}

% >>> Appendix <<<
\section{Computational Algorithm}\label{sec:supp:comp}

%\subsection{Solving Constrained Group Lasso}
%For example, when the original $\Z$ matrix is standardized before fitting an unweighted constrained group lasso model, it is equivalent to set each $\W_j$ as a diagonal matrix with the standard deviations of the columns of $\Z_j$ on its diagonal. In our problem $\A_j$s are simply the identity matrix, but there could be circumstances that general forms of linear constraints are needed.

We consider an estimation criterion that is slightly more general than \eqref{eq:glasso} in the main paper,
\begin{align*}%\label{eq:glassow}
\min_{\beta_0,\bbeta_c,\bbeta}\left\{\frac{1}{2n}\| \y -\beta_0\1_n - \Z_c\bbeta_c - \Z\bbeta\|^2  + \lambda \sum_{j=1}^{p} \|\W_j\bbeta_j\|\right\}, \mbox{s.t.} \sum_{j=1}^{p}\A_j\bbeta_j=\b,
\end{align*}
where each $\W_j \in \mathbb{R}^{k\times k}$ is invertable, e.g., a diagonal matrix with positive diagonal elements, and the linear constraints, with choices of conformable $\A_j$s and $\b$, remain feasible, i.e., $\{\bbeta; \sum_{j=1}^{p}\A_j\bbeta_j=\b\} \neq \emptyset$. The problem is convex and can be solved by an augmented Lagrangian algorithm \citep{Boyd2011}.

% First, the augmented Lagrangian is constructed as
% \begin{align*}
% L_1(\beta_0,\bbeta_c,\bbeta;\bgamma,\mu)
% & =
% \frac{1}{2n}\| \y -\beta_0\1_n - \Z_c\bbeta_c - \Z\bbeta\|^2 \\
% %+ \lambda \sum_{j=1}^{p} \|\W_j\bbeta_j\| \\
% & +
% \bgamma \trans (\sum_{j=1}^{p}\A_j\bbeta_j - \b)
% + \frac{\mu}{2} \| \sum_{j=1}^{p}\A_j\bbeta_j-\b \|^2 + \lambda \sum_{j=1}^{p} \|\W_j\bbeta_j\|,
% \end{align*}
% where $\bgamma = (\gamma_1,\cdots,\gamma_k)\trans \in \mathbb{R}^k$ is the Lagrange multipliers, and $\mu>0$ is a prespecified penalty parameter. Define the scaled Lagrange multiplier $\balpha=\bgamma/\mu \in \mathbb{R}^k$.

%which differs from $L_1$ by only a constant $\mu\|\balpha\|_2^2/2$.

% The method of multipliers for problem (\ref {eq:glasso}) consists of the iterations
% \begin{align*}
% \big(\beta_0^\ell,\bbeta_c^\ell,\bbeta^\ell)
% & \leftarrow \min_{\beta_0,\bbeta_c,\bbeta} L_1(\beta_0,\bbeta_c,\bbeta,\bgamma^\ell,\mu^\ell), \\
% \bgamma^{\ell+1}
% & \leftarrow \bgamma^\ell + \mu \sum_{j=1}^{p}\bbeta_j^\ell.
% \end{align*}

%\begin{align*}
%& L_2^\ell (\beta_0,\bbeta_c,\bbeta)\\
%& = \frac{1}{2n}\| \y -\beta_0\1_n - \Z_c\bbeta_c - \Z\bbeta\|^2 + \mu^\ell \| \sum_{j=1}^{p}\bbeta_j + \balpha^\ell/2 \|^2 \\
%& = \frac{1}{2n} \left\{ \| \y -\beta_0\1_n - \Z_c\bbeta_c - \Z\bbeta\|^2 +
%\| \frac{\sqrt{2n\mu^\ell}}{2} \balpha^\ell + \sqrt{2n\mu^\ell} (\1_p \trans \otimes \I_k) \bbeta\|^2. \right\}
%\end{align*}

To derive the algorithm, we first construct the scaled augmented Lagrangian function
\begin{align*}
  L(\bbeta_{0c},\bbeta;\balpha,\mu)  =  & \frac{1}{2n}\| \y -\beta_0\1_n - \Z_c\bbeta_c - \Z\bbeta\|^2\\
 & + \frac{\mu}{2} \| \sum_{j=1}^{p}\A_j\bbeta_j -\b + \balpha \|^2
 + \lambda \sum_{j=1}^{p} \|\W_j\bbeta_j\|,
%- \frac{\mu}{2}\|\balpha\|_2^2,
\end{align*}
where $\mu>0$ is a prespecified penalty parameter, $\balpha \in \mathbb{R}^k$ is the scaled Lagrange multiplier, and $\bbeta_{0c} = (\beta_0\trans, \bbeta_c\trans)\trans$ collects the unpenalized coefficients.

The algorithm alternates between two steps, a primal step and a dual step, until convergence. Let $\ell = 0, 1, ...$ denote the iteration number. The primal step minimizes $L(\bbeta_{0c},\bbeta;\balpha,\mu)$ with respect to $(\bbeta_{0c}, \bbeta)$: $\big(\bbeta_{0c}^{\ell+1},\bbeta^{\ell+1}\big)\leftarrow
\min_{\bbeta_{0c},\bbeta} \{L (\bbeta_{0c},\bbeta;\balpha^{\ell},\mu)\}$.
%+ \lambda \sum_{j=1}^{p} \|\W_j\bbeta_j\| \}.
%\leftarrow \min_{\beta_0,\bbeta_c,\bbeta} \frac{1}{2n}\| \y -\beta_0\1_n - \Z_c\bbeta_c - \Z\bbeta\|^2  + \lambda \sum_{j=1}^{p} \|\bbeta_j\| \nonumber \\
%& \qquad+ \mu^\ell \| \sum_{j=1}^{p}\bbeta_j + \balpha^\ell/2 $\|^2, \\
% \balpha^{\ell+1}
% & \leftarrow (\balpha^\ell+\sum_{j=1}^{p}\bbeta_j^{\ell+1})\frac{\mu^\ell}{\mu^{\ell+1}}
%\balpha^{\ell+1}
%& \leftarrow \balpha^\ell+\sum_{j=1}^{p}\A_j\bbeta_j^{\ell+1} - \b
%\nonumber.
The problem is equivalent to a standard group lasso problem, for which many algorithms are available \citep{huang2012}. To see this, consider
\begin{align*}
%\arg\min_{\bbeta} \left\{\frac{1}{2n}\| \y - \Z\bbeta\|^2 + \frac{\mu}{2} \| \sum_{j=1}^{p}\A_j\bbeta_j -\b + \balpha^{\ell} \|^2\right\} + \lambda \sum_{j=1}^{p} \|\W_j\bbeta_j\|.
\arg\min_{\bbeta} \left\{\frac{1}{2n}\| \y - \Z\bbeta\|^2 + \frac{\mu}{2} \| \sum_{j=1}^{p}\A_j\bbeta_j -\b + \balpha^{\ell} \|^2 + \lambda \sum_{j=1}^{p} \|\W_j\bbeta_j\| \right\}.
\end{align*}
Here we have omitted the intercept term and the control variables as they can be treated as a group with zero penalty. Define $\A = (\A_1,\ldots, \A_p)$, %Then the problem is equivalent to
%\begin{align*}
%\arg\min_{\bbeta} \left\{\frac{1}{2n}\| \y - \Z\bbeta\|^2 + \frac{\mu^{\ell}}{2} \| \A\bbeta -\b + \balpha^{\ell} \|^2\right\} + \lambda %\sum_{j=1}^{p} \|\W_j\bbeta_j\|.
%\end{align*}
%Let $\widetilde{\bbeta}_j = \W_j\bbeta_j$, and $\widetilde{\bbeta} = \W\bbeta =\mbox{diag}(\W_1,\ldots,\W_p)\bbeta$. When $\W$ is invertable, the problem can be expressed in terms of $\widetilde{\bbeta}$ as
%\begin{align*} \label{eq:GMD}
%\arg\min_{\widetilde{\bbeta}} \left\{\frac{1}{2n}\| \y - \Z\W^{-1}\widetilde{\bbeta}\|^2 + \frac{\mu^{\ell}}{2} \| \A\W^{-1}\widetilde{\bbeta} -\b + \balpha^{\ell} \|^2\right\} + \lambda \sum_{j=1}^{p} \|\widetilde{\bbeta}_j\|,
%\arg\min_{\widetilde{\bbeta}} \left\{\frac{1}{2n}\widetilde{\bbeta}\trans\W^{-1}\trans(\Z\trans\Z + n\mu \A\trans\A)\W^{-1}\widetilde{\bbeta} - \frac{1}{n}(\y\trans\Z + n\mu(\b-\balpha)\trans \A)\W^{-1}\widetilde{\bbeta} + \lambda \sum_{j=1}^{p} \|\widetilde{\bbeta}_j\|.\right\},
%\end{align*}
% \begin{align*}
% \arg\min_{\bbeta}
% \left\{\frac{1}{2n}\bbeta\trans(\Z\trans\Z + n\mu^{\ell} \A\trans\A)\bbeta -
% \frac{1}{n}(\y\trans\Z + n\mu^{\ell}(\b-\balpha^{\ell})\trans \A)\bbeta +
% \lambda \sum_{j=1}^{p} \|\W_j\bbeta_j\|.\right\}.
% \end{align*}
$\widetilde{\bbeta}_j = \W_j\bbeta_j$, and $\widetilde{\bbeta} = \W\bbeta =\mbox{diag}(\W_1,\ldots,\W_p)\bbeta$. Then the objective can be expressed in terms of $\widetilde{\bbeta}$ as
\begin{align*} %\label{eq:GMD}
%\arg\min_{\widetilde{\bbeta}}
%\left\{
  & \frac{1}{2n}\widetilde{\bbeta}\trans(\W^{-1})\trans(\Z\trans\Z + n\mu \A\trans\A)\W^{-1}\widetilde{\bbeta} \\
& -
\frac{1}{n}(\y\trans\Z + n\mu(\b-\balpha^{\ell})\trans \A)\W^{-1}\widetilde{\bbeta} +
\lambda \sum_{j=1}^{p} \|\widetilde{\bbeta}_j\|.
%\right\},
\end{align*}
%The estimates of $\widetilde{\bbeta}$ can be easliy transformed back to obtain the estimates of $\bbeta$.
The dual step updates $\balpha$ as
$\balpha^{\ell+1}
 \leftarrow \balpha^\ell+ \sum_{j=1}^{p}\A_j\bbeta_j^{\ell+1} - \b$. To speed up computation, $\mu$ can be set to slowly increase along iterations \citep{Boyd2011}.

The optimization procedure for any fixed $\lambda$ is summarized in Algorithm \ref{alg:ALM}. When the model is fitted for a sequence of $\lambda$ values, a warm start strategy is adopted, i.e., the solution for the previous $\lambda$ value is used as the initial value for the next one.

% \begin{algorithm}
%   \caption{}
%   \begin{algorithmic}\label{alg:ALM}
%       \STATE  Initialize $\balpha^{0} \geq \0$, $\mu^{0} \geq0$. Choose $\rho>1$, e.g., $\rho = 1.05$. Choose convergence thresholds $\epsilon_1>0$ and $\epsilon_2>0$, e.g., $\epsilon_1 = \epsilon_2 = 10^{-4}$. Set $\ell \gets 0$.
%       \vspace{0.1cm}
%       \REPEAT
%       \STATE (1) Primal step: $\big(\bbeta_{0c}^{\ell+1},\bbeta^{\ell+1}\big)\leftarrow
% \min_{\bbeta_{0c},\bbeta} \{L (\bbeta_{0c},\bbeta;\balpha^{\ell},\mu^{\ell})\}$.
%       \STATE (2) Dual step:
%     %$\balpha^{\ell+1} \leftarrow \balpha^\ell+\sum_{j=1}^{p}\A_j\bbeta_j^{\ell+1} - \b$.
%     $\balpha^{\ell+1} \leftarrow (\balpha^\ell+\sum_{j=1}^{p}\A_j\bbeta_j^{\ell+1} - \b)/\rho$.
%  \STATE $\mu^{\ell+1} \gets \rho\mu^{\ell}$.
%       \vspace{0.1cm}
%       \STATE $\ell \gets \ell+1$.
%       \UNTIL{convergence, i.e., $\big(\|\bbeta_{0c}^{\ell+1} - \bbeta_{0c}^{\ell} \|^2  + \|\bbeta^{\ell+1} - \bbeta^{\ell}\|^2 \big)/\big( \|\bbeta_{0c}^{\ell}\|^2 + \|\bbeta^{\ell} \|^2\big)\leq \epsilon_1$ and $\|\sum_{j=1}^p \A_j\bbeta_j^{\ell+1} - \b\|^2 \le \epsilon_2$.}
%       \RETURN{$\widehat{\bbeta}_{0c} = \bbeta_{0c}^{\ell}$ and  $\widehat{\bbeta} = \bbeta^{\ell}$.}
%   \end{algorithmic}
% \end{algorithm}

\begin{algorithm}[h]
\caption{}\label{alg:ALM}
\begin{algorithmic}
\State Initialize $\balpha^{0} \geq \0$, $\mu^{0} \geq0$. Choose $\rho>1$, e.g., $\rho = 1.05$. Choose convergence thresholds $\epsilon_1>0$ and $\epsilon_2>0$, e.g., $\epsilon_1 = \epsilon_2 = 10^{-4}$. Set $\ell \gets 0$.
%\State Initialize $\bA$, $\bB$ and the Lagrange multiplier $\bLambda$;
\Repeat 
\bi
\item Primal step: $\big(\bbeta_{0c}^{\ell+1},\bbeta^{\ell+1}\big)\leftarrow
\min_{\bbeta_{0c},\bbeta} \{L (\bbeta_{0c},\bbeta;\balpha^{\ell},\mu^{\ell})\}$.
\item Dual step:
    %$\balpha^{\ell+1} \leftarrow \balpha^\ell+\sum_{j=1}^{p}\A_j\bbeta_j^{\ell+1} - \b$.
    $\balpha^{\ell+1} \leftarrow (\balpha^\ell+\sum_{j=1}^{p}\A_j\bbeta_j^{\ell+1} - \b)/\rho$.
\item $\mu^{\ell+1} \gets \rho\mu^{\ell}$.
\item  $\ell \gets \ell+1$.
\ei
\Until{convergence, i.e., $\big(\|\bbeta_{0c}^{\ell+1} - \bbeta_{0c}^{\ell} \|^2  + \|\bbeta^{\ell+1} - \bbeta^{\ell}\|^2 \big)/\big( \|\bbeta_{0c}^{\ell}\|^2 + \|\bbeta^{\ell} \|^2\big)\leq \epsilon_1$ and $\|\sum_{j=1}^p \A_j\bbeta_j^{\ell+1} - \b\|^2 \le \epsilon_2$.}\\
\Return{$\widehat{\bbeta}_{0c} = \bbeta_{0c}^{\ell}$ and  $\widehat{\bbeta} = \bbeta^{\ell}$.}
\end{algorithmic}
\end{algorithm}

\clearpage

\section{Proofs}\label{sec:supp:th}

We study the properties of the constrained group lasso
estimator,
\begin{align}\label{supp:eq:glasso2}
\widehat{\bbeta} = \arg\min_{\bbeta}\left\{\frac{1}{2n}\| \y - \Z\bbeta\|^2  + \lambda \sum_{j=1}^{p} \|\bbeta_j\|\right\}, \qquad \mbox{s.t.} \sum_{j=1}^{p}\bbeta_j=\0.
\end{align}
as define in \eqref{eq:glasso2} of the main paper. For the sake of completeness, we reproduce the theorems in the main paper. 

\begin{theorem}[Error Bounds]\label{supp:th:1}
Suppose Assumptions \ref{as:2}--\ref{as:3} presented in the main paper hold. Choose
$$
\lambda \geq \min_{r}\max_{j\neq r}\frac{2\sigma}{\sqrt{n}} \sqrt{\mbox{tr}(\bPsi_{\bar{r},j}) + 2\sigma_{\max}(\bPsi_{\bar{r},j})(2q\log(p-1)+\sqrt{kq\log(p-1)})}.
$$
Then, with probability at least $1-2(p-1)^{1-q}$, the constrained group lasso estimator $\widehat{\bbeta}$ in \eqref{supp:eq:glasso2} satisfies that
\begin{align}
\frac{1}{n}\|\Z(\widehat{\bbeta} - \bbeta^*)\|^2 \leq \frac{16\lambda^2s^*}{\kappa^2},\label{supp:eq:th11}\\
\sum_{j=1}^{p}\|\widehat{\bbeta}_j - \bbeta_j^*\| + \min_j \|\widehat{\bbeta}_j - \bbeta_j^*\| \leq \frac{16\lambda s^*}{\kappa^2}.\label{supp:eq:th12}
\end{align}
\end{theorem}

\begin{corollary}[Selection Consistency]\label{supp:th:2}
Suppose Assumptions \ref{as:2}--\ref{as:4} presented in the main paper hold. Let
$$
\widehat{\mathcal{S}} = \{j: \|\widehat{\bbeta}_j\|> \frac{8\lambda s^*}{\kappa^2}\}.
$$
Then, with probability at least $1-2(p-1)^{1-q}$, we have that $\widehat{\mathcal{S}} = \mathcal{S}$.
\end{corollary}

\begin{proof}[Proof of Theorem \ref{supp:th:1}]
For all $\bbeta = [\bbeta_1\trans,\ldots, \bbeta_p\trans]\trans \in \mathbb{R}^{pk}$, $\sum_{j=1}^{p}\bbeta_j=\0$, it holds that
\begin{align*}
\frac{1}{n}\|\y - \Z\widehat{\bbeta}\|^2 + 2\lambda \sum_{j=1}^{p}\|\widehat{\bbeta}_j\|
\leq  \frac{1}{n}\|\y - \Z\bbeta\|^2 + 2\lambda \sum_{j=1}^{p}\|\bbeta_j\|,
\end{align*}
by the optimality of the constrained group lasso estimator $\widehat{\bbeta}$. Using $\y = \Z\bbeta^*  + \e$, we have that
\begin{align}\label{eq:basicineq}
\frac{1}{n}\|\Z(\widehat{\bbeta} - \bbeta^*)\|^2
\leq \frac{1}{n}\|\Z(\bbeta - \bbeta^*)\|^2 + \frac{2}{n}\e\trans\Z(\widehat{\bbeta} - \bbeta)+ 2\lambda\sum_{j=1}^{p}(\|\bbeta_j\| - \|\widehat{\bbeta}_j\|).
\end{align}

We first bound the stochastic term $\e\trans\Z(\widehat{\bbeta} - \bbeta)$. Due to the zero-sum constrains, it is important to realize that for any $r = 1,\ldots, p$,
$$
\e\trans\Z(\widehat{\bbeta} - \bbeta) = \e\trans\widetilde{\Z}_{\bar{r}}(\widehat{\bbeta}_{\bar{r}} - \bbeta_{\bar{r}}).
$$
The following tail bound is from Lemma A.1 in \citet{Lounici2011}.

\begin{lemma}\label{lemma:1}
Let $\v = [v_1,\ldots,v_n]\trans \neq \0$, $\eta_{\v} = \sum_{i=1}^{n}(e_i^2-1)v_i/(\sqrt{2}\|\v\|)$ %$\eta_{\v} = \sum_{i=1}^{n}(e_i^2-1)v_i/\sqrt{2\|\v\|}$
, and $m(\v) = \|\v\|_\infty /\|\v\|$. Then, under Assumption \ref{as:2} in the main paper, for all $t>0$,
$$
\mathbb{P}(|\eta_{\v}| > t) \leq 2 \exp\left(-\frac{t^2}{2(1+\sqrt{2}tm(\v))}\right).
$$
\end{lemma}

For any fixed $r$, it can be shown using Lemma \ref{lemma:1} \citep{Lounici2011} that if we choose
$\lambda \geq \lambda_r$, where
$$
\lambda_r = \max_{j\neq r}\frac{2\sigma}{\sqrt{n}} \sqrt{\mbox{tr}(\bPsi_{\bar{r},j}) + 2\sigma_{\max}(\bPsi_{\bar{r},j})(2q\log(p-1)+\sqrt{kq\log(p-1)})},
$$
then with probability at least $1-2(p-1)^{1-q}$,
$$
\frac{2}{n}\e\trans\Z(\widehat{\bbeta} - \bbeta)\leq \lambda \sum_{j\neq r}^p\|\widehat{\bbeta}_j - \bbeta_j\|.
$$
Therefore, as long as we choose $\lambda \geq \min_r \lambda_r$, the preceding inequality holds for some $r$; it then follows that with probability at least $1-2(p-1)^{1-q}$, we have
$$
\frac{2}{n}\e\trans\Z(\widehat{\bbeta} - \bbeta)\leq \lambda \max_{r}\sum_{j\neq r}^p\|\widehat{\bbeta}_j - \bbeta_j\|.
$$
%{\color{red} (Another way is to choose $\lambda \geq \max_r \lambda_r$, then with an inflated probability, the right-hand-side becomes the minimum. Which way is better? )}

By \eqref{eq:basicineq}, we get
$$
\frac{1}{n}\|\Z(\widehat{\bbeta} - \bbeta^*)\|^2
\leq \frac{1}{n}\|\Z(\bbeta - \bbeta^*)\|^2 + \lambda \max_{r}\sum_{j\neq r}^p\|\widehat{\bbeta}_j - \bbeta_j\|
+ 2\lambda\sum_{j=1}^{p}(\|\bbeta_j\| - \|\widehat{\bbeta}_j\|).
$$
It then follows that
\begin{align*}
&\frac{1}{n}\|\Z(\widehat{\bbeta} - \bbeta^*)\|^2 + \lambda \sum_{j=1}^{p}\|\widehat{\bbeta}_j - \bbeta_j\| + \lambda \min_j \|\widehat{\bbeta}_j - \bbeta_j\|\notag\\
\leq & \frac{1}{n}\|\Z(\bbeta - \bbeta^*)\|^2
+ 2\lambda\sum_{j=1}^{p}(\|\bbeta_j\| - \|\widehat{\bbeta}_j\|+\|\widehat{\bbeta}_j - \bbeta_j\| ).%\label{eq:ineq2}
\end{align*}
Now take $\bbeta= \bbeta^*$, we get that
\begin{align}
&\frac{1}{n}\|\Z(\widehat{\bbeta} - \bbeta^*)\|^2 + \lambda \sum_{j=1}^{p}\|\widehat{\bbeta}_j - \bbeta_j^*\| + \lambda \min_j \|\widehat{\bbeta}_j - \bbeta_j^*\|\notag\\
& \leq 4\lambda\sum_{j\in \mathcal{S}}\min(\|\bbeta_j^*\|, \|\widehat{\bbeta}_j - \bbeta_j^*\|).\label{eq:ineq2}
\end{align}

The inequality in \eqref{eq:ineq2} implies that
$$
\lambda\sum_{j=1}^{p}\|\widehat{\bbeta}_j - \bbeta_j^*\| + \lambda \min_j \|\widehat{\bbeta}_j - \bbeta_j^*\|
\leq 4\lambda\sum_{j\in \mathcal{S}}\|\widehat{\bbeta}_j - \bbeta_j^*\|,
$$
which is equivalent to
$$
\sum_{j\in \mathcal{S}^c}\|\widehat{\bbeta}_j - \bbeta_j^*\| + \min_j \|\widehat{\bbeta}_j - \bbeta_j^*\|
\leq 3\sum_{j\in \mathcal{S}}\|\widehat{\bbeta}_j - \bbeta_j^*\|.
$$
Therefore, by the restricted eigenvalue condition in Assumption \ref{as:2} in the main paper, we know that
\begin{align}
\|\widehat{\bbeta}_{\mathcal{S}} - \bbeta_{\mathcal{S}}^*\| \leq \frac{\|\Z(\widehat{\bbeta} - \bbeta^*)\|}{\kappa\sqrt{n}}.\label{eq:ineq3}
\end{align}

It follows from \eqref{eq:ineq2}--\eqref{eq:ineq3} that
\begin{align*}
\frac{1}{n}\|\Z(\widehat{\bbeta} - \bbeta^*)\|^2
%& \leq 4\lambda \|\bbeta^*\|_{2,1}\\
& \leq 4\lambda\sum_{j\in \mathcal{S}}\|\widehat{\bbeta}_j - \bbeta_j^*\|\\
& \leq 4 \lambda \sqrt{s^*} \|\widehat{\bbeta}_{\mathcal{S}} - \bbeta_{\mathcal{S}}^*\|\\
& \leq 4 \lambda \sqrt{s^*}\frac{\|\Z(\widehat{\bbeta} - \bbeta^*)\|}{\kappa\sqrt{n}},
\end{align*}
which leads to \eqref{supp:eq:th11}. Also,
\begin{align*}
\sum_{j=1}^{p}\|\widehat{\bbeta}_j - \bbeta_j^*\| + \min_j \|\widehat{\bbeta}_j - \bbeta_j^*\|
& \leq 4\sum_{j\in \mathcal{S}}\|\widehat{\bbeta}_j - \bbeta_j^*\|\\
& \leq 4\sqrt{s^*} \|\widehat{\bbeta}_{\mathcal{S}} - \bbeta_{\mathcal{S}}^*\|\\
& \leq 4\sqrt{s^*}\frac{\|\Z(\widehat{\bbeta} - \bbeta^*)\|}{\kappa\sqrt{n}}\\
& \leq 4\sqrt{s^*}\sqrt{\frac{16\lambda^2s^*}{\kappa^2}}\frac{1}{\kappa}\\
& = \frac{16\lambda s^*}{\kappa^2},
\end{align*}
which leads to \eqref{supp:eq:th12}. This completes the proof.
\end{proof}

\begin{proof}[Proof of Corollary \ref{supp:th:2}]
Theorem \ref{supp:th:1} implies that
\begin{align}
\|\widehat{\bbeta} - \bbeta^*\|_{2,\infty} \leq \frac{8\lambda s^*}{\kappa^2} =a.\label{eq:ineq4}
\end{align}
If $\bbeta_j^*=0$, then $\|\widehat{\bbeta}_{j}\|\leq a$; %$\|\widehat{\bbeta}\|\leq a$
 so that $j\notin \widehat{\mathcal{S}}$. Now consider $\bbeta_j^*\neq 0$. By the $\beta$-min condition, i.e., $\|\bbeta_{j}^*\|>2a$ %$\|\bbeta^*\|>2a$
, together with \eqref{eq:ineq4}, it must be true that $\|\widehat{\bbeta}_j\|>a$ %$\widehat{\bbeta}_j \neq 0$,
, so that $j \in \widehat{\mathcal{S}}$. This completes the proof.
\end{proof}

\clearpage
\section{Additional Simulation Results}\label{sec:supp:sim}

We present additional simulation results for various models with the signal to noise ratio (SNR) is set to 2.

\vspace{1cm}

%\FloatBarrier
\begin{figure}[h!]
    \captionsetup[subfigure]{singlelinecheck=true}
    \centering
    \subcaptionbox[singlelinecheck=true]{$n = 50, p = 30$}{
        \includegraphics[ width = 2.2in, angle = -90]{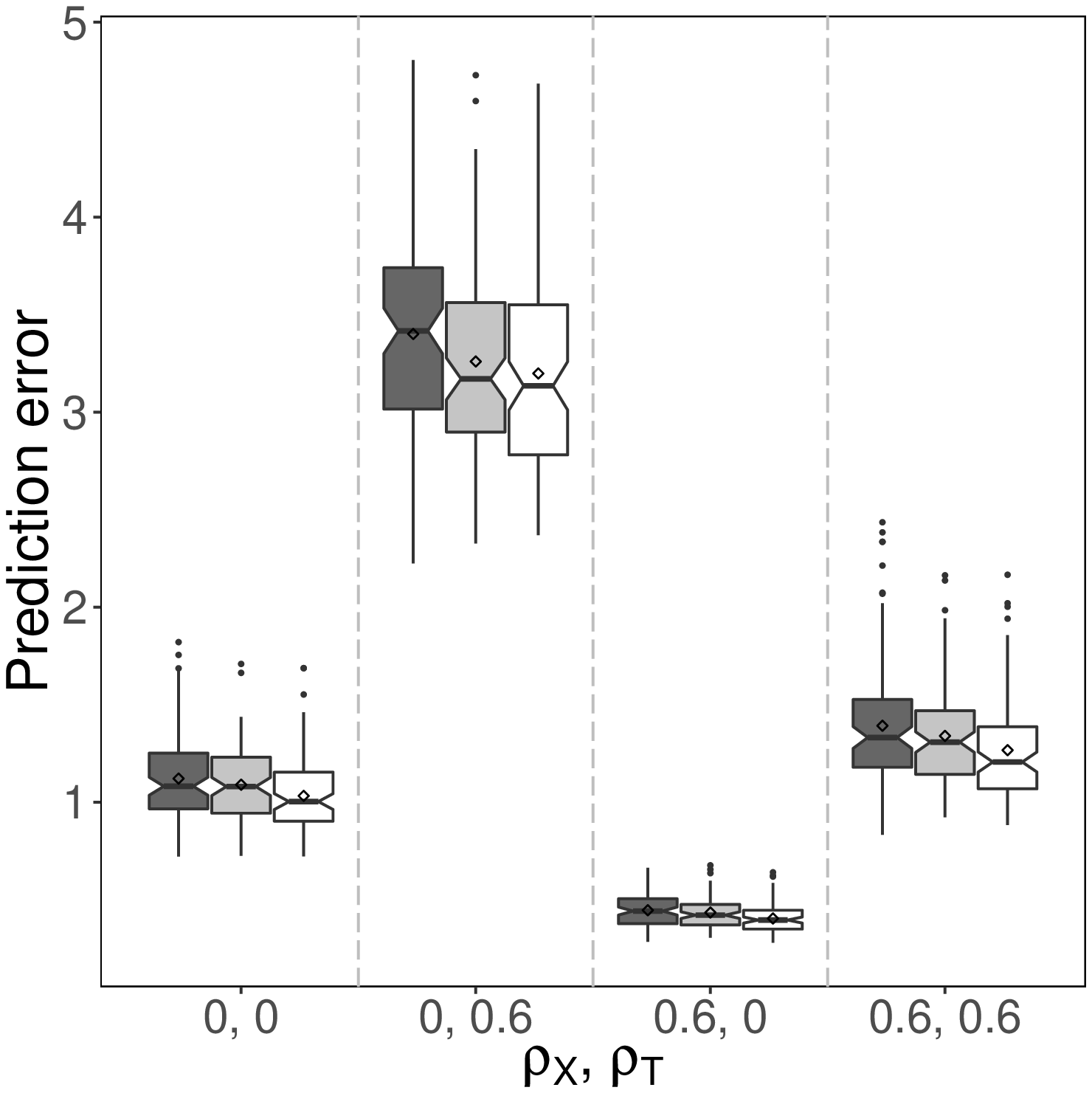}
        %\makebox[0pt][l]{%
        %   \hspace{0.04\columnwidth}
    }\hspace{0.5cm}%
    \subcaptionbox{$n = 100, p = 30$}{
        \includegraphics[ width = 2.2in, angle = -90]{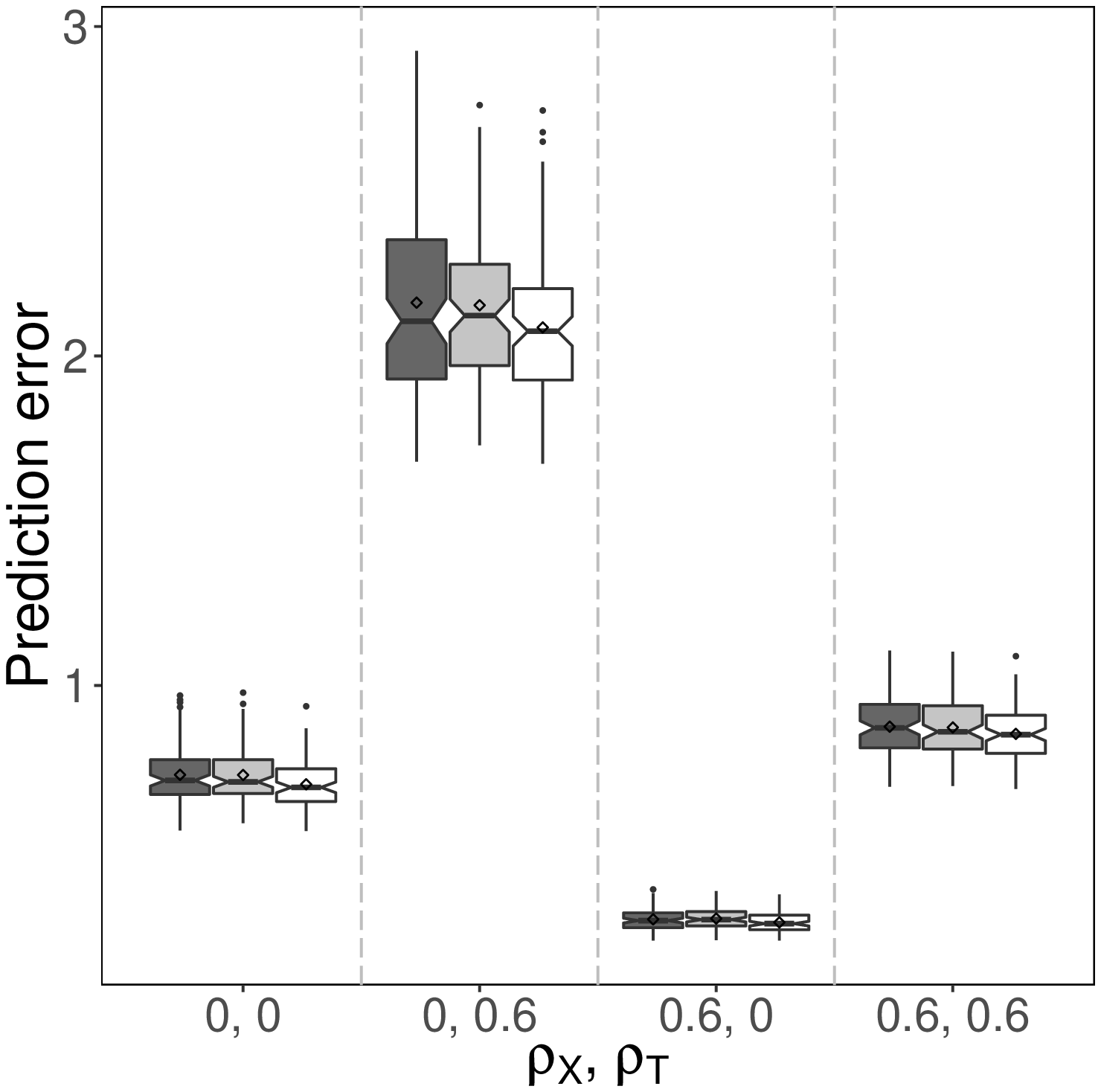}
    }\hspace{0.5cm}%
    \subcaptionbox{$n = 100, p = 100$}{
        \includegraphics[ width = 2.2in, angle = -90]{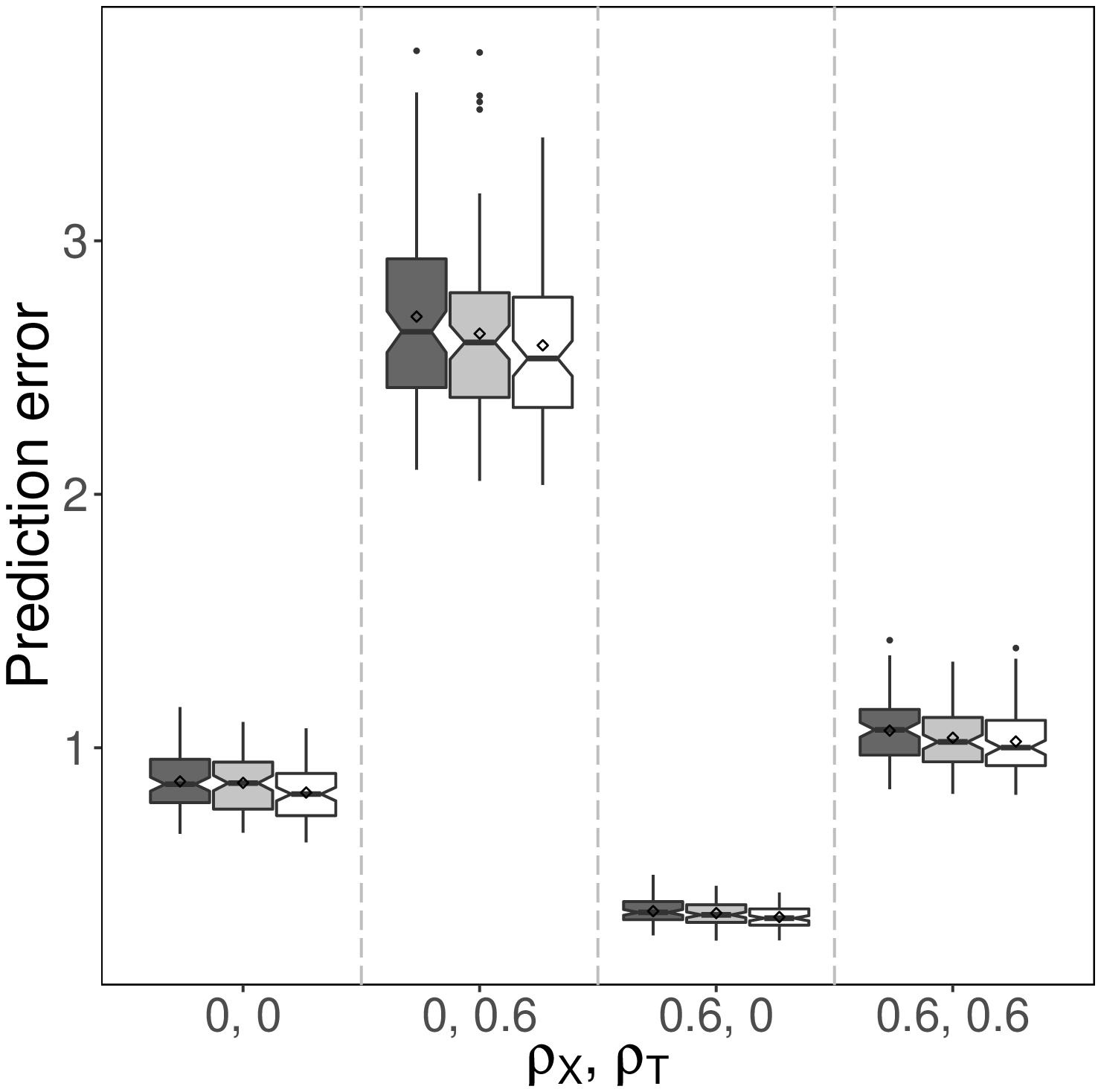}
    }\hspace{0.5cm}%
    \subcaptionbox{$n = 100, p = 200$}{
        \includegraphics[ width = 2.2in, angle = -90]{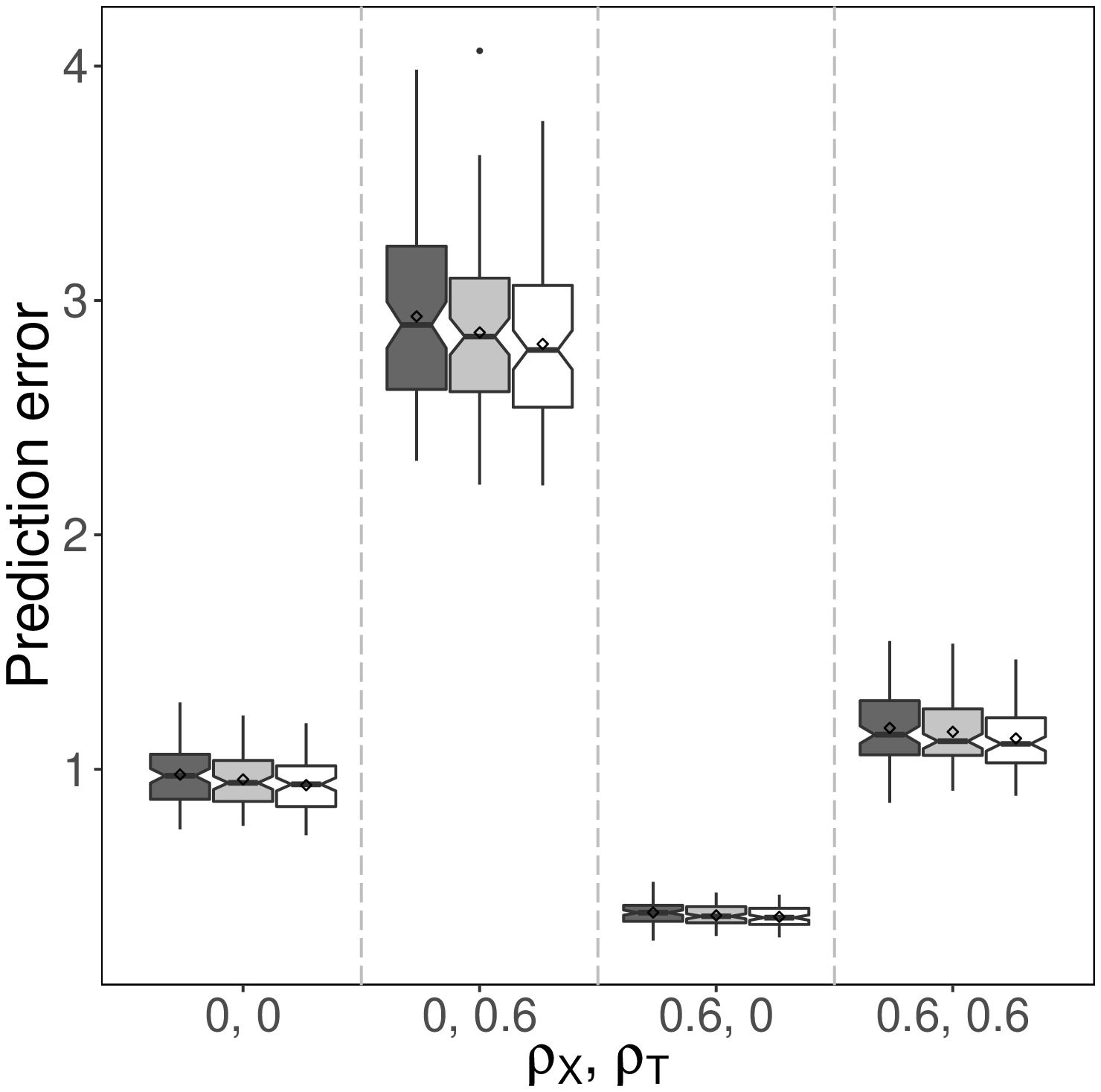}
    }
    \caption{Boxplots of prediction errors for various simulation settings with $\mbox{SNR} = 2$. The layout is the same as in Figure \ref{fig:sim1} of the main paper.
    }\label{fig:sim2}
\end{figure}

% <<< Appendix >>>

%\input{tables}
%\input{figures}

\end{document}